\documentclass{article}
\usepackage[utf8]{inputenc}

\usepackage[a4paper,bindingoffset=0.2in,
           top=1.1in, bottom=1.1in, left=1.1in, right=1.1in,
           footskip=.5in]{geometry}

\usepackage{authblk}
\usepackage{amsmath}
\usepackage{amssymb}
\usepackage{amsfonts,amsthm}
\usepackage{thmtools,thm-restate}
\usepackage{tabularx,lipsum,environ}
\usepackage{multirow}
\usepackage{hyperref}
\usepackage{graphicx,float}
\usepackage{enumitem,linegoal}
\usepackage{caption}
\usepackage{subcaption}
\usepackage{enumerate}
\usepackage{complexity}
\usepackage{cleveref}
\usepackage[numbers]{natbib}

\setlist[itemize]{noitemsep}
\setlist{nolistsep}

\newtheorem{theorem}{Theorem}[section]
\newtheorem{lemma}[theorem]{Lemma}
\newtheorem{corollary}[theorem]{Corollary}
\newtheorem{claim}{Claim}
\newtheorem{observation}{Observation}

\newenvironment{claimproof}{\begin{proof}\renewcommand{\qedsymbol}{\claimqed}}{\end{proof}\renewcommand{\qedsymbol}{\plainqed}}
\let\plainqed\qedsymbol

\newcommand{\yes}{\textsc{Yes}}
\newcommand{\no}{\textsc{No}}

\usepackage{tikz}
\usetikzlibrary{arrows}
\usetikzlibrary{arrows.meta}
\usetikzlibrary{decorations.pathmorphing}
\usetikzlibrary{decorations.markings}
\usetikzlibrary{calc}
\usetikzlibrary{positioning}
\tikzset{
  circ/.style = {circle,draw,fill,inner sep=1.2pt},
  cir/.style = {circle,draw,fill,inner sep=1.2pt},
  circR/.style = {circle,draw=red,fill=red,text=red,inner sep=1.3pt},
  circb/.style = {circle,draw=blue,fill=blue,text=blue,inner sep=1.1pt},
  circr/.style = {circle,draw=red,fill=red,inner sep=1pt},
  scirc/.style = {circle,draw,fill,inner sep=.8pt},
  invisible/.style = {draw=none,inner sep=0pt,font=\tiny},
  nonedge/.style={decorate,decoration={snake,amplitude=.3mm,segment length=1mm},draw}
}

\newcommand{\longPaw}[5]{
\node[circ,label=below:{\tiny #4}] (a) at (#1+0,#2) {};
\node[circ,label=below:{\tiny #5}] (b) at (#1+1,#2) {};
\node[circ] (c) at (#1+.5,#2+.3) {};
\node[circ,label=right:{\tiny #3}] at (#1+.5,#2+.6) {};
\node[circ] (d) at (#1+.5,#2+.9) {};                                                                                                                                                                                                                                                                                                                                                                                                                                                            

\draw (a) -- (b)
(a) -- (c) 
(b) -- (c) 
(c) -- (d);
}

\newcommand{\shortPaw}[6]{
\node[circ,label=below:{\tiny #3}] (a) at (#1,#2) {};
\node[circ,label=below:{\tiny #4}] (b) at (#1+1,#2) {};
\node[circ,label=right:{\tiny #5}] (c) at (#1+.5,#2+.4) {};
\node[circ,label=right:{\tiny #6}] (d) at (#1+.5,#2+.8) {};

\draw (a) -- (b)
(a) -- (c)
(b) -- (c)
(c) -- (d);
}

\begin{document}

\title{The complexity of blocking (semi)total dominating sets with edge contractions}
\date{}
\author[1]{Esther Galby}
\affil[1]{CISPA Helmholtz Center for Information Security, Saarbr\"ucken, Germany}

\maketitle

\begin{abstract}
We consider the problem of reducing the (semi)total domination number of graph by one by contracting edges. It is known  that this can always be done with at most three edge contractions and that deciding whether one edge contraction suffices is an $\mathsf{NP}$-hard problem. We show that for every fixed $k \in \{2,3\}$, deciding whether exactly $k$ edge contractions are necessary is $\mathsf{NP}$-hard and further provide for $k=2$ complete complexity dichotomies on monogenic graph classes.
\end{abstract}

%------------------------------------------------------------------------------------------------------------------------------------------------------------------------------------

\section{Introduction}

A \emph{blocker problem} asks whether given a graph $G$, a graph a parameter $\pi$, a set $\mathcal{O}$ of one or more graph operations and an integer $k \geq 1$, $G$ can be transformed into a graph $H$ such that $\pi(H) \leq \pi(G) - d$ for some \emph{threshold} $d \geq 1$, by using at most $k$ operations from $\mathcal{O}$. These problems are so called because the set of vertices or edges involved can be seen as "blocking" the parameter $\pi$. Identifying such sets may provide important information on the structure of the input graph: for instance, if $\pi = \alpha$, $k = d = 1$ and $\mathcal{O} = \{\text{vertex deletion}\}$, then the problem is equivalent to testing whether the input graph contains a vertex which belongs to every maximum independent set \cite{paulusma2017blocking}. While the set $\mathcal{O}$ generally consists of a single operation (namely vertex deletion, edge deletion, edge addition or edge contraction), a variety of parameters have been considered in the literature including the chromatic number \cite{bazgan2015blockers,DPPR15,diner2018contraction,PPR16,paulusma2018critical}, the stability number \cite{BTT11,paulusma2017blocking}, the clique number \cite{nasirian2019exact,PBP}, the matching number \cite{RBPDCZ10,ZENKLUSEN2}, domination-like parameters \cite{isaac2019,domcontract,semitotcon,GALBY202118,HX10,pajouh2015minimum} and others \cite{CHEN,keller2018blockers,MFCS2020,rautenbach,ZENKLUSEN1}. In this paper, we focus on two well-known variants of the domination number, namely the \emph{total domination number} and the \emph{semitotal domination number}, let $\mathcal{O}$ consists of an edge contraction and set the threshold $d$ to one. 

Formally, let $G$ be a graph. The \emph{contraction} of an edge $xy \in E(G)$ removes vertices $x$ and $y$ from $G$ and replaces them with a new vertex that is made adjacent to precisely those vertices which were adjacent to $x$ or $y$ (without introducing self-loops nor multiple edges). For a parameter $\pi$, we denote by $ct_{\pi}(G)$ the minimum number of edge contractions required to modify $G$ into a graph $H$ such that $\pi (H) = \pi (G) -1$. A set $D \subseteq V(G)$ is a \emph{total dominating set} of $G$ if every vertex in $V(G)$ has a neighbour in $D$, and the \emph{total domination number $\gamma_t(G)$} of $G$ is the size of a minimum total dominating set of $G$. A set $D \subseteq V(G)$ is a \emph{semitotal dominating set} of $G$ if every vertex in $V(G) \setminus D$ has a neighbour in $D$ and every vertex in $D$ is at distance at most two from another vertex in $D$. The \emph{semitotal domination number $\gamma_{t2}(G)$} of $G$ is the size of a minimum semitotal dominating set of $G$. We are interested in the following problem for $\pi \in \{\gamma_t,\gamma_{t2}\}$.

\begin{center}
\fbox{
\begin{minipage}{5.5in}
\textsc{Contraction Number($\pi,k$)}
\begin{description}
\item[Instance:] A graph $G$.
\item[Question:] Is $ct_{\pi}(G) = k$?
\end{description}
\end{minipage}}
\end{center}

It is known \cite{semitotcon,HX10} that, contrary to other parameters such as the chromatic number, the stability number or the clique number\footnote{To see that the minimum number of edge contractions required to decrease the value of the parameter may be arbitrarily large for (1) the chromatic number, consider e.g. stars, (2) the stability number, consider e.g., graphs obtained by identifying one vertex in two otherwise disjoint cliques, (3) the clique number, consider e.g. paths.}, $ct_{\pi}$ is bounded for $\pi \in \{\gamma_t,\gamma_{t2}\}$, namely by three in both cases. It follows that for $\pi \in \{\gamma_t,\gamma_{t2}\}$ and $k >3$, every instance of the above problem is always negative. Similarly, for $\pi \in \{\gamma_t,\gamma_{t2}\}$ and $k \geq3$, deciding $ct_{\pi}(G) \leq k$ (the so-called \textsc{$k$-Edge Contraction($\pi$)} problem \cite{domcontract}) is trivial. In contrast, it was shown \cite{semitotcon,GALBY202118} that for $\pi \in \{\gamma_t,\gamma_{t2}\}$, \textsc{$1$-Edge Contraction($\pi$)} (or, equivalently, \textsc{Contraction Number($\pi,1$)}) is $\mathsf{NP}$-hard; the complexity status for $k =2,3$ remained open. In this paper, we settle these questions and show that for $\pi \in \{\gamma_t,\gamma_{t2}\}$ and $k =2,3$, \textsc{Contraction Number($\pi,k$)} is $\mathsf{NP}$-hard. Thus, combined with \cite{semitotcon,GALBY202118}, the following dichotomy holds.

\begin{theorem}
For $\pi \in \{\gamma_t,\gamma_{t2}\}$, \textsc{Contraction Number($\pi,k$)} is $\mathsf{NP}$-hard if and only if $k\leq 3$.
\end{theorem} 

Let us note that since for $\pi \in \{\gamma_t,\gamma_{t2}\}$, a graph $G$ is a \no-instance for \textsc{Contraction Number($\pi,\allowbreak 3$)} if and only if $G$ is a \yes-instance for \textsc{$2$-Edge Contraction($\pi$)}, the above implies the following.  

\begin{theorem}
For $\pi \in \{\gamma_t,\gamma_{t2}\}$, \textsc{$k$-Edge Contraction($\pi$)} is $\mathsf{(co)NP}$-hard if and only if $k \leq 2$.
\end{theorem}

It is, however, not difficult to find cases which are easy to solve: for instance, in any graph class closed under edge contraction and where $\pi$ can be efficiently computed, the {\sc Contraction Number($\pi$,$k$)} problem can efficiently be solved through a simple brute force approach. Based on this observation, a natural question is whether there are other efficiently solvable instances for which computing $\pi$ is in fact hard. Motivated by such questions, we consider these problems on monogenic graph classes (that is, graph classes excluding one graph as an induced subgraph) for which the complexity status of {\sc Total Domination} and {\sc Semitotal Domination} \cite{semitot} is known: both problems are polynomial-time solvable on $H$-free graph if $H \subseteq_i P_4+tK_1$ for some $t \geq 0$, and $\mathsf{NP}$-hard otherwise. This investigation led us to establish complete complexity dichotomies for $\pi \in \{\gamma_t,\gamma_{t2}\}$ and $k =2$ on monogenic graph classes, which read as follows.  

\begin{theorem}
\label{thm:dictd2}
{\sc Contraction Number($\gamma_t$,2)} is polynomial-time solvable on $H$-free graphs if $H \subseteq_i P_5+tK_1$ for some $t \geq 0$ or $H \subseteq_i P_4+tP_3$ for some $t \geq 0$, and $\mathsf{(co)NP}$-hard otherwise.
\end{theorem}

\begin{theorem}
\label{thm:dicstd2}
{\sc Contraction Number($\gamma_{t2}$,2)} is polynomial-time solvable on $H$-free graphs if $H \subseteq_i P_5+tK_1$ for some $t \geq 0$ or $H \subseteq_i P_3+tP_2$ for some $t \geq 0$, and $\mathsf{(co)NP}$-hard otherwise.
\end{theorem}

\noindent
\textbf{Related work.} The study of the blocker problem for $\pi = \gamma_t$ and $\mathcal{O} = \{\text{edge contraction}\}$ was initiated by Huang and Xu \cite{HX10} who characterised for every $k \in[3]$, the graphs for which $ct_{\gamma_t}(G) =k$ in terms of the structure of their total dominating sets. More specifically, they proved the following theorem (see \Cref{sec:prelim} for missing definitions).

\begin{theorem}[\cite{HX10}]
\label{theorem:1totcontracdom}
For any graph $G$, the following holds.
\begin{itemize}
\item[(i)] $ct_{\gamma_t} (G)=1$ if and only if there exists a minimum total dominating set of $G$ that contains a (not necessarily induced) $P_3$.
\item[(ii)]  $ct_{\gamma_t} (G)=2$ if and only if no minimum total dominating set of $G$ contains a $P_3$ and there exists a total dominating set of $G$ of size $\gamma_t(G)+1$ that contains a (not necessarily induced) $P_4$, $K_{1,3}$ or $2P_3$.
\end{itemize}
\end{theorem}

On the other hand, the algorithmic study of {\sc Contraction Number($\gamma_t$,1)} was initiated in \cite{GALBY202118} where the following dichotomy theorem was established.

\begin{theorem}[\cite{GALBY202118}]
\label{thm:dictd1}
{\sc Contraction Number($\gamma_t$,1)} is polynomial-time solvable on $H$-free graphs if $H \subseteq_i P_5+tK_1$ for some $t \geq 0$ or $H \subseteq_i P_4+tP_3$ for some $t \geq 0$, and $\mathsf{(co)NP}$-hard otherwise.
\end{theorem}

Similarly to the total domination parameter, Galby et al. \cite{semitotcon} later characterised for every $k \in [3]$ the graphs for which $ct_{\gamma_{t2}}(G) = k$ in terms of the structure of their semitotal dominating sets, as follows (see \Cref{sec:prelim} for missing definitions).

\begin{theorem}[\cite{semitotcon}]
\label{thm:friendlytriple}
For any graph $G$, the following holds.
\begin{itemize}
\item[(i)] $ct_{\gamma_{t2}}(G) = 1$ if and only if there exists a minimum semitotal dominating set of $G$ that contains a friendly triple.
\item[(ii)] $ct_{\gamma_{t2}}(G) = 2$ if and only if no minimum semitotal dominating set of $G$ contains a friendly triple and there exists a semitotal dominating set of $G$ of size $ \gamma_{t2}(G) + 1$ that contains a (not necessarily induced) ST-configuration.
\end{itemize}
\end{theorem}

\begin{figure}
\begin{minipage}[b]{0.23\textwidth}
\begin{subfigure}[b]{\linewidth}
\centering
\begin{tikzpicture}[node distance=0.55cm]
\node[scirc] (1) at (0,0) {};
\node[scirc,right of=1] (2) {};
\node[scirc,right of=2] (3) {};
\node[scirc,above of=1] (4) {};
\node[scirc,right of=4] (5) {};
\node[scirc,right of=5] (6) {};

\draw[-] (2) -- (3)
(5) -- (6)
(1) -- (2)
(4) -- (5);
\end{tikzpicture}
\caption*{$O_1$}
\end{subfigure}

\vspace*{1.5mm}

\begin{subfigure}[b]{\linewidth}
\centering
\begin{tikzpicture}[node distance=0.55cm]
\node[scirc] (1) at (0,0) {};
\node[scirc,right of=1] (2) {};
\node[scirc,right of=2] (3) {};
\node[scirc,above of=1] (4) {};
\node[scirc,right of=4] (5) {};
\node[scirc,right of=5] (6) {};

\draw[-] (5) -- (6);
\draw[dashed] (2) -- (3);
\draw (1) -- (2)
(4) -- (5) ;
\end{tikzpicture}
\caption*{$O_2$}
\end{subfigure}
\end{minipage}
\begin{minipage}[b]{0.23\textwidth}
\begin{subfigure}[b]{\linewidth}
\centering
\begin{tikzpicture}[node distance=0.55cm]
\node[scirc] (1) at (0,0) {};
\node[scirc,right of=1] (2) {};
\node[scirc,right of=2] (3) {};
\node[scirc,above of=1] (4) {};
\node[scirc,right of=4] (5) {};
\node[scirc,right of=5] (6) {};

\draw (1) -- (2)
(4) -- (5);
\draw[dashed] (2) -- (3);
\draw[dashed] (5) -- (6);
\draw (3) edge[nonedge] (6);
\end{tikzpicture}
\caption*{$O_3$}
\end{subfigure}

\vspace*{1.5mm}

\begin{subfigure}[b]{\linewidth}
\centering
\begin{tikzpicture}[node distance=0.55cm]
\node[scirc] (1) at (0,0) {};
\node[scirc,right of=1] (2) {};
\node[scirc,right of=2] (3) {};
\node[scirc,right of=3] (4) {};
\node[invisible] at (1,-.3) {};
\node[invisible] at (1,.33) {};

\draw (1) -- (2)
(2) -- (3);
\draw[dashed] (3) -- (4);
\end{tikzpicture}
\caption*{$O_4$}
\end{subfigure}
\end{minipage}
\begin{minipage}[b]{0.23\textwidth}
\begin{subfigure}[b]{\linewidth}
\centering
\begin{tikzpicture}[node distance=0.48cm]
\node[scirc] (1) at (0,0) {};
\node[scirc,above of=1] (2) {};
\node[scirc,below right of =1] (3) {};
\node[scirc,below left of=1] (4) {};

\draw (1) -- (2)
(1) -- (3);
\draw[-](1) -- (4);
\end{tikzpicture}
\caption*{$O_5$}
\end{subfigure}

\vspace*{1.5mm}

\begin{subfigure}[b]{\linewidth}

\end{subfigure}
\end{minipage}
\begin{minipage}[b]{0.23\textwidth}
\begin{subfigure}[b]{\linewidth}
\centering
\begin{tikzpicture}[node distance=0.55cm]
\node[scirc] (1) at (0,0) {};
\node[scirc,right of=1] (2) {};
\node[scirc,right of=2] (3) {};
\node[scirc,above of=2] (4) {};

\draw (1) -- (2)
(2) -- (4);
\draw[dashed] (2) -- (3);
\end{tikzpicture}
\caption*{$O_6$}
\end{subfigure}

\vspace*{1.5mm}

\begin{subfigure}[b]{\linewidth}
\centering
\begin{tikzpicture}[node distance=0.55cm]
\node[scirc] (1) at (0,0) {};
\node[scirc,right of=1] (2) {};
\node[scirc,right of=2] (3) {};
\node[scirc,right of=3] (4) {};
\node[invisible] at (1,-.3) {};
\node[invisible] at (1,.33) {};

\draw (1) -- (2)
(3) -- (4);
\draw[dashed] (2) -- (3);
\end{tikzpicture}
\caption*{$O_7$}
\end{subfigure}
\end{minipage}
\caption{The ST-configurations (the dashed lines indicate that the corresponding vertices are at distance two and the serpentine line indicates that the corresponding vertices can be identified).} 
\label{fig:conf}
\end{figure}
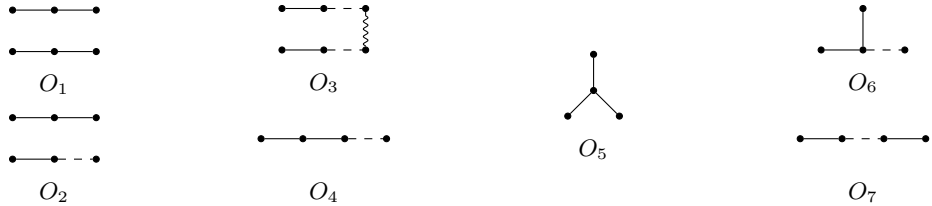

In the same paper \cite{semitotcon}, the authors further established the following dichotomy theorem on monogenic graph classes.

\begin{theorem}
\label{thm:dicstd1}
{\sc Contraction Number($\gamma_{t2}$,1)} is polynomial-time solvable for $H$-free graphs if $H \subseteq_i P_5+tK_1$ for some $t \geq 0$ or $H \subseteq_i P_3+tP_2$ for some $t\geq 0$, and $\mathsf{(co)NP}$-hard otherwise.
\end{theorem}

The paper is organised as follows. In \Cref{sec:prelim}, we present the necessary definitions and notations, as well as some preliminary results. In \Cref{sec:TD}, we study the complexity of \textsc{Contraction Number($\gamma_t,k$)} for $k = 2,3$, and in \Cref{sec:STD}, that of \textsc{Contraction Number($\gamma_{t2},k$)} for $k = 2,3$. The proof of \Cref{thm:dictd2} can be found in \Cref{sec:dictd2} and that of \Cref{thm:dicstd2}, in \Cref{sec:dicstd2}.

%------------------------------------------------------------------------------------------------------------------------------------------------------------------------------------

\section{Preliminaries}
\label{sec:prelim}

\subsection{Definitions and notations}

Throughout this paper, we only consider finite and simple graphs. Furthermore, since for any parameter $\pi$ and any $k \geq 0$, a non-connected graph $G$ is a \yes-instance for {\sc Contraction Number($\pi$,$k$)} if and only if at least one connected component of $G$ is a \yes-instance for {\sc Contraction Number($\pi$,$k$)}, we restrict our study to connected graphs.

For a graph $G$, we denote its vertex set by $V(G)$ and its edge set by $E(G)$. For any edge $e=\{x,y\} \in E(G)$, we call $x$ and $y$ the \emph{endvertices} of $e$. We may also denote the edge $e$ by $xy$ instead of $\{x,y\}$. Given a set $S \subseteq V(G)$, we let $G[S]$ denote the graph \emph{induced} by $S$, that is, the graph with vertex set $S$ and edge set $\{xy \in E(G)~|~x,y \in S\}$. A graph $H$ is an \emph{induced subgraph} of $G$, which we denote by $H \subseteq_i G$, if there exists a set $S \subseteq V(G)$ such that $H$ is isomorphic to $G[S]$. The \emph{neighbourhood} $N(x)$ of a vertex $x \in V(G)$ is the set $\{y \in V(G)~|~xy \in E(G)\}$ and the \emph{closed neighbourhood} $N[x]$ of $x$ is the set $N(x) \cup \{x\}$. Similarly, for a set $X \subseteq V(G)$, the \emph{neighbourhood} $N(X)$ of $X$ is the set $\{y \in V(G)~|~\exists x \in X \text{ s.t. } xy \in E(G)\} \setminus X$ and the \emph{closed neighbourhood} $N[X]$ of $X$ is the set $N(X) \cup X$. The \emph{degree} $d(x)$ of a vertex $x \in V(G)$ is the size of its neighbourhood. For any two vertices $x,y \in V(G)$, the \emph{distance between $x$ and $y$} is the number of edges in a shortest path from $x$ to $y$ and is denoted $d(x,y)$. Similarly, for any two sets $X,Y\subseteq V(G)$, the distance $d(X,Y)$ from $X$ to $Y$ is the number of edges in a shortest path from $X$ to $Y$, that is, $d(X,Y) = \min_{x \in X,y \in Y} d(x,y)$. If $X = \{x\}$, we may write $d(x,Y)$ instead of $d(\{x\},Y)$. For any two sets $X,Y \subseteq V(G)$, $X$ is \emph{complete} (\emph{anticomplete}, respectively) to $Y$, which we denote by $X - Y$ ($X \cdots Y$, respectively) if every vertex in $X$ is adjacent (nonadjacent, respectively) to every vertex in $Y$.

For $n \geq 1$, we denote by $P_n$ the path on $n$ vertices and by $K_{1,n}$ the complete bipartite graph with partitions of size one and $n$. If $H$ is a graph and $k \in \mathbb{N}$, we denote by $kH$ the graph consisting of $k$ disjoint copies of $H$.  The \emph{paw} is the graph depicted in \Cref{fig:paw} and unless specified otherwise, the vertices of a paw $P$ will be denoted by $P(1),\ldots,P(4)$ (as indicated in \Cref{fig:paw}) throughout the paper. The \emph{long paw} is the graph depicted in \Cref{fig:longpaw} and unless specified otherwise, the vertices of a long paw $P$ will be denoted by $P(1),\ldots,P(5)$ (as indicated in \Cref{fig:longpaw}) throughout the paper.

\begin{figure}
\centering
\begin{subfigure}[b]{.45\textwidth}
\centering
\begin{tikzpicture}
\node[circ,label=left:{\small $P(1)$}] (x) at (0,0) {};
\node[circ,label=right:{\small $P(2)$}] (xbar) at (1,0) {};
\node[circ,label=right:{\small $P(3)$}] (u) at (.5,.7) {};
\node[circ,label=right:{\small $P(4)$}] (w) at (.5,1.3) {};

\draw (x) -- (xbar)
(x) -- (u)
(xbar) -- (u) -- (w);
\end{tikzpicture}
\caption{A paw $P$.}
\label{fig:paw}
\end{subfigure}
\vspace*{.5cm}
\begin{subfigure}[b]{.45\textwidth}
\centering
\begin{tikzpicture}
\node[circ,label=left:{\small $P(1)$}] (x) at (0,0) {};
\node[circ,label=right:{\small $P(2)$}] (xbar) at (1,0) {};
\node[circ,label=right:{\small $P(3)$}] (u) at (.5,.7) {};
\node[circ,label=right:{\small $P(4)$}] at (.5,1.3) {};
\node[circ,label=right:{\small $P(5)$}] (w) at (.5,1.8) {};

\draw (x) -- (xbar)
(x) -- (u)
(xbar) -- (u)
(u) -- (w);
\end{tikzpicture}
\caption{A long paw $P$.}
\label{fig:longpaw}
\end{subfigure}
\caption{Some special graphs.}
\end{figure}
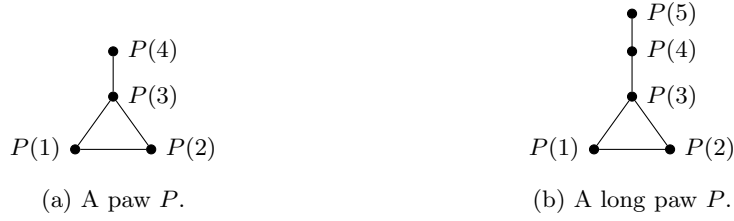 

Let $G$ be a graph with at least two vertices. A set $D \subseteq V(G)$ is \emph{dominating} if every vertex in $V(G) \setminus D$ has a neighbour in $D$. A set $D \subseteq V(G)$ is a \emph{total dominating set} (\emph{TD set} for short) of $G$ if every vertex in $V(G)$ has a neighbour in $D$. The \emph{total domination number $\gamma_t(G)$} of $G$ is the size of a minimum total dominating set of $G$. Note that $\gamma_t(G) \geq 2$ by definition. For every $x \in V(G)$ and every neighbour $y \in N(x) \cap D$, we say that \emph{$y$ dominates $x$} or that \emph{$x$ is dominated by $y$}. A set $D \subseteq V(G)$ is a \emph{semitotal dominating set} (\emph{SD set} for short) of $G$ if $D$ is dominating and every vertex in $D$ is at distance at most two from another vertex in $D$. The \emph{semitotal domination number $\gamma_{t2}(G)$} is the size of a minimum semitotal dominating set of $G$. Note that $\gamma_{t2}(G) \geq 2$ by definition. For every $x \in V(G)$ and every $y \in N[x] \cap D$, we say that \emph{$y$ dominates $x$} or that \emph{$x$ is dominated by $y$} (note that contrary to a total dominating set, $x$ can dominate itself). If $x,y \in D$ and $d(x,y) \leq 2$ then $y$ is called a \emph{witness for $x$}. The set of witnesses of a vertex $x \in D$ is denoted by $w_D(x)$. Note that with this terminology, a semitotal dominating set $D$ is equivalently a dominating set where every vertex in $D$ has a witness. For every $x \in D$, we denote by $PN_D(x) = \{y\in V(G)~|~N(y) \cap D = \{x\}\}$ the \emph{private neighbourhood} of $x$ with respect to $D$. A \emph{friendly triple} is a subset of three vertices $x,y$ and $z$ such that $xy \in E(G)$ and $d_G(y,z) \leq 2$.

Finally, let us define the {\sc Positive Exactly 3-Bounded 1-In-3 3-Sat} problem which will be used in several of our hardness reductions: it is an $\mathsf{NP}$-hard \cite{1IN3} variant of the {\sc 3-Sat} problem where given a formula $\Phi$ in which all literals are positive, every clause contains exactly three literals and every variable appears in exactly three clauses, the problem is to determine whether there exists a truth assignment for $\Phi$ such that each clause contains exactly one true literal.

\subsection{Preliminary results}

In this section, we present some useful technical results.

\begin{lemma}
\label{clm:condct2}
Let $G$ be a graph such that $\gamma_t(G) \geq 3$. Then for any minimum TD set $D$ of $G$, if there exist $x,y,z \in D$ such that $xy \in E(G)$ and $d(z,\{x,y\}) \leq 2$, then $ct_{\gamma_t}(G) \leq 2$.
\end{lemma} 

\begin{proof}
Suppose that there exist three such vertices $x,y,z \in D$. If $d(z,\{x,y\}) =1$, say $z$ is adjacent to $y$ without loss of generality, then $D$ contains the $P_3$ $xyz$ and thus, $ct_{\gamma_t}(G) =1$ by \Cref{theorem:1totcontracdom}. Suppose next that $d(z,\{x,y\}) =2$, say $d(z,\{x,y\}) = d(z,y)$ without loss of generality, and let $t$ be a common neighbour of $y$ and $z$. If $t$ belongs to $D$ as well, then $D$ contains the $P_3$ $ytz$ and we conclude as previously. Otherwise, $t \notin D$ in which case $D \cup \{t\}$ contains the $P_4$ $xytz$ and thus, $ct_{\gamma_t}(G) \leq 2$ by \Cref{theorem:1totcontracdom}.
\end{proof}

\begin{lemma}
\label{lem:edge}
Let $G$ be a graph such that $\gamma_{t2}(G) \geq 3$. If $G$ has a minimum SD set $D$ such that one of the following holds:
\begin{itemize}
\item[(i)] $D$ contains an edge, or
\item[(ii)] there exists $x \in D$ such that $|w_D(x)| \geq 2$,
\end{itemize}
then $ct_{\gamma_{t2}}(G) \leq 2$.
\end{lemma}

\begin{proof}
Assume first that $G$ has a minimum SD set $D$ containing an edge $uv$. Let $w \in D \setminus \{u,v\}$ be a closest vertex from the edge $uv$, that is, $d(w,\{u,v\}) = \min_{x \in D \setminus \{u,v\}} d(x, \{u,v\})$. Then $d(w,\{u,v\}) \leq 3$: indeed, if $d(w,\{u,v\}) > 3$ then the vertex at distance two from $\{u,v\}$ on a shortest path from $\{u,v\}$ to $w$ is not dominated, by the choice of $w$. Now if $d(w,\{u,v\}) \leq 2$ then $u,v,w$ is a friendly triple of $D$ and thus, $ct_{\gamma_{t2}}(G) = 1$ by \Cref{thm:friendlytriple}. Otherwise, $d(w,\{u,v\}) = 3$ in which case, denoting by $x$ the vertex at distance one from $\{u,v\}$ on a shortest path from $\{u,v\}$ to $w$, $D \cup \{x\}$ contains the $O_4$ $u,v,x,w$ and thus, $ct_{\gamma_{t2}}(G) \leq 2$ by \Cref{thm:friendlytriple} which proves item~(i). 

To prove item (ii), assume that $G$ has a minimum SD set $D$ such that there exists $x \in D$ where $|w_D(x)| \geq 2$, say $y,z \in w_D(x)$. If $d(x,\{y,z\}) = 1$ then we conclude as previously that $ct_{\gamma_{t2}}(G) \leq 2$. Otherwise, $d(x,y) = d(x,z) =2$ in which case, denoting by $t$ a common neighbour of $x$ and $y$, $D \cup \{t\}$ contains the $O_4$ $y,t,x,z$ and thus, $ct_{\gamma_{t2}}(G) \leq 2$ by \Cref{thm:friendlytriple}. 
\end{proof}

\begin{lemma}
\label{lem:noedge}
Let $G$ be a graph such that $ct_{\gamma_{t2}}(G) \leq 2$ and for every minimum SD set $D$ of $G$ the following hold.
\begin{itemize}
\item[(i)] $D$ contains no edge.
\item[(ii)] For every $x \in D$, $|w_D(x)| =1$.
\end{itemize} 
Then every SD set $D$ of $G$ of size $\gamma_{t2}(G)+1$ containing an ST-configuration is either a minimal SD set of $G$ or contains an $O_6$.
\end{lemma}

\begin{proof}
Let $D$ be an SD set of size $\gamma_{t2}(G)+1$ containing an ST-configuration and suppose that $D$ is not minimal. Then there exists $x \in D$ such that $D_x = D \setminus \{x\}$ is a minimum SD set of $G$; in particular, $D_x$ contains no edge by item (i). It follows that $D$ contains no $O_1, O_2, O_3$ and $O_7$: indeed, since these ST-configurations contain two vertex-disjoint edges, if $D$ contains one of them, then $D_x$ must contain at least one edge. Now if $D$ contains an $O_4$ on vertices $u,v,w,t$ where $uv,vw \in E(G)$ and $d(w,t) =2$, necessarily $x = v$ ($D_x$ would otherwise contain an edge); but then, $u,t \in w_{D_x}(w)$, a contradiction to item (ii). Similarly, if $D$ contains an $O_5$ on vertices $u,v,w,t$ where $t$ is the vertex of degree three, necessarily $x=t$ ($D_x$ would otherwise contain an edge); but then, $u,v \in w_{D_x}(w)$, a contradiction to item (ii). Thus, $D$ contains an $O_6$.
\end{proof}

Let us finally remark the following.

\begin{observation}
\label{obs:hard12}
Let $\pi \in \{\gamma_t,\gamma_{t2}\}$ and let $\mathcal{G}$ be a graph class where {\sc Contraction Number($\pi$,3)} is polynomial-time solvable. Then {\sc Contraction Number($\pi$,1)} is $\mathsf{NP}$-hard on $\mathcal{G}$ if and only if {\sc Contraction Number($\pi$,2)} is $\mathsf{coNP}$-hard on $\mathcal{G}$.
\end{observation}

%------------------------------------------------------------------------------------------------------------------------------------------------------------------------------------

\section{Total Domination}
\label{sec:TD}

In this section, we consider the {\sc Contraction Number($\gamma_t$,$k$)} problem for $k\in\{2,3\}$. In \Cref{sec:tdeasy}, we cover polynomial-time solvable cases and in \Cref{sec:tdhard}, we examine hard cases. The proof of \Cref{thm:dictd2} can be found in \Cref{sec:dictd2}. 

%------------------------------------------------------------------------------------------------------------------------------------------------------------------------------------

\subsection{Polynomial-time solvable cases}
\label{sec:tdeasy}

The algorithms for {\sc Contraction Number($\gamma_t$,2)} and {\sc Contraction Number($\gamma_t$,3)} outlined thereafter will rely on the following key result.

\begin{lemma}
\label{lem:ectdp6kp3}
For every $k \geq 0$, {\sc 2-Edge Contraction($\gamma_t$)} is polynomial-time solvable on $(P_6+kP_3)$-free graphs.
\end{lemma}

\begin{proof}
We prove the result by induction on $k$.\\

\noindent
\textbf{Base Case.} \emph{$k=0$.} Let $G$ be a $P_6$-free graph such that $\gamma_t(G) \geq 3$. We show that $ct_{\gamma_t}(G) \leq 2$, that is, $G$ is a \yes-instance for {\sc 2-Edge Contraction($\gamma_t$)}. To this end, let $D$ be a minimum TD set of $G$. Consider two adjacent vertices $u,v \in D$ and let $w \in D \setminus \{u,v\}$ be a closest vertex from $\{u,v\}$, that is, $d(w,\{u,v\}) = \min_{x \in D\setminus \{u,v\}} d(x,\{u,v\})$. Then $d(w,\{u,v\}) \leq 3$: indeed, if $d(w,\{u,v\}) > 3$ then the vertex at distance two from $\{u,v\}$ on a shortest path from $\{u,v\}$ to $w$ is not dominated. Now if $d(w,\{u,v\}) \leq 2$ then $ct_{\gamma_t}(G) \leq 2$ by \Cref{clm:condct2}; thus, we assume henceforth that $d(w,\{u,v\}) = 3$, say $d(w,\{u,v\}) = d(w,u)$ without loss of generality. 

Let $P= uxyw$ be a shortest path from $u$ to $w$ and let $t \in D$ be a neighbour of $w$. Suppose first that $t$ is nonadjacent to $y$ (note that $t$ is nonadjacent to $x$ as $d(w,\{u,v\}) = \min_{b \in D\setminus \{u,v\}} d(b,\{u,v\}) \allowbreak = 3$). Then any neighbour $a$ of $t$ must be adjacent to $w,y$ or $x$, for $atwyxv$ would otherwise induce a $P_6$ (note indeed that $a$ is nonadjacent to $v$ as $d(w,\{u,v\}) = \min_{b \in D\setminus \{u,v\}} d(b,\{u,v\}) = 3$). But then, $(D \setminus \{t\}) \cup \{x,y\}$ is a TD set of $G$ of size $\gamma_t(G) +1$ containing the $P_4$ $vxyw$ and thus, $ct_{\gamma_t}(G) \leq 2$ by \Cref{theorem:1totcontracdom}. Second, suppose that $t$ is adjacent to $y$. If every neighbour of $v$ is adjacent to $x$ or $y$, then $(D \setminus \{v\}) \cup \{x,y\}$ is a TD set of $G$ of size $\gamma_t(G)+1$ containing the $P_4$ $xywt$ and so, $ct_{\gamma_t}(G) \leq 2$ by \Cref{theorem:1totcontracdom}. Thus, assume that $v$ has a neighbour $z$ which is nonadjacent to both $x$ and $y$. If $z =u$ then, by symmetry, we conclude as in the previous case. Suppose therefore that $z \neq u$. Then every neighbour $a$ of $w$ is adjacent to $z,x$ or $y$, for $awyxvz$ would otherwise induce a $P_6$; and we conclude symmetrically that every neighbour of $t$ is adjacent to $z,x$ or $y$. It then follows that $(D \setminus \{w,t\}) \cup \{x,y,z\}$ is a TD set of $G$ of size $\gamma_t(G)+1$ containing the $P_4$ $yxvz$ and thus, $ct_{\gamma_t}(G) \leq 2$ by \Cref{theorem:1totcontracdom}.\\

\noindent
\textbf{Inductive step.} Let $G$ be a $(P_6+kP_3)$-free graph. We aim to show that if $G$ contains an induced $P_6+(k-1)P_3$ and $G$ is a \no-instance for {\sc 2-Edge Contraction($\gamma_t$)} then $\gamma_t(G)$ is bounded by some function $f$ of $k$ (and $k$ only). Assuming for now that this claim is correct, the following algorithm solves the {\sc 2-Edge Contraction($\gamma_t$)} problem on $(P_6+kP_3)$-free graphs.\\

\begin{itemize}
\item[1.] If $G$ contains no induced $P_6+(k-1)P_3$ then use the algorithm for $(P_6+(k-1)P_3)$-free graphs to determine whether $G$ is a \yes-instance for {\sc 2-Edge Contraction($\gamma_t$)} or not.
\item[2.] Check whether there exists a TD set of $G$ of size at most $f(k)$.
\begin{itemize}
\item[2.1] If the answer is no then output \yes.
\item[2.2] Check whether there exists a minimum TD set of $G$ containing a $P_3$, or a TD set of $G$ of size $\gamma_t(G)+1$ containing a $P_4,K_{1,3}$ or $2P_3$ (see \Cref{theorem:1totcontracdom}) and output the answer accordingly. 
\end{itemize}
\end{itemize}

\bigskip

Now observe that checking whether $G$ is $(P_6+(k-1)P_3)$-free can be done in time $|V(G)|^{O(k)}$ and that step 2 can be done in time $|V(G)|^{O(f(k))}$ by simple brute force. Since {\sc 2-Edge Contraction(\allowbreak $\gamma_t$)} is polynomial-time solvable on $(P_6+(k-1)P_3)$-free graphs by the induction hypothesis, the above algorithm indeed runs in polynomial time (for fixed $k$). The remainder of this proof is devoted to showing that $f(k) =k^4 + 4k^2 + 21k + 19$.\\

Assume henceforth that $G$ contains an induced $P_6+(k-1)P_3$. Let $A \subseteq V(G)$ be a set of vertices such that $G[A]$ is isomorphic to $P_6+(k-1)P_3$, let $B \subseteq V(G) \setminus A$ be the set of vertices at distance one from $A$ and let $C = V(G) \setminus (A \cup B)$. Note that since $G$ is $(P_6+kP_3)$-free, $C$ is a disjoint union of cliques; we denote by $\mathcal{K}$ the set of maximal cliques in $C$. Now let $D$ be a minimum TD set of $G$ such that $|D \cap B|$ is maximum amongst all minimum TD set of $G$. Denote by $\mathcal{K}_D \subseteq \mathcal{K}$ the set of cliques $K \in \mathcal{K}$ such that $D \cap (N(K) \cap B) \neq \varnothing$, and set $\overline{\mathcal{K}_D} = \mathcal{K} \setminus \mathcal{K}_D$. We aim to upperbound $|D \cap N[\mathcal{K}_D]|$ and $|D \cap N[\overline{\mathcal{K}_D}]|$ when $G$ is a \no-instance for {\sc 2-Edge Contraction($\gamma_t$)}. To this end, we first prove the following.

\begin{claim}
\label{obs:DK}
If $G$ is a \no-instance for {\sc 2-Edge Contraction($\gamma_t$)} then for every $K \in \mathcal{K}$, exactly one of the following holds.
\begin{itemize}
\item[(i)] $D \cap V(K) \neq \varnothing$ and $D \cap N[K] = \{x,y\}$ for some edge $xy \in E(G)$.
\item[(ii)] $N[K] \cap D \subseteq B$.
\end{itemize}
\end{claim}

\begin{claimproof}
Assume that $G$ is a \no-instance for {\sc 2-Edge Contraction($\gamma_t$)} and consider a clique $K \in \mathcal{K}$. Then $D \cap N[K] \neq \varnothing$ as the vertices of $K$ must be dominated. Now if $D \cap V(K) \neq \varnothing$, say $x \in D \cap V(K)$, then $x$ has a neighbour $y \in D$ as $x$ must be dominated; but then, $D \cap N[K] = \{x,y\}$ for otherwise $ct_{\gamma_t}(G) \leq 2$ by \Cref{clm:condct2}  (note indeed that every vertex in $N[K]$ is within distance at most two from $x$). Otherwise, $D \cap V(K) = \varnothing$ in which case $D \cap (N(K) \cap B) \neq \varnothing$.
\end{claimproof} 

\begin{claim}
\label{clm:KD}
If $G$ is a \no-instance for {\sc 2-Edge Contraction($\gamma_t$)} then $|D \cap N[\mathcal{K}_D]| \leq 4|A|$.
\end{claim}

\begin{claimproof}
Assume that $G$ is a \no-instance for {\sc 2-Edge Contraction($\gamma_t$)} and suppose for a contradiction that $|D \cap N[\mathcal{K}_D]| > 4|A|$. We contend that at least half of those vertices belong to $B$, that is, $|D \cap (N(\mathcal{K}_D) \cap B)| > 2|A|$.  Indeed, for every $x \in D \cap (N(\mathcal{K}_D) \cap B)$, denote by $\mathcal{K}_D^x$ the set of cliques $K \in \mathcal{K}_D$ such that $x \in N(K)$. Let us first show that for every $x \in D \cap (N(\mathcal{K}_D) \cap B)$, there exists at most one clique $K \in \mathcal{K}_D^x$ such that $D \cap V(K) \neq \varnothing$, and that furthermore, if such a clique exists then it in fact contains only one element of $D$. Consider $x \in D \cap (N(\mathcal{K}_D) \cap B)$. If $|\mathcal{K}_D^x| = 1$ then the result readily follows from \Cref{obs:DK}. Suppose therefore that $|\mathcal{K}_D^x| \geq 2$. Then for all but at most one of the cliques in $\mathcal{K}_D^x$, item (ii) of \Cref{obs:DK} must hold: indeed, if there exist $K_1,K_2 \in \mathcal{K}_D^x$ such that $K_1$ and $K_2$ both satisfy item (i) of \Cref{obs:DK}, then $D$ contains the $P_3$ $y_1xy_2$ where $y_1 \in V(K_1)$ and $y_2 \in V(K_2)$, a contradiction to \Cref{theorem:1totcontracdom}. 

Now observe that for any $x,y \in D \cap (N(\mathcal{K}_D) \cap B)$, if there exist $K_x \in \mathcal{K}_D^x$ and $K_y \in \mathcal{K}_D^y$ such that $D \cap V(K_x) \neq \varnothing$ and $D \cap V(K_y) \neq \varnothing$, then surely $K_x \neq K_y$ as otherwise $|D \cap N[K_x]| \geq 3$, a contradiction to \Cref{obs:DK}(i). Since for every $K \in \mathcal{K}_D$, there exists $x \in D \cap (N(\mathcal{K}_D) \cap B)$ such that $K \in \mathcal{K}_D^x$, by definition of $\mathcal{K}_D$, it follows that 
\begin{equation*}
\begin{split}
|D \cap (N(\mathcal{K}_D) \cap B)|~&\geq \sum_{x \in D \cap (N(\mathcal{K}_D) \cap B)} |D \cap \bigcup_{K \in \mathcal{K}_D^x} V(K)|\\ 
&\geq |D \cap \bigcup_{x \in D \cap (N(\mathcal{K}_D) \cap B)}  \bigcup_{K \in \mathcal{K}_D^x} V(K)|\\
&\geq |D \cap \bigcup_{K \in \mathcal{K}_D} V(K)|
\end{split}
\end{equation*}
But
\begin{equation*}
\begin{split}
|D \cap N[\mathcal{K}_D]| &= |D \cap (N(\mathcal{K}_D) \cap B)| + |D \cap \bigcup_{K \in \mathcal{K}_D} V(K)|
\end{split}
\end{equation*}
and so, we conclude that
\begin{equation*}
\begin{split}
2|D \cap (N(\mathcal{K}_D) \cap B)| &\geq |D \cap N[\mathcal{K}_D]| > 4|A|
\end{split}
\end{equation*}
as claimed. It follows that there must exist at least three vertices $x,y,z \in D \cap (N(\mathcal{K}_D )\cap B)$ such that $x,y$ and $z$ have a common neighbour in $A$: indeed, if no such three vertices exist then every vertex in $A$ has at most two neighbours in $D \cap (N(\mathcal{K}_D )\cap B)$ and so, $|D \cap (N(\mathcal{K}_D )\cap B)| \leq 2|A|$, a contradiction to the above. This implies, in particular, that for every $u,v \in \{x,y,z\}$, $d(u,v) \leq 2$; but $x$ must have a neighbour in $D$ (possibly $y$ or $z$) and so, $ct_{\gamma_t}(G) \leq 2$ by \Cref{clm:condct2}, a contradiction to our assumption. Therefore, $|D \cap N[\mathcal{K}_D]| \leq 4|A|$.
\end{claimproof}

Now note that since for every $K \in \overline{\mathcal{K}_D}$, $D \cap (N(K) \cap B) = \varnothing$ by definition, the following ensues from \Cref{obs:DK}.

\begin{observation}
\label{obs:DcK}
If $G$ is a \no-instance for {\sc 2-Edge Contraction($\gamma_t$)} then for every $K \in \overline{\mathcal{K}_D}$, $|D \cap N[K]| = |D \cap V(K)|= 2$.
\end{observation}

For every $K \in \overline{\mathcal{K}_D}$, let us denote by $x_Ky_K \in E(G)$ the edge contained in $V(K) \cap D$ when $G$ is a \no-instance for {\sc 2-Edge Contraction($\gamma_t$)}. Observe that by \Cref{clm:condct2}, the following holds.

\begin{observation}
\label{obs:nocn}
If $G$ is a \no-instance for {\sc 2-Edge Contraction($\gamma_t$)} then for every clique $K \in \overline{\mathcal{K}_D}$, $d(\{x_K,y_K\}, D \setminus \{x_K,y_K\}) > 2$.
\end{observation}

\begin{claim}
\label{clm:PN}
If $G$ is a \no-instance for {\sc 2-Edge contraction($\gamma_t$)} then for every $K \in \overline{\mathcal{K}_D}$, the following hold.
\begin{itemize}
\item[(i)] For every $u \in \{x_K,y_K\}$, $PN_D(u) \neq \varnothing$. 
\item[(ii)] For $u \in \{x_K,y_K\}$ and $v \in \{x_K,y_K\} \setminus \{u\}$, no vertex of $PN_D(u)$ is complete to $PN_D(v)$.
\end{itemize}
\end{claim}

\begin{claimproof}
Assume that $G$ is a \no-instance for {\sc 2-Edge Contraction($\gamma_t$)} and consider a clique $K \in \overline{\mathcal{K}_D}$. Observe first the $PN_D(x_K) \cup PN_D(y_K) \subseteq B$ as $K$ is a clique. Furthermore, we may assume that $(N(x_K) \cup N(y_K)) \cap B \neq \varnothing$ for it otherwise suffices to consider, in place of $D$, the TD set $(D \setminus \{x_K\}) \cup \{x\}$, where $x \in V(K)$ has at least one neighbour in $B$. Now suppose that $PN_D(u) = \varnothing$ for some $u \in \{x_K,y_K\}$, say $PN_D(x_K) = \varnothing$ without loss of generality. Then $N(y_K) \cap B \neq \varnothing$: if not, then $N(x_K) \cap B \neq \varnothing$ by the above, and since no vertex in $N(x_K) \cap B$ can be adjacent to another vertex in $D \setminus \{x_K,y_K\}$ by \Cref{obs:nocn}, necessarily $N(x_K) \cap B \subseteq PN_D(x_K)$, a contradiction to our assumption. But then, letting $y \in B$ be a neighbour of $y_K$, the set $(D \setminus \{x_K\}) \cup \{y\}$ is a minimum TD set of $G$ containing strictly more vertices from $B$ than $D$, a contradiction to the choice of $D$. Thus, $PN_D(u) \neq \varnothing$ for every $u \in \{x_K,y_K\}$. Now if for some $u \in \{x_K,y_K\}$, there exists $y \in PN_D(u)$ such that $y$ is complete to $PN_D(v)$ where $v \in \{x_K,y_K\} \setminus \{u\}$, then the set $(D \setminus \{v\}) \cup \{y\})$ is a minimum TD set of $G$ containing strictly more vertices from $B$ than $D$, a contradiction to the choice of $D$.  
\end{claimproof}

\begin{claim}
\label{clm:AcK}
If $G$ is a \no-instance for {\sc 2-Edge Contraction($\gamma_t$)} and there exists a set $S \subseteq \overline{\mathcal{K}_D}$ such that for every $K,K' \in S$, $PN_D(x_K) \cup PN_D(y_K) \cdots PN_D(x_{K'}) \cup PN_D(y_{K'})$ then $|S| \leq |A|$.
\end{claim}

\begin{claimproof}
Assume that $G$ is a \no-instance for {\sc 2-Edge Contraction($\gamma_t$)} and there exists such a set $S \subseteq \overline{\mathcal{K}_D}$. Suppose for a contradiction that $|S| > |A|$. Then there must exist $u \in PN_D(x_K) \cup PN_D(y_K)$ and $v \in PN_D(x_{K'}) \cup PN_D(y_{K'})$ for two distinct cliques $K,K' \in S$, such that $u$ and $v$ have a common neighbour $w \in A$. Then for all but at most $k-1$ cliques $K'' \in S$, $w$ is complete to $PN_D(x_{K''}) \cup PN_D(y_{K''})$: indeed, if there exist at least $k$ cliques $K_1,\ldots,K_k \in S$ such that for every $i \in [k]$, $w$ is nonadjacent to a vertex $u_i \in PN_D(x_{K_i}) \cup PN_D(y_{K_i})$, then $\{x_K,y_K,u,w,v,x_{K'},y_{K'}\} \cup \bigcup_{i \in [k]} \{u_i,x_{K_i},y_{K_i}\}$ induces a $P_7 + kP_3$, a contradiction. Since $|S| > |A| = 6+3(k-1)$, it follows that there are at least two cliques $K_1,K_2 \in S \setminus \{K,K'\}$ such that $w$ is complete to $PN_D(x_{K_1}) \cup PN_D(y_{K_1})$ and to $PN_D(x_{K_2}) \cup PN_D(y_{K_2})$. However, by \Cref{clm:PN}(i), $PN_D(y_{K_1}) \neq \varnothing$ and $PN_D(y_{K_2}) \neq \varnothing$ and so, the set $(D \setminus \{x_{K_1},x_{K_2}\}) \cup \{w,u_1,u_2\}$ where $u_1 \in PN_D(y_{K_1})$ and $u_2 \in PN_D(y_{K_2})$, is a TD set of $G$ of size $\gamma_t(G)+1$ containing the $P_4$ $y_{K_1}u_1wu_2$, a contradiction to \Cref{theorem:1totcontracdom}.
\end{claimproof}

\begin{claim}
\label{clm:sets}
Let $K,K' \in \overline{\mathcal{K}_D}$. If there exist $u \in PN_D(x_K) \cup PN_D(y_K)$ and $v \in PN_D(x_{K'}) \cup PN_D(y_{K'})$ such that $uv \in E(G)$, then there exist a set $S_{uv} \subseteq \overline{\mathcal{K}_D} \setminus \{K,K'\}$ 
and a set $T_{uv} \subseteq \bigcup_{L \in S_{uv}} PN_D(x_L) \cup PN_D(y_L)$ such that the following hold.
\begin{itemize}
\item[(i)] For every $L \in S_{uv}$, every vertex in $PN_D(x_L) \cup PN_D(y_L)$ is adjacent to at least one vertex of $T_{uv} \cup \{u,v\}$.
\item[(ii)] $|\overline{\mathcal{K}_D} \setminus S_{uv}| \leq 2 + k(k-1)/2$.
\end{itemize}
\end{claim} 

\begin{claimproof}
Assume that there exist $u \in PN_D(x_K) \cup PN_D(y_K)$ and $v \in PN_D(x_{K'}) \cup PN_D(y_{K'})$ such that $uv \in E(G)$. Let $T \subseteq \bigcup_{L \in \overline{\mathcal{K}_D} \setminus \{K,K'\}} PN_D(x_L) \cup PN_D(y_L)$ be a maximum size independent set such that $T \cdots \{u,v\}$ and for every $L \in \overline{\mathcal{K}_D}$, $|T \cap (PN_D(x_L) \cup PN_D(y_L))| \leq 1$. Further denote by $S = \{L \in \overline{\mathcal{K}_D}~|~ T \cap (PN_D(x_L) \cup PN_D(y_L)) \neq \varnothing\} \cup \{K,K'\}$. Observe that since $\{u,v\} \cup T \cup \bigcup_{L\in S} \{x_L,y_L\}$ induces $P_6+|T|P_3$, necessarily $|T| \leq k-1$ and thus, $|S| \leq k+1$. Now consider the sequence of sets constructed according to the following procedure.\\

\begin{itemize}
\item[1.] Initialise $i=1$, $T_1 = T$ and $R_1 = S_1 = S$.
\item[2.] Increase $i$ by one.
\item[3.] Let $T_i \subseteq \bigcup_{L \in \overline{\mathcal{K}_D} \setminus S_{i-1}} PN_D(x_L) \cup PN_D(y_L)$ be a maximum size independent set such that $T_i \cdots \{u,v\}$ and for every $L \in \overline{\mathcal{K}_D} \setminus S_{i-1}$, $|T_i \cap (PN_D(x_L) \cup PN_D(y_L))| \leq 1$. Set $R_i = \{L \in \overline{\mathcal{K}_D}~|~ T_i \cap (PN_D(x_L) \cup PN_D(y_L)) \neq \varnothing\}$ and $S_i = S_{i-1} \cup R_i$.
\item[4.] If $|T_i| = |T_{i-1}|$ then stop the procedure.
\item[5.] Return to step 2. 
\end{itemize}

\bigskip

\noindent
Consider the value of $i$ at the end of the procedure (note that $i \geq 2$). Let us show that we may take $$T_{uv} = T_{i-1} \cup T_i$$ and 
$$
S_{uv} = \left\{
    \begin{array}{ll}
        \overline{\mathcal{K}_D} \setminus S_{i-2} & \mbox{if } i > 2 \\
        \overline{\mathcal{K}_D} \setminus \{K,K'\} & \mbox{otherwise.}
    \end{array}
\right.
 $$
Observe first that by construction, for every $L \in \overline{\mathcal{K}_D} \setminus S_i$ and every vertex $x \in PN_D(x_L) \cup PN_D(y_L)$, $x$ is adjacent to at least one vertex in $T_i \cup \{u,v\}$ for otherwise, the procedure would have output $T_i \cup \{x\}$ in place of $T_i$. Similarly, for every $L \in R_i$ and every vertex $x \in PN_D(x_L) \cup PN_D(y_L)$, $x$ is adjacent to at least one vertex in $T_{i-1} \cup \{u,v\}$ for otherwise, the procedure would have output $T_{i-1} \cup \{x\}$ in place of $T_{i-1}$; and for $L \in R_{i-1}$ and every vertex $x \in PN_D(x_L) \cup PN_D(y_L)$, $x$ is adjacent to at least one vertex in $T_i \cup \{u,v\}$ for otherwise, the procedure would have output $T_i \cup \{x\}$ in place of $T_{i-1}$ (recall that by construction $|T_{i-1}| = |T_i|$). In particular, every vertex in $T_i$ is adjacent to at least one vertex in $T_{i-1}$ and thus, item (i) holds true. Now for every $1 \leq p < q \leq i-1$, $|T_q| < |T_p|$ by construction; and since $|T_1| \leq k-1$, it follows that $i \leq k+1$ and for every $1 \leq p \leq i-1$, $|T_p| \leq k-p$. Thus, if $i > 2$ then
\[
|S_{i-2}| = 2 + \sum_{1 \leq p \leq i-2} |T_p| \leq 2 + \sum_{1 \leq p \leq i-2} k-p \leq 2 + \frac{k(k-1)}{2}
\]
and so, item (ii) holds true.
\end{claimproof}

In the following, for any $K,K' \in \overline{\mathcal{K}_D}$, $u \in PN_D(x_K) \cup PN_D(y_K)$ and $v \in PN_D(x_{K'}) \cup PN_D(y_{K'})$ such that $uv \in E(G)$, we denote by $S_{uv} \subseteq \overline{\mathcal{K}_D}$ and $T_{uv} \subseteq B$ the two sets given by \Cref{clm:sets}. 

\begin{claim}
\label{clm:yesinst}
If there exist $K_1,K'_1,K_2,K'_2 \in \overline{\mathcal{K}_D}$ such that 
\begin{itemize}
\item[(i)] $K_2,K'_2 \in S_{u_1u'_1}$ for some $u_1 \in PN_D(x_{K_1}) \cup PN_D(y_{K_1})$ and $u'_1 \in PN_D(x_{K'_1}) \cup PN_D(y_{K'_1})$ where $u_1u'_1 \in E(G)$, and
\item[(ii)] one of $K_1$ and $K'_1$ belongs to $S_{u_2u'_2}$ for some $u_2 \in PN_D(x_{K_2}) \cup PN_D(y_{K_2})$ and $u'_2 \in PN_D(x_{K'_2}) \cup PN_D(y_{K'_2})$ where $u_2u'_2 \in E(G)$,
\end{itemize}
then $ct_{\gamma_t}(G) \leq 2$.
\end{claim}

\begin{claimproof}
Assume that four such cliques $K_1,K'_1,K_2,K'_2 \in \overline{\mathcal{K}_D}$ exist and let $u_1,u'_1,u_2,u'_2$ be the vertices given by items (i) and (ii). Then for every $z \in T_{u_1u'_1}$, there exists, by construction, a clique $L_z \in S_{u_1u'_1}$ such that $z$ is a private neighbour of $x_{L_z}$ or $y_{L_z}$: let us denote by $w_z \in \{x_{L_z},y_{L_z}\}$ the vertex such that $z$ is not a private neighbour of $w_z$. Similarly, for every $z \in T_{u_2u'_2}$, there exists a clique $L_z \in S_{u_2u'_2}$ such that $z$ is a private neighbour of $x_{L_z}$ or $y_{L_z}$: let us denote by $w_z \in \{x_{L_z},y_{L_z}\}$ the vertex such that $z$ is not a private neighbour of $w_z$. Now assume without loss of generality that for every $i \in [2]$, $u_i \in PN_D(y_{K_i})$ and $u'_i \in PN_D(y_{K'_i})$, and that furthermore, $K'_1 \in S_{u_2u'_2}$. We contend that the set $D' = (D \setminus (\{w_z~|~ z \in T_{u_1u'_1} \cup T_{u_2u'_2}\} \cup \{x_{K'_1},x_{K_2},x_{K'_2}\})) \cup T_{u_1u'_1} \cup T_{u_2u'_2} \cup \{u_1,u'_1,u_2,u'_2\}$ is a TD set of $G$. Indeed, by \Cref{clm:sets}(i), every vertex in $PN_D(w_z)$ for $z \in T_{u_1u'_1}$, is adjacent to at least one vertex in $T_{u_1u'_1} \cup \{u_1,u'_1\}$; and similarly, every vertex in $PN_D(w_z)$ for $z \in T_{u_2u'_2}$, is adjacent to at least one vertex in $T_{u_2u'_2} \cup \{u_2,u'_2\}$. Furthermore, by \Cref{clm:sets}(i), since $K_2,K'_2 \in S_{u_1u'_1}$ and $K'_1 \in S_{u_2u'_2}$, every vertex in $PN_D(x_{K_2}) \cup PN_D(x_{K'_2})$ is adjacent to at least one vertex in $T_{u_1u'_1} \cup \{u_1,u'_1\}$, and every vertex in $PN_D(x_{K'_1})$ is adjacent to at least one vertex in $T_{u_2u'_2} \cup \{u_2,u'_2\}$. It follows that $D'$ is indeed a TD set of $G$; and since $|D'| = \gamma_t(G) + 1$ and $D'$ contains the $P_4$ $y_{K_1}u_1u'_1y'_{K'_1}$, we conclude by \Cref{theorem:1totcontracdom} that $ct_{\gamma_t}(G) \leq 2$.
\end{claimproof}

\begin{claim}
\label{clm:barKD}
If $G$ is a \no-instance for {\sc 2-Edge Contraction($\gamma_t$)} then $|\overline{\mathcal{K}_D}| \leq k^2(k(k-1)/2+2)+|A| -1$.
\end{claim}

\begin{claimproof}
Assume that $G$ is a \no-instance for {\sc 2-Edge Contraction($\gamma_t$)} and suppose to the contrary that $|\overline{\mathcal{K}_D}| > k^2(k(k-1)/2+2)+|A| -1$. Let us show that there then exist four cliques in $\overline{\mathcal{K}_D}$ satisfying items (i) and (ii) of \Cref{clm:yesinst}, which would contradict the fact that $ct_{\gamma_t}(G) > 2$. Let $K_1,K'_1 \in \overline{\mathcal{K}_D}$ be two cliques for which there exist $u_1 \in PN_D(x_{K_1}) \cup PN_D(y_{K_1})$ and $u'_1 \in PN_D(x_{K'_1}) \cup PN_D(y_{K'_1})$ such that $u_1u'_1 \in E(G)$ (the existence of such cliques is guaranteed by \Cref{clm:AcK}). We claim that the algorithm below always outputs four cliques of $\overline{\mathcal{K}_D}$ satisfying items (i) and (ii) of \Cref{clm:yesinst}.\\

\begin{itemize}
\item[1.] Initialise $i=1$ and $C_1 = S_{u_1u'_1}$.
\item[2.] Increase $i$ by one.
\item[3.] If there exist $K,K' \in C_{i-1}$ with $u \in PN_D(x_K) \cup PN_D(y_K)$ and $u' \in PN_D(x_{K'}) \cup PN_D(y_{K'})$ such that 
\begin{itemize}
\item[$\cdot$] $uu' \in E(G)$ and 
\item[$\cdot$] $K_{i-1} \in S_{uu'}$ or $K'_{i-1} \in S_{uu'}$,
\end{itemize}
then output $K_{i-1},K'_{i-1},K,K'$.
\item[4.] If $|C_{i-1}| \geq |A|$ then let $K_i,K'_i \in C_{i-1}$ be two cliques for which there exist $u_i \in PN_D(x_{K_i}) \cup PN_D(y_{K_i})$ and $u'_i \in PN_D(x_{K'_i}) \cup PN_D(y_{K'_i})$ such that $u_iu'_i \in E(G)$; and set $C_i = C_{i-1} \cap S_{u_iu'_i}$.
\item[5.] If $|C_{i-1}| < |A|$ then set $C_i = C_{i-1}$.
\item[6.] Return to step 2.
\end{itemize}

\bigskip

\noindent
Before showing correctness of the above algorithm, let us first note that for every $j \geq 1$ satisfying the condition in step 4 (that is, $|C_j| \geq |A|$), the existence of the cliques $K_j$ and $K'_j$ is guaranteed by \Cref{clm:AcK}. Now let us show that if the above algorithm terminates then its output is indeed as claimed. Assume that the algorithm terminates when the counter $i$ equals some value $j \geq 1$ and let $K_{j-1},K'_{j-1},K,K'$ be the output. Observe first that since the algorithm terminates, the condition in step 5 is never satisfied during its run (the algorithm would have otherwise looped indefinitely); in particular $C_{j-1} = \bigcap_{1 \leq \ell \leq j-1} S_{u_\ell u'_\ell}$. Since by construction, $K$ and $K'$ both belong to $C_{j-1}$, they belong in particular to $S_{u_{j-1}u'_{j-1}}$ and so, item (i) of \Cref{clm:yesinst} indeed holds. Furthermore, by construction, there exist $u \in PN_D(x_K) \cup PN_D(y_K)$ and $u' \in PN_D(x_{K'}) \cup PN_D(y_{K'})$ with $uu'\in E(G)$, such that $K_{j-1} \in S_{uu'}$ or $K'_{j-1} \in S_{uu'}$, and thus, item (ii) of \Cref{clm:yesinst} holds true as well. 

There remains to show that the algorithm indeed terminates. To this end, let us first show by induction that for every $j \in [k^2+1]$, the condition in step 4 is always satisfied, that is, $|C_{j-1}| \geq |A|$. More specifically, we show that for every $j \in [k^2]$, $|C_j| \geq (k^2-j)(k(k-1)/2 + 2) + |A|$. For $j=1$, the result readily follows from the fact that $|\overline{\mathcal{K}_D}| = |C_1| + |\overline{\mathcal{K}_D} \setminus C_1| \geq k^2(k(k-1)/2+2) + |A|$ by assumption and $|\overline{\mathcal{K}_D} \setminus C_1| \leq k(k-1)/2 +2$ by \Cref{clm:sets}(ii). For $j > 1$, $|C_j| = |S_{u_ju'_j} \cap C_{j-1}|$ by definition and so, 
\begin{equation*}
\begin{split}
|C_j| &= |\overline{\mathcal{K}_D}| - |\overline{\mathcal{K}_D} \setminus C_j|\\
&= |\overline{\mathcal{K}_D}| - |(\overline{\mathcal{K}_D}\setminus S_{u_ju'_j}) \cup (\overline{\mathcal{K}_D} \setminus C_{j-1})|\\
&\geq |\overline{\mathcal{K}_D}| - (|\overline{\mathcal{K}_D}\setminus S_{u_ju'_j}| + |\overline{\mathcal{K}_D} \setminus C_{j-1}|)\\
&= |C_{j-1}| - |\overline{\mathcal{K}_D}\setminus S_{u_ju'_j}|.
\end{split}
\end{equation*}
But $|C_{j-1}| \geq (k^2-(j-1))(k(k-1)/2 +2) + |A|$ by the induction hypothesis and $|\overline{\mathcal{K}_D}\setminus S_{u_ju'_j}| \geq k(k-1)/2+2$ by \Cref{clm:sets}(ii) and thus, $|C_j| \geq (k^2-j)(k(k-1)/2+2) + |A|$ as claimed.

Now suppose to the contrary that the algorithm has not yet terminated by the time the counter $i$ reaches the value $t = k(k-1)/4 + 3$. Then for every $j < t$, $K_j,K'_j \notin S_{u_tu'_t}$: indeed, if there exist indices $j \in [t-1]$ such that one of $K_j$ and $K'_j$ belongs to $S_{u_tu'_t}$ then, letting $\ell$ be the smallest such index, we have $C_{t-1} \subseteq C_\ell$ and $K_t,K'_t \in C_{t-1}$; but then, the algorithm would have terminated after setting the counter $i$ to $\ell+1$ and output $K_\ell,K'_\ell,K=K_t,K'=K'_t$. It follows that $\{K_j,K'_j~|~j\in [t-1]\} \subseteq \overline{\mathcal{K}_D} \setminus S_{u_tu'_t}$; but $|\{K_j,K'_j~|~j\in [t-1]\}| = 2(t-1) = k(k-1)/2 + 4$, a contradiction to \Cref{clm:sets}(ii).
\end{claimproof}

To conclude, assume that $G$ is a \no-instance for {\sc 2-Edge Contraction($\gamma_t$)}. Then by \Cref{obs:DcK} and \Cref{clm:barKD}, $$|D \cap N[\overline{\mathcal{K}_D}]| = 2|\overline{\mathcal{K}_D}| \leq k^2(k(k-1)+4) + 2(|A|-1)$$ and thus, combined with \Cref{clm:KD}, we conclude that $$|D \cap N[C]| = |D \cap N[\mathcal{K}_D]| + |D \cap N[\overline{\mathcal{K}_D}| \leq k^2(k(k-1)+4) + 6|A| -2.$$ Now observe that $D'=D \setminus N[C]$ has size at most $|A|$: indeed, since $D'$ dominates (only) vertices from $A \cup B$, if $|D'| > |A|$ then $(D \setminus D') \cup A$ is a TD set of $G$ of size strictly less than that of $D$, a contradiction. It follows that $$\gamma_t(G) = |D \cap N[C]| + |D \setminus N[C]| \leq k^2(k(k-1)+4) + 7|A| -2$$ 
and so, we may take $f(k) =k^4 + 4k^2 + 21k + 19$ as claimed.
\end{proof}

Since for any graph $G$, $G$ is a \yes-instance for {\sc 2-Edge Contraction($\gamma_t$)} if and only if $G$ is a \no-instance for {\sc Contraction Number($\gamma_t$,3)}, the following ensues from \Cref{lem:ectdp6kp3}.

\begin{corollary}
\label{lem:p6kp3td}
For every $k \geq 0$, {\sc Contraction Number($\gamma_t$,3)} is polynomial-time solvable on $(P_6+kP_3)$-free graphs.
\end{corollary} 

\begin{lemma}
\label{lem:easytd2}
{\sc Contraction Number($\gamma_t$,2)} is polynomial-time solvable on $H$-free graphs if $H \subseteq_i P_5+tK_1$ for some $t \geq 0$, or $H \subseteq_i P_4+tP_3$ for some $t \geq 0$.
\end{lemma}

\begin{proof}
Assume that $H \subseteq_i P_5+tK_1$ for some $t \geq 0$ (the case where $H \subseteq_i P_4+tP_3$ for some $t \geq 0$ is handled similarly) and let $G$ be an $H$-free graph. Since $H$ is a fortiori an induced subgraph of $P_6+tP_3$, we may use the polynomial-time algorithm for $(P_6+tP_3)$-free graphs given by \Cref{lem:p6kp3td} to determine whether $G$ is a \yes-instance for {\sc Contraction Number($\gamma_t$,3)} or not. If the answer is yes then we output \no; otherwise, we use the polynomial-time algorithm for $(P_5+tK_1)$-free graphs given by \Cref{thm:dictd1} to determine whether $G$ is a \yes-instance for {\sc Contraction Number($\gamma_t$,1)} or not, and output the answer accordingly.
\end{proof}

%------------------------------------------------------------------------------------------------------------------------------------------------------------------------------------

\subsection{Hardness results}
\label{sec:tdhard}

In this section, we show that {\sc Contraction Number($\gamma_t$,2)} and {\sc Contraction Number($\gamma_t$,3)} are $\mathsf{NP}$-hard on a number of monogenic graph classes. We first consider the case $k=2$.

\begin{lemma}
\label{lem:clawtd2}
{\sc Contraction Number($\gamma_t$,2)} is $\mathsf{NP}$-hard on $K_{1,3}$-free graphs.
\end{lemma}

\begin{proof}
We use the same construction as in \cite[Theorem 6]{GALBY202118}. More precisely, we reduce from {\sc Positive Exactly 3-Bounded 1-In-3 3-Sat} (see \Cref{sec:prelim} for a precise definition of this problem): given an instance $\Phi$ of this problem, with variable set $X$ and clause set $C$, we construct an instance $G$ of {\sc Contraction Number($\gamma_t$,2)} as follows. For every variable $x \in X$ appearing in clauses $c,c'$ and $c''$, we introduce the gadget $G_x$ depicted in \Cref{fig:tdclawvargad} (where a rectangle indicates that the corresponding set of vertices is a clique). For every clause $c \in C$ containing variables $x,y$ and $z$, we introduce the gadget $G_c$ which is the disjoint union of the graph $G^T_c$ and the graph $G^F_c$ depicted in \Cref{fig:tdclawclausegad} (where a rectangle indicates that the corresponding set of vertices is a clique) and further add for every $\ell \in \{x,y,z\}$, an edge between $t^c_\ell$ and $t^\ell_c$, and an edge between $f^c_\ell$ and $f^\ell_c$. We let $G$ be the resulting graph. Let us show that $\Phi$ is satisfiable if and only if $ct_{\gamma_t}(G) = 2$. To do so, we will rely on the following key result shown in \cite{GALBY202118}.

\begin{figure}
\centering
\begin{tikzpicture}[node distance=1.3cm]
\node[cir,label=right:{\small $u_x$}] (ux) at (0,0) {};
\node[cir,label=right:{\small  $v_x$}] at ($(ux) + (0,.75)$) (vx) {};
\node[cir,below left of=ux,label=below:{\small  $T_x$}] (Tx) {};
\node[cir,below right of=ux,label=below:{\small  $F_x$}] (Fx) {};

\node[cir,left of=Tx,label=below:{\small  $a^{c'}_x$}] (b1) {};
\node[cir,above of=b1,label=below:{\small  $a^c_x$}] (a1) {};
\node[cir,below of=b1,label=below:{\small  $a^{c''}_x$}] (c1) {};

\node[cir,left of=a1,label=below:{\small  $b^c_x$}] (d1) {};
\node[cir,left of=d1,label=below:{\small  $d^c_x$}] (j1) {};
\node[cir,label=above:{\small  $c^c_x$}] at ($(j1) + (.65,.4)$) (g1) {};
\node[cir,left of=j1,label=below:{\small  $t^c_x$}] (t1x) {};

\node[cir,left of=b1,label=below:{\small  $b^{c'}_x$}] (e1) {};
\node[cir,left of=e1,label=below:{\small $d^{c'}_x$}] (k1) {};
\node[cir,label=above:{\small  $c^{c'}_x$}] at ($(k1) + (.65,.4)$) (h1) {};
\node[cir,left of=k1,label=below:{\small  $t^{c'}_x$}] (t2x) {};

\node[cir,left of=c1,label=below:{\small  $b^{c''}_x$}] (f1) {};
\node[cir,left of=f1,label=below:{\small  $d^{c''}_x$}] (l1) {};
\node[cir,label=above:{\small  $c^{c''}_x$}] at ($(l1) + (.65,.4)$) (i1) {};
\node[cir,left of=l1,label=below:{\small  $t^{c''}_x$}] (t3x) {};

\node[cir,right of=Fx,label=below:{\small  $g^{c'}_x$}] (b2) {};
\node[cir,above of=b2,label=below:{\small  $g^c_x$}] (a2) {};
\node[cir,below of=b2,label=below:{\small  $g^{c''}_x$}] (c2) {};

\node[cir,right of=a2,label=below:{\small  $h^c_x$}] (d2) {};
\node[cir,right of=d2,label=below:{\small  $j^c_x$}] (j2) {};
\node[cir,label=above:{\small  $i^c_x$}] at ($(d2) + (.65,.4)$) (g2) {};
\node[cir,right of=j2,label=below:{\small  $f^c_x$}] (f1x) {};

\node[cir,right of=b2,label=below:{\small  $h^{c'}_x$}] (e2) {};
\node[cir,right of=e2,label=below:{\small  $j^{c'}_x$}] (k2) {};
\node[cir,label=above:{\small  $i^{c'}_x$}] at ($(e2) + (.65,.4)$) (h2) {};
\node[cir,right of=k2,label=below:{\small  $f^{c'}_x$}] (f2x) {};

\node[cir,right of=c2,label=below:{\small  $h^{c''}_x$}] (f2) {};
\node[cir,right of=f2,label=below:{\small  $j^{c''}_x$}] (l2) {};
\node[cir,label=above:{\small  $i^{c''}_x$}] at ($(f2) + (.65,.4)$) (i2) {};
\node[cir,right of=l2,label=below:{\small  $f^{c''}_x$}] (f3x) {};

\draw[-] (ux) -- (vx)
(ux) -- (Tx) 
(ux) -- (Fx)
(Tx) -- (Fx)
(Tx) -- (a1)
(Tx) -- (b1)
(Tx) -- (c1)
(a1) -- (d1)
(d1) -- (g1)
(d1) -- (j1)
(g1) -- (j1)
(j1) -- (t1x)
(b1) -- (e1)
(e1) -- (h1)
(e1) -- (k1)
(h1) -- (k1)
(k1) -- (t2x)
(c1) -- (f1)
(f1) -- (i1)
(f1) -- (l1)
(i1) -- (l1)
(l1) -- (t3x) 
(Fx) -- (a2)
(Fx) -- (b2)
(Fx) -- (c2)
(a2) -- (d2)
(d2) -- (g2)
(d2) -- (j2)
(g2) -- (j2)
(j2) -- (f1x)
(b2) -- (e2)
(e2) -- (h2)
(e2) -- (k2)
(h2) -- (k2)
(k2) -- (f2x)
(c2) -- (f2)
(f2) -- (i2)
(f2) -- (l2)
(i2) -- (l2)
(l2) -- (f3x); 

\draw ($(c1) + (-.3,-.65)$) rectangle ($(a1) + (.3,.3)$);
\draw ($(c2) + (-.3,-.65)$) rectangle ($(a2) + (.3,.3)$);
\end{tikzpicture}
\caption{The variable gadget $G_x$ for a variable $x$ contained in clauses $c,c'$ and $c''$ (a rectangle indicates that the corresponding set of vertices induces a clique).}
\label{fig:tdclawvargad}
\end{figure}
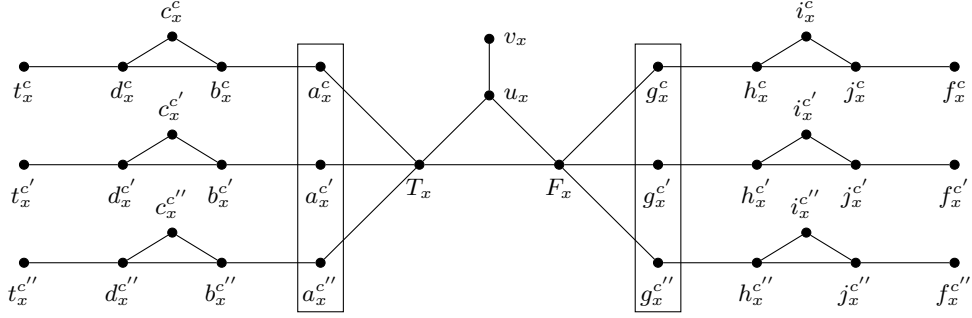

\begin{figure}
\centering
\begin{subfigure}[b]{.45\textwidth}
\centering
\begin{tikzpicture}[node distance=1.3cm]
\node[cir,label=below:{\small  $u_c$}] (u) at (0,0) {};
\node[cir,right of=u,label=below:{\small  $a^y_c$}] (b) {};
\node[cir,above of=b,label=below:{\small  $a^x_c$}] (a) {};
\node[cir,below of=b,label=below:{\small  $a^z_c$}] (c) {};

\node[cir,right of=a,label=below:{\small  $c^x_c$}] (g1) {};
\node[cir,label=above:{\small  $b^x_c$}] at ($(a) + (.65,.4)$) (d1) {};
\node[cir,right of=g1,label=below:{\small  $d^x_c$}] (j1) {};
\node[cir,right of=j1,label=below:{\small  $t^x_c$}] (txc) {};

\node[cir,right of=b,label=below:{\small  $c^y_c$}] (g2) {};
\node[cir,label=above:{\small  $b^y_c$}] at ($(b) + (.65,.4)$) (d2) {};
\node[cir,right of=g2,label=below:{\small  $d^y_c$}] (j2) {};
\node[cir,right of=j2,label=below:{\small  $t^y_c$}] (tyc) {};

\node[cir,right of=c,label=below:{\small  $c^z_c$}] (g3) {};
\node[cir,label=above:{\small  $b^z_c$}] at ($(c) + (.65,.4)$) (d3) {};
\node[cir,right of=g3,label=below:{\small  $d^z_c$}] (j3) {};
\node[cir,right of=j3,label=below:{\small  $t^z_c$}] (tzc) {};

\draw[-] (u) -- (a)
(u) -- (b)
(u) -- (c)
(a) -- (d1)
(a) -- (g1)
(d1) -- (g1)
(g1) -- (j1)
(j1) -- (txc)
(b) -- (d2)
(b) -- (g2)
(d2) -- (g2)
(g2) -- (j2)
(j2) -- (tyc)
(c) -- (d3)
(c) -- (g3)
(d3) -- (g3)
(g3) -- (j3)
(j3) -- (tzc);

\draw ($(c) + (-.25,-.55)$) rectangle ($(a) + (.25,.2)$);
\end{tikzpicture}
\caption{The graph $G^T_c$.}
\end{subfigure}
\hspace*{.5cm}
\begin{subfigure}[b]{.45\textwidth}
\centering
\begin{tikzpicture}[node distance=1.3cm]
\node[cir,label=below:{\small  $v_c$}] (vc) at (0,0) {};
\node[cir,right of=vc,label=below:{\small  $w_c$}] (a) {};
\node[cir,left of=vc,label=below:{\small  $g^y_c$}] (c) {};
\node[cir,above of=c,label=below:{\small $g^x_c$}] (b) {};
\node[cir,below of=c,label=below:{\small  $g^z_c$}] (d) {};
\node[cir,left of=b,label=below:{\small  $f^x_c$}] (fxc) {};
\node[cir,left of=c,label=below:{\small  $f^y_c$}] (fyc) {};
\node[cir,left of=d,label=below:{\small  $f^z_c$}] (fzc) {};

\draw[-] (vc) -- (a)
(vc) -- (b)
(vc) -- (c)
(vc) -- (d)
(b) -- (fxc)
(c) -- (fyc)
(d) -- (fzc);

\draw ($(d) + (-.25,-.55)$) rectangle ($(b) + (.25,.2)$);
\end{tikzpicture}
\caption{The graph $G^F_c$.}
\end{subfigure}
\caption{The clause gadget $G_c$ for a clause $c = x \lor y \lor z$ is the disjoint union of $G_c^T$ and $G_c^F$ (a rectangle indicates that the corresponding set of vertices is a clique).}
\label{fig:tdclawclausegad}
\end{figure}
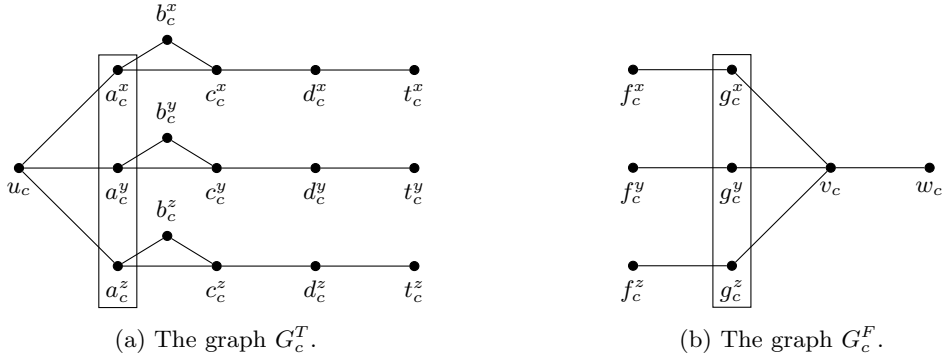

\begin{claim}[\cite{GALBY202118}]
\label{clm:tdclaw1}
The following statements are equivalent.
\begin{itemize}
\item[(i)] $\Phi$ is satisfiable.
\item[(ii)] $\gamma_t(G) = 14|X| + 8|C|$.
\item[(iii)] $ct_{\gamma_t}(G) > 1$.
\end{itemize}
\end{claim} 

Now assume that $\Phi$ is satisfiable and consider a truth assignment satisfying $\Phi$. We construct a minimum TD set $D$ of $G$ as follows. For every variable $x \in X$ appearing in clauses $c,c'$ and $c''$, if $x$ is true then we include $\{u_x,T_x\} \cup \{d^p_x,t^p_x,h^p_x,j^p_x~|~p \in \{c,c',c''\}\}$ in $D$; otherwise, we include $\{u_x,F_x\} \cup \{b^p_x,d^p_x,j^p_x,f^p_x~|~p \in \{c,c',c''\}\}$ in $D$. For every clause $c \in C$ containing variables $x,y$ and $z$, exactly one variable is true, say $x$ without loss of generality, in which case we include $\{v_c,g^x_c\} \cup \{c_c^x,a_c^x\} \cup \{c_c^t,d_c^t~|~ t \in \{y,z\}\}$ in $D$. It is not difficult to see that the constructed set $D$ is indeed a TD set and since $|D| = 14|X| + 8|C|$, we conclude by \Cref{clm:tdclaw1} that $D$ is minimum. Now consider a clause $c \in C$ containing variables $x,y$ and $z$ and assume without loss of generality that $x$ is true (and thus $y$ and $z$ are false). Then $D \cup \{a_c^y\}$ contains the $P_4$ $c_c^xa_c^xa_c^yc_c^y$ and thus, $ct_{\gamma_t}(G) \leq 2$ by \Cref{theorem:1totcontracdom}; but $ct_{\gamma_t}(G) > 1$ by \Cref{clm:tdclaw1} and so, $ct_{\gamma_t}(G) = 2$. Conversely, if $ct_{\gamma_t}(G) = 2$ then $\Phi$ is satisfiable by \Cref{clm:tdclaw1}. Since $G$ is $K_{1,3}$-free, the lemma follows.
\end{proof}

\begin{lemma}
\label{lem:p6td2}
{\sc Contraction Number($\gamma_t$,2)} is $\mathsf{coNP}$-hard on $P_6$-free graphs.
\end{lemma}

\begin{proof}
If $G$ is a $P_6$-free graph then $ct_{\gamma_t}(G) \leq 2$ as shown in the proof of \Cref{lem:ectdp6kp3}; and since {\sc Contraction Number($\gamma_t$,1)} is $\mathsf{NP}$-hard on $P_6$-free graphs by \Cref{thm:dictd1}, the lemma follows from \Cref{obs:hard12}.
\end{proof}

\begin{lemma}
\label{lem:2p4td2}
{\sc Contraction Number($\gamma_t$,2)} is $\mathsf{NP}$-hard on $2P_4$-free graphs.
\end{lemma}

\begin{proof}
We use the same reduction as in \cite[Theorem 5]{GALBY202118}. More precisely, we reduce from {\sc 3-Sat}: given an instance $\Phi$ of this problem, with variable set $X$ and clause set $C$, we construct an instance $G$ of {\sc Contraction Number($\gamma_t$,2)} as follows. For every variable $x \in X$, we introduce the gadget $G_x$ consisting of a triangle on vertex set $\{x,\overline{x},u_x\}$ and an additional vertex $v_x$ adjacent to $u_x$ (that is, $G_x$ is a paw); the vertices $x$ and $\overline{x}$ are referred to as \emph{literal vertices}. For every clause $c \in C$, we introduce a \emph{clause vertex}, denoted by $c$, and add an edge between $c$ and every literal vertex whose corresponding literal appears in the clause $c$. Finally, we add an edge between every two clause vertices so that the set of clause vertices induces a clique. We let $G$ be the resulting graph. We next show that $\Phi$ is satisfiable if and only if $ct_{\gamma_t}(G) = 2$. To do so, we will rely on the following key results proved in \cite{GALBY202118}.

\begin{claim}[\cite{GALBY202118}]
\label{clm:td2p4}
The following statements are equivalent.
\begin{itemize}
\item[(i)] $\Phi$ is satisfiable.
\item[(ii)] $\gamma_t(G) = 2|X|$.
\item[(iii)] $ct_{\gamma_t}(G) > 1$.
\end{itemize}
\end{claim}

Now assume that $\Phi$ is satisfiable and consider a truth assignment satisfying $\Phi$. We construct a minimum TD set $D$ of $G$ as follows. For every variable $x \in X$, if $x$ is true then we include $\{u_x,x\}$ in $D$; otherwise, we include $\{u_x,\overline{x}\}$ in $D$. It is not difficult to see that the constructed set $D$ is indeed a TD set of $G$ and since $|D| = 2|X|$, we conclude by \Cref{clm:td2p4} that $D$ is minimum. Now consider a clause $c \in C$ containing variables $x,y$ and $z$, and assume without loss of generality that $x$ and $y$ both appear positive in $c$. Then the set $(D \setminus (V(G_x) \cup V(G_y))) \cup \{x,u_x,y,u_y,c\}$ is a TD set of $G$ of size $\gamma_t(G) + 1$ containing the $P_4$ $u_x,x,c,u_y$ and thus, $ct_{\gamma_t}(G) \leq 2$ by \Cref{theorem:1totcontracdom}; but $ct_{\gamma_t}(G)> 1$ by \Cref{clm:td2p4} and so, $ct_{\gamma_t}(G) =2$. Conversely, if $ct_{\gamma_t}(G) = 2$ then $\Phi$ is satisfiable by \Cref{clm:td2p4}. Since $G$ is readily seen to be $2P_4$-free, the lemma follows.
\end{proof}

\begin{lemma}
\label{lem:p5p2td2}
{\sc Contraction Number($\gamma_t$,2)} is $\mathsf{coNP}$-hard on $(P_5+P_2)$-free graphs.
\end{lemma}

\begin{proof}
Since $P_5+P_2 \subseteq_i P_6+P_3$, {\sc Contraction Number($\gamma_t$,3)} is polynomial-time solvable on $(P_5+P_2)$-free graphs by \Cref{lem:p6kp3td}; but {\sc Contraction Number($\gamma_t$,1)} is $\mathsf{NP}$-hard on $(P_5+P_2)$-free graphs by \Cref{thm:dictd1} and thus, the lemma follows from \Cref{obs:hard12}.
\end{proof}

\begin{lemma}
\label{lem:cyclestd}
Let $G$ be a graph such that $\gamma_t(G) \geq 3$ and let $H$ be the graph obtained by 4-subdividing every edge of $G$. Then $ct_{\gamma_t}(G) = ct_{\gamma_t}(H)$.
\end{lemma}

\begin{proof}
It was shown in \cite[Lemma 7]{GALBY202118} that $\gamma_t(H) = \gamma_t(G) + 2|E(G)|$. Let us show how to construct from a TD set $D$ of $G$ a TD set $T(D)$ of $H$ of size $|D| + 2|E(G)|$ (see \Cref{fig:GtoH}). For every edge $e= uv \in E(G)$, we will denote by $ue_1e_2e_3e_4v$ the $P_6$ in $H$ replacing the edge $e$. Firstly, we include in $T(D)$ every vertex of $D$.  Then for every edge $e=uv \in E(G)$, if $D\cap \{u,v\} = \varnothing$, we further include $\{e_2,e_3\}$ in $T(D)$; if $|D \cap \{u,v\}| =1 $, say $u \in D$ without loss of generality, we further include $\{e_3,e_4\}$ in $T(D)$; and if $u,v \in D$, we further include $\{e_1,e_4\}$ in $T(D)$. It is not difficult to see that the constructed set $T(D)$ is indeed a TD set of $H$ and that $|T(D)| = |D| + 2|E(G)|$. 

\begin{figure}
\centering
\begin{tikzpicture}[node distance=.5cm]
\node[circ,label=above:{\small $u$}] (u1) at (0,3) {};
\node[circ,label=above:{\small $v$}] (v1) at (2,3) {};
\draw[-] (u1) -- (v1) node[midway,above] {\small $e$};

\draw[-Implies,line width=.6pt,double distance=2pt] (2.75,3) -- (3.25,3); 

\node[circ,label=above:{\small $u$}] (ub1) at (4,3) {};
\node[circ,label=above:{\small $e_1$},right of=ub1] (e11) {};
\node[circ,red,label=above:{\small $e_2$},right of=e11] (e21) {};
\node[circ,red,label=above:{\small $e_3$},right of=e21] (e31) {};
\node[circ,label=above:{\small $e_4$},right of=e31] (e41) {};
\node[circ,label=above:{\small $v$},right of=e41] (vb1) {};
\draw[-] (ub1) -- (e21)
(e21) -- (e31)
(e31) -- (vb1);

\node[circ,red,label=above:{\small $u$}] (u2) at (0,2) {};
\node[circ,label=above:{\small $v$}] (v2) at (2,2) {};
\draw[-] (u2) -- (v2) node[midway,above] {\small $e$};

\draw[-Implies,line width=.6pt,double distance=2pt] (2.75,2) -- (3.25,2); 

\node[circ,red,label=above:{\small $u$}] (ub2) at (4,2) {};
\node[circ,label=above:{\small $e_1$},right of=ub2] (e12) {};
\node[circ,label=above:{\small $e_2$},right of=e12] (e22) {};
\node[circ,red,label=above:{\small $e_3$},right of=e22] (e32) {};
\node[circ,red,label=above:{\small $e_4$},right of=e32] (e42) {};
\node[circ,label=above:{\small $v$},right of=e42] (vb2) {};
\draw[-] (ub2) -- (e32)
(e32) -- (e42)
(e42) -- (vb2);

\node[circ,red,label=above:{\small $u$}] (u3) at (0,1) {};
\node[circ,red,label=above:{\small $v$}] (v3) at (2,1) {};
\draw[-] (u3) -- (v3) node[midway,above] {\small $e$};

\draw[-Implies,line width=.6pt,double distance=2pt] (2.75,1) -- (3.25,1); 

\node[circ,red,label=above:{\small $u$}] (ub3) at (4,1) {};
\node[circ,red,label=above:{\small $e_1$},right of=ub3] (e13) {};
\node[circ,label=above:{\small $e_2$},right of=e13] (e23) {};
\node[circ,label=above:{\small $e_3$},right of=e23] (e33) {};
\node[circ,red,label=above:{\small $e_4$},right of=e33] (e43) {};
\node[circ,red,label=above:{\small $v$},right of=e43] (vb3) {};
\draw[-] (ub3) -- (e13)
(e13) -- (e43)
(e43) -- (vb3);

\node[draw=none] at (1,4) {\small $G$};
\node[draw=none] at (5.25,4) {\small $H$};
\end{tikzpicture}
\caption{Constructing a total dominating set of $H$ from a total dominating set of $G$ (vertices in red belong to the corresponding total dominating set).}
\label{fig:GtoH}
\end{figure}
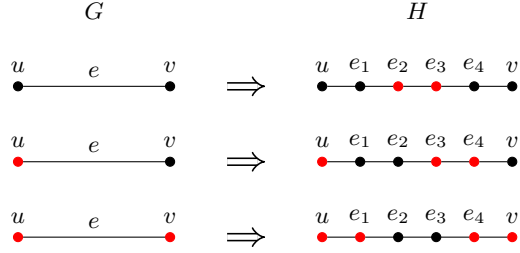

Conversely, given a TD set $D$ of $H$, let us show how to construct from $D$ a TD set $T^{-1}(D)$ of $G$ of size at most $|D| - 2|E(G)|$. To this end, we show how, given a graph $F$, a graph $F'$ obtained from $F$ by 4-subdividing one edge $uv \in E(F)$ and a TD set $D_{F'}$ of $F'$, we can construct from $D_{F'}$ a TD set $D_F$ of $F$ of size at most $|D_{F'}| - 2$. Then by iterating this procedure on $H$, we will obtain the desired TD set $T^{-1}(D)$ of $G$. 

Let $ue_1e_2e_3e_4v$ be the path in $F'$ corresponding to the 4-subdivision of the edge $uv \in E(F)$. If $e_1 \in D_{F'}$ and $u \notin D_{F'}$ then necessarily $e_2 \in D_{F'}$, for $e_1$ would otherwise not be dominated. Similarly, if $e_4 \in D_{F'}$ and $v \notin D_{F'}$ then necessarily $e_3 \in D_{F'}$, for $e_4$ would otherwise not be dominated. Thus, if $e_1,e_4 \in D_{F'}$, we let $D_F = (D_{F'} \setminus \{e_1,e_2,e_3,e_4\}) \cup \{u,v\}$. Suppose next that $e_4 \notin D_{F'}$. Then $e_2 \in D_{F'}$, for $e_3$ would otherwise not be dominated; and if $v \notin D_{F'}$ then $e_3 \in D_{F'}$, for $e_4$ would otherwise not be dominated. Thus, if $e_1 \in D_{F'}$ and $e_4 \notin D_{F'}$, then either $v \in D_{F'}$ in which case we let $D_F =D_{F'} \setminus \{e_1,e_2,e_3,e_4\}$, or $v \in D_{F'}$ in which case we let $D_F = (D_{F'} \setminus \{e_1,e_2,e_3,e_4\}) \cup \{v\}$. We proceed symmetrically if $e_1 \notin D_{F'}$ and $e_4 \in D_{F'}$. Finally, if $e_1,e_4 \notin D_{F'}$ then $e_2,e_3 \in D_{F'}$ and so, we may take $D_F = D_{F'} \setminus \{e_1,e_2,e_3,e_4\}$. Now let us observe that the TD set $T^{-1}(D)$ constructed as such satisfies the following property.

\begin{observation}
\label{obs:DFDF}
For every edge $e =uv \in E(G)$, if $e_1 \in D$ ($e_4 \in D$, respectively) then $v \in T^{-1}(D)$ ($u \in T^{-1}(D)$, respectively).
\end{observation}

We next show that $ct_{\gamma_t}(G) = ct_{\gamma_t}(H)$. As the following was shown in \cite[Lemma 7]{GALBY202118}, it in fact suffices to show that $ct_{\gamma_t}(G) = 2$ if and only if $ct_{\gamma_t}(H) = 2$. 

\begin{claim}[\cite{GALBY202118}]
\label{clm:4sub}
$ct_{\gamma_t}(G) = 1$ if and only if $ct_{\gamma_t}(H) = 1$.
\end{claim}

Assume first that $ct_{\gamma_t}(G) = 2$ and let $D$ be a TD set of $G$ of size $\gamma_t(G) + 1$ containing a $P_4$, a $2P_3$ or a $K_{1,3}$ (see \Cref{theorem:1totcontracdom}). Suppose first that $D$ contains a $P_4$ with edge set $e^1=uv$, $e^2 = vw$ and $e^3 =wt$. Then by construction, the TD set $T(D)$ of $G$ contains the $2P_3$ $e^1_4ve^2_1,e^2_4we^3_1$. Second, suppose that $D$ contains a $2P_3$ with edge set $e^1=uv$, $e^2=vw$, $e^3=xy$ and $e^4=yz$. Then by construction, the TD set $T(D)$ of $G$ contains the $2P_3$ $e^1_4ve^2_1,e^3_4ye^4_1$. Suppose finally that $D$ contains a $K_{1,3}$ on edge set $e^1 =uv$, $e^2 = wv$ and $e^3= tv$. Then by construction, the TD set $T(D)$ contains the $K_{1,3}$ $v,e^1_4,e^2_4,e^3_4$. In all three cases, we conclude by \Cref{theorem:1totcontracdom} that $ct_{\gamma_t}(H) \leq 2$ and by \Cref{clm:4sub} that in fact $ct_{\gamma_t}(H) = 2$.

Conversely assume that $ct_{\gamma_t}(H) = 2$ and denote by $\mathcal{D}$ the set of TD sets of $H$ of size $\gamma_t(H)+1$ containing a $P_4$, a $2P_3$ or a $K_{1,3}$ (note that $\mathcal{D} \neq \varnothing$ by \Cref{theorem:1totcontracdom}). Observe that if there exists $D \in \mathcal{D}$ such that $T^{-1}(D)$ contains a $P_4$, a $2P_3$ or a $K_{1,3}$, then since $ct_{\gamma_t}(G) > 1$ by \Cref{clm:4sub}, it must be that $|T^{-1}(D)| = \gamma_t(G) + 1$ and thus, $ct_{\gamma_t}(G) = 2$ by \Cref{theorem:1totcontracdom}. Now if there exists $D \in \mathcal{D}$ such that $D$ contains a $K_{1,3}$ $u,e^1_1,e^2_1,e^3_1$ where $u$ is the vertex of degree three and $e^i = ux_i$ for every $i \in [3]$, then by \Cref{obs:DFDF}, the TD set $T^{-1}(D)$ contains the $K_{1,3}$ $u,x_1,x_2,x_3$ and so, $ct_{\gamma_t}(G) = 2$ by the above.

Assume henceforth that no TD set in $\mathcal{D}$ contains a $K_{1,3}$. We contend that there then exists a TD set in $\mathcal{D}$ containing a $2P_3$. To prove this claim, let us first show that there exists $D \in \mathcal{D}$ such that for every edge $e \in E(G)$, $\{e_1,e_2,e_3,e_4\} \not\subseteq D$. In the following, given a TD set $D \in \mathcal{D}$ and an edge $e \in E(G)$, if $\{e_1,e_2,e_3,e_4\} \subseteq D$ then $D$ is said to \emph{accommodate} the edge $uv$. Consider a TD set $D \in \mathcal{D}$ accommodating the minimum number of edges amongst every TD set in $\mathcal{D}$, and suppose to the contrary that $D$ accommodates an edge $e=uv \in E(G)$. Let us show that $\mathcal{D} \setminus \{D\}$ contains a TD set accommodating fewer edges than $D$, which would contradict the choice of $D$. First note that one of $u$ and $v$ does not belong to $D$: indeed, if $\{u,v\} \subseteq D$ then $D \setminus \{e_3\}$ is a minimum TD set of $H$ containing the $P_3$ $ue_1e_2$, a contradiction to \Cref{theorem:1totcontracdom}. Furthermore, since $\gamma_t(G) \geq 3$, at least one of $u$ and $v$ has degree at least two in $G$. 

Suppose first that there exists $x \in \{u,v\} \setminus D$ such that $d_G(x) \geq 2$, say $x = v$ without loss of generality, and let $f = vw$ be an edge of $G$. Note that since $v \notin D$ by assumption, necessarily $f_2 \in D$, for $f_1$ would otherwise not be dominated. Now if $f_1 \notin D$ then $f_3 \in D$ as $f_2$ should be dominated; but then, the TD set $(D \setminus \{e_4\}) \cup \{f_1\}$ belongs to $\mathcal{D}$ (note indeed that it contains $e_1e_2e_3$ and $f_1f_2f_3$) and accommodates fewer edges than $D$, a contradiction to the choice of $D$. Thus, it must be that $f_1 \in D$; but then, the TD set $(D \setminus \{e_3\}) \cup \{v\}$ belongs to $\mathcal{D}$ (note indeed that it contains the $P_4$ $e_4vf_1f_2$) and accommodates fewer edges than $D$, a contradiction to the choice of $D$. 

Suppose second that for every $x \in \{u,v\} \setminus D$, $d_G(x) = 1$. By the above, $\{u,v\} \setminus D \neq \varnothing$ and at least one of $u$ and $v$ has degree at least two in $G$: let us assume without loss of generality that $u \notin D$ and $d_G(v) \geq 2$ (note that then $d_G(u) = 1$ and $v \in D$). Observe that since $D \setminus \{e_3\}$ is a minimum TD set of $H$, it cannot contain a $P_3$ by \Cref{theorem:1totcontracdom}. Now among the neighbours of $v$, there must be one of degree at least two, for $G$ is otherwise a star thereby contradicting the fact that $\gamma_t(G) \geq 3$. Thus, let $f= vw$ be an edge of $G$ such that $d_G(w) \geq 2$. Then since $f_1 \notin D$ ($D \setminus \{e_3\}$ would otherwise contain the $P_3$ $e_4vf_1$), necessarily $f_3 \in D$ as $f_2$ should be dominated; and since $f_3$ should be dominated but $D \setminus \{e_3\}$ cannot contain a $P_3$, necessarily $|D \cap \{f_2,f_4\}| = 1$. Now if $f_2 \in D$ then the TD set $(D \setminus \{e_4\}) \cup \{f_1\}$ belongs to $\mathcal{D}$ (note indeed that it contains the $P_4$ $vf_1f_2f_3$) and accommodates fewer edges than $D$ (recall that in this case $f_4 \notin D$), a contradiction to the choice of $D$. Now if $f_4 \in D$ then $w \notin D$, for $D \setminus \{e_3\}$ would otherwise contain the $P_3$ $f_3f_4w$; but then, by considering the TD set $(D \setminus \{e_3,e_4\}) \cup \{f_1,f_2\}$ in place of $D$, we may argue as in the previous case (recall indeed that $w \notin D$ and $d_G(w) \geq 2$) and conclude similarly to a contradiction. Therefore, $D$ accommodates no edge.

Let us next show that among those TD sets of $\mathcal{D}$ accommodating no edge, there exists a TD set $D$ such that $D$ contains a $2P_3$. Indeed, let $D \in \mathcal{D}$ be a TD set which accommodates no edge and suppose that $D$ contains no $2P_3$. Let us show how obtain from $D$ a TD set with a $2P_3$. Note that since $D$ contains no $2P_3$, $D$ contains a $P_4$ by assumption. Now suppose first that $D$ contains a $P_4$ of the form $ue_1e_2e_3$, where $e = uv \in E(G)$. Since $D$ does not accommodate the edge $uv$, $e_4 \notin D$ and so, there must exist $f = vw \in E(G)$ such that $f_1 \in D$ ($v$ would otherwise not be dominated). But then, either $v \notin D$ in which case $f_2 \in D$ ($f_1$ would otherwise not be dominated) and thus, the TD set $(D \setminus \{e_3\}) \cup \{v\}$ contains the $2P_3$ $ue_1e_2, vf_1f_2$; or $v \in D$ in which case the TD set $(D \setminus \{e_3\}) \cup \{e_4\}$ contains the $2P_3$ $ue_1e_2,e_4vf_1$. Second, suppose that $D$ contains a $P_4$ of the form $e_4vf_1f_2$, where $e = uv, f= vw \in E(G)$. Then we may assume that $f_3 \notin D$ (we revert to the previous case otherwise) which implies that $w \in D$ for otherwise, $f_4$ would not be dominated. But then, either $f_4 \in D$ in which case the TD set $(D \setminus \{f_2\}) \cup \{f_3\}$ contains the $2P_3$ $e_4vf_1,f_3f_4w$; or $f_4 \notin D$ in which case the TD set $(D \setminus \{f_2\}) \cup \{f_4\}$ contains the $2P_3$ $e_4vf_1,f_4wt$ where $t \in N_H(w) \cap D$.

Finally, let us show that among those TD sets in $\mathcal{D}$ accommodating no edge and containing a $2P_3$, there exists a TD set $D$ such that (1) no $P_3$ in $D$ is of the form $e_1e_2e_3$ where $e \in E(G)$, and (2) if $D$ contains a $P_3$ of the form $ue_1e_2$, where $e=uv \in E(G)$, then $v$ has a neighbour $f_1 \in D \setminus \{e_4\}$ where $f=vw \in E(G)$. To this end, let $D \in \mathcal{D}$ be a TD set accommodating no edge and containing a $2P_3$. Suppose that $D$ contains a $P_3$ of the form $e_1e_2e_3$ where $e = uv \in E(G)$. Since $D$ does not accommodate the edge $uv$, $e_4 \notin D$, which implies that $v$ must have a neighbour $f_1 \in D$ where $f = vw \in E(G)$. But then, either $w \in D$ in which case we may consider the TD set $(D \setminus \{e_3\}) \cup \{e_4\}$ in place of $D$; or $w \notin D$ in which case $f_2 \in D$ ($f_1$ would otherwise not be dominated) and we may consider the TD set $(D \setminus \{e_3\}) \cup \{e_4\}$ in place of $D$. By iterating this process, we eventually reach a TD set $D \in \mathcal{D}$ accommodating no edge, containing a $2P_3$ and satisfying (1). Now suppose that this TD set $D$ contains a $P_3$ of the form $ue_1e_2$, where $e =uv \in E(G)$, and suppose that $N_H(v) \cap D = \{e_4\}$. Then since $D$ does not accommodate the edge $uv$, necessarily $e_3 \notin D$ and so, $v \in D$ as $e_4$ would otherwise not be dominated. Thus, if there exists an edge $f = vw \in E(G)$, then we may consider the TD set $(D \setminus \{e_4\}) \cup \{f_1\}$ in place of $D$. Assume therefore that no such edge exists, that is, $d_G(v) = 1$. Then surely, there exists an edge $f = uw \in E(G)$ since $\gamma_t(G) \geq 3$. But then, either $f_1 \in D$ in which case we may consider the TD set $D'=(D \setminus \{e_2\}) \cup \{e_3\}$ in place of $D$ (note indeed that the path $ve_4e_3$ in $D'$ then satisfies (2) as $f_1 \in D'$); or $f_1 \notin D$ in which case we may consider the TD set $(D \setminus \{e_2\}) \cup \{f_1\}$ (note that we here simply discard the path $ue_1e_2$). By iterating this process, we eventually reach a TD set in $\mathcal{D}$ accommodating no edge, containing a $2P_3$ and satisfying both (1) and (2). 

To conclude, let $D \in \mathcal{D}$ be a TD set accommodating no edge, containing a $2P_3$ and satisfying both (1) and (2). Then by (1), $D$ contains only $P_3$s of the form $f_4ve_1$ where $f= uv, e=vw \in E(G)$, or the form $ue_1e_2$ where $e=uv \in E(G)$. Now if $D$ contains a $P_3$ of the form $f_4ve_1$ where $f = uv, e= vw \in E(G)$, then $uvw$ is a $P_3$ of $T^{-1}(D)$ by \Cref{obs:DFDF}. Similarly, if $D$ contains a $P_3$ of the form $ue_1e_2$ where $e=uv \in E(G)$, then by (2), $v$ has a neighbour $f_1 \in D \setminus \{e_4\}$ where $f = vw$, and thus, $uvw$ is a $P_3$ of $T^{-1}(D)$ by \Cref{obs:DFDF}. Since any two $P_3$s $uvw$ and $u'v'w'$ in $T^{-1}(D)$ corresponding to two distinct $P_3$s in $D$ have at most two common vertices (it may be indeed that $\{u,w\} \cap \{u',w'\} \neq \varnothing$), we conclude that $T^{-1}(D)$ contains a $2P_3$ if $\{u,w\} \cap \{u',w'\} = \varnothing$, and a $P_4$ otherwise. Therefore, $ct_{\gamma_t}(G) \leq 2$ by \Cref{theorem:1totcontracdom} and since $ct_{\gamma_t}(G) > 1$ by \Cref{clm:4sub}, in fact $ct_{\gamma_t}(G) = 2$.
\end{proof}

By 4-subdividing an instance of {\sc Contraction Number($\gamma_t,k$)} sufficiently many times, the following ensues from \Cref{lem:cyclestd}.

\begin{lemma}
\label{lem:cyclestd1}
For every $k \in \{2,3\}$ and $\ell \geq 3$, {\sc Contraction Number($\gamma_t,k$)} is $\mathsf{NP}$-hard on $\{C_3,\ldots,C_\ell\}$-free graphs.
\end{lemma}

The last result of this section concerns {\sc Contraction Number($\gamma_t$,3)}. Combining \Cref{lem:cyclestd} with the following result, we obtain in particular that if {\sc Contraction Number($\gamma_t$,3)} is polynomial-time solvable on $H$-free graphs, then $H$ must be a linear forest.

\begin{lemma}
\label{lem:clawtd3}
{\sc Contraction Number($\gamma_t$,3)} is $\mathsf{NP}$-hard on $K_{1,3}$-free graphs. 
\end{lemma}

\begin{proof}
We reduce from {\sc Positive Exactly 3-Bounded 1-In-3 3-Sat} (see \Cref{sec:prelim} for a precise definition of this problem): given an instance $\Phi$ of this problem, with variable set $X$ and clause set $C$, we construct an instance $G$ of {\sc Contraction Number($\gamma_t$,3)} as follows. For every variable $x \in X$ appearing in clauses $c,c'$ and $c''$, we introduce the gadget $G_x$ depicted in \Cref{fig:vargadclawtd3} (where a rectangle indicates that the corresponding set of vertices is a clique). In the following, we denote by $P_x$ the paw induced by $\{T_x,F_x,v_x,u_x\}$ and we may refer to the vertices of $P_x$ as $P_x(1),\ldots,P_x(4)$ where $P_x(1) = T_x$ and $P_x(2) = F_x$. For every clause $c \in C$ containing variables $x,y$ and $z$, we introduce the gadget $G_c$ which is the disjoint union of the graphs $G_c^T$ and $G_c^F$ depicted un \Cref{fig:clausegadclawtd3} (where a rectangle indicates that the corresponding set of vertices is a clique) and further add for every $\ell \in \{x,y,x\}$, an edge between $P_{\ell,T}^c(1)$ $t^c_\ell$, and an edge between $P_{x,F}^c(2)$ and $f^c_\ell$. In the following, we denote by $K_c^T$ the clique induced by $p^x_c,p^y_c$ and $p^z_c$, and by $K_c^F$ the clique induced by $q^x_c,q^y_x$ and $q^z_c$. We let $G$ be the resulting graph. We next show that $ct_{\gamma_t}(G) = 3$ if and only if $\Phi$ is satisfiable through a series of claims.

\begin{figure}
\centering
\begin{tikzpicture}
\shortPaw{0}{0}{$P_{x,T}^{c''}(1)$}{$P_{x,T}^{c''}(2)$}{}{};
\node[circ,label=below:{\small $p_x^{c''}$}] (qc) at (2,0) {};
\draw (qc) -- (1,0);

\shortPaw{0}{1.5}{$P_{x,T}^{c'}(1)$}{$P_{x,T}^{c'}(2)$}{}{};
\node[circ,label=below:{\small  $p_x^{c'}$}] (qc') at (2,1.5) {};
\draw (qc') -- (1,1.5);

\shortPaw{0}{3}{$P_{x,T}^c(1)$}{$P_{x,T}^c(2)$}{}{};
\node[circ,label=below:{\small  $p_x^c$}] (qc'') at (2,3) {};
\draw (qc'') -- (1,3);

\draw (1.7,-.75) rectangle (2.3,3.25);

\node[circ,label=below:{\small $T_x$}] (tx) at (3.5,1.5) {};
\node[circ,label=below:{\small $F_x$}] (fx) at (4.5,1.5) {};
\node[circ,label=right:{\small $v_x$}] (vx) at (4,1.9) {};
\node[circ,label=right:{\small $u_x$}] (ux) at (4,2.3) {};
\draw (tx) -- (fx) -- (vx) -- (tx)
(vx) -- (ux);

\node[circ,label=below:{\small  $q_x^{c''}$}] (pc) at (6,0) {};
\shortPaw{7}{0}{$P_{x,F}^{c''}$}{$P_{x,F}^{c''}$}{}{};
\draw (pc) -- (7,0);

\node[circ,label=below:{\small  $q_x^{c'}$}] (pc') at (6,1.5) {};
\shortPaw{7}{1.5}{$P_{x,F}^{c'}(1)$}{$P_{x,F}^{c'}(2)$}{}{};
\draw (pc') -- (7,1.5);

\node[circ,label=below:{\small  $q_x^c$}] (pc'') at (6,3) {};
\shortPaw{7}{3}{$P_{x,F}^c(1)$}{$P_{x,F}^c(2)$}{}{};
\draw (pc'') -- (7,3);

\draw (5.7,-.75) rectangle (6.3,3.25);

\draw (3.5,1.5) -- (qc)
(3.5,1.5) -- (qc')
(3.5,1.5) -- (qc'')
(4.5,1.5) -- (pc)
(4.5,1.5) -- (pc')
(4.5,1.5) -- (pc'');
\end{tikzpicture}
\caption{The variable gadget $G_x$ for a variable $x$ contained in clause $c,c'$ and $c''$ (a rectangle indicates that the corresponding set of vertices induces a clique).}
\label{fig:vargadclawtd3}
\end{figure}
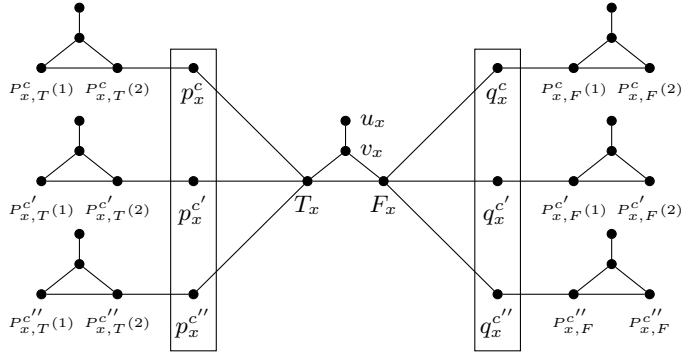

\begin{figure}
\centering
\begin{subfigure}[b]{.45\textwidth}
\centering
\begin{tikzpicture}
\node[circ,label=below:{\small $t^x_c$}] (tx) at (0,0) {};
\node[circ,label=below:{\small $p^x_c$}] (px) at (1,0) {};
\draw (tx) -- (px);

\node[circ,label=below:{\small $t^y_c$}] (ty) at (0,1) {};
\node[circ,label=below:{\small $p^y_c$}] (py) at (1,1) {};
\draw (ty) -- (py);

\node[circ,label=below:{\small $t^z_c$}] (tz) at (0,2) {};
\node[circ,label=below:{\small $p^z_c$}] (pz) at (1,2) {};
\draw (tz) -- (pz);

\draw (.75,-.55) rectangle (1.25,2.25);
\end{tikzpicture}
\caption{The graph $G_c^T$.}
\end{subfigure}
\hspace{.5cm}
\begin{subfigure}[b]{.45\textwidth}
\centering
\begin{tikzpicture}
\node[circ,label=below:{\small $u_c$}] (uc) at (0,1) {};
\node[circ,label=below:{\small $v_c$}] (vc) at (1,1) {};
\draw (uc) -- (vc);

\node[circ,label=below:{\small $q_c^x$}] (qx) at (2,0) {};
\node[circ,label=below:{\small $f_c^x$}] (fx) at (3,0) {};
\draw (qx) -- (fx);

\node[circ,label=below:{\small $q_c^y$}] (qy) at (2,1) {};
\node[circ,label=below:{\small $f_c^y$}] (fy) at (3,1) {};
\draw (qy) -- (fy);

\node[circ,label=below:{\small $q_c^z$}] (qz) at (2,2) {};
\node[circ,label=below:{\small $f_c^z$}] (fz) at (3,2) {};
\draw (qz) -- (fz);

\draw (1.75,-.55) rectangle (2.25,2.25);

\draw (vc) -- (qx)
(vc) -- (qy)
(vc) -- (qz);
\end{tikzpicture}
\caption{The graph $G_c^F$.}
\end{subfigure}
\caption{The clause gadget $G_c$ for a clause $c= x \lor y \lor z$ is the disjoint union of $G_c^T$ and $G_c^F$ (a rectangle indicates that the corresponding set of vertices is a clique).}
\label{fig:clausegadclawtd3}
\end{figure}
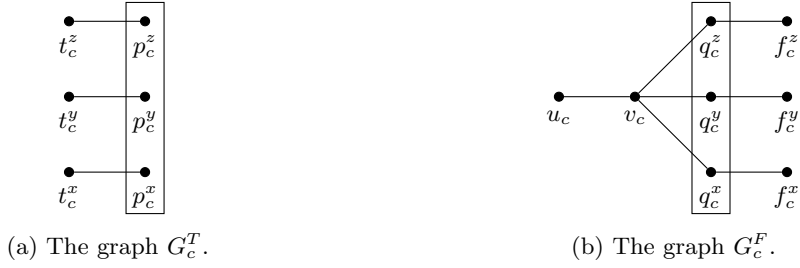

\begin{claim}
\label{clm:sizeclawtd3}
For every TD set $D$ of $G$, the following hold.
\begin{itemize}
\item[(i)] For every clause $c \in C$ and for every $R \in \{T,F\}$, $|D \cap V(G_c^R)| \geq 2$. Furthermore, $v_c \in D$. 
\item[(ii)] For every variable $x \in X$ and every paw $P$ of $G_x$, $|D \cap V(P)| \geq 2$ and $P(3) \in D$. 
\end{itemize}
\end{claim}

\begin{claimproof}
To prove (i), let $c \in C$ be a clause containing variables $x,y$ and $z$. Then since $u_c$ should be dominated, necessarily $v_c \in D$. Furthermore, $D \cap \{u_c,q_c^x,q_c^y,q_c^z\} \neq \varnothing$ as $v_c$ should be dominated. Thus, $|D \cap V(G_c^F)| \geq 2$. Now since every vertex of $K_c^T$ must be dominated, either $D \cap V(K_c^T) = \varnothing$ in which case $\{t_c^x,t_c^y,t_c^z\} \subseteq D$; or $D \cap V(K_c^T) \neq \varnothing$, say $p_c^x \in D$ without loss of generality, in which case $D \cap \{t_c^x,p_c^y,p_c^z\} \neq \varnothing$ as $p_c^x$ should be dominated as well. In both cases, $|D \cap V(G_c^T)| \geq 2$.

To prove (ii), let $x \in X$ be a variable and let $P$ be a paw in $G_x$. Then since $P(4)$ should be dominated, necessarily $P(3) \in D$; and since $P(3)$ should be dominated, necessarily $D \cap \{P(1),P(2),P(4)\} \neq \varnothing$. Thus, $|D \cap V(P)| \geq 2$.
\end{claimproof}

The following is an immediate consequence of \Cref{clm:sizeclawtd3}(ii).

\begin{observation}
\label{obs:cliquesgx}
For every TD set $D$ of $G$ and every variable $x \in X$ contained in clauses $c,c'$ and $c''$, if $|D \cap V(G_x)| = 14$ then $D \cap \{q_x^\ell,p_x^\ell~|~\ell \in \{c,c',c''\}\} = \varnothing$.
\end{observation}

The following is an immediate consequence of \Cref{clm:sizeclawtd3}(i).

\begin{observation}
\label{obs:larger3}
Let $D$ be a TD set of $G$. If $|D \cap V(G_c^F)| = 2$ for some clause $c \in C$ containing variables $x,y$ and $z$, then $D \cap \{f_c^x,f_c^y,f_c^z\} = \varnothing$.
\end{observation}

\begin{claim}
\label{clm:truthvar}
Let $D$ be a TD set of $G$ and let $x \in X$ be a variable appearing in clauses $c,c'$ and $c''$. If $|D \cap V(G_x)| = 14$ then the following hold.
\begin{itemize}
\item[(i)] If there exists $\ell \in \{c,c',c''\}$ such that $P_{x,T}^\ell(1) \in D$ then $T_x \in D$.
\item[(ii)] It there exists $\ell \in \{c,c',c''\}$ such that $P_{x,F}^\ell(2) \in D$ then $F_x \in D$.
\end{itemize}
\end{claim}

\begin{claimproof}
Assume that $|D \cap V(G_x)| = 14$. If there exists $\ell \in \{c,c',c''\}$ such that $P_{x,T}^\ell(1) \in D$, then $P_{x,T}^\ell(2) \notin D$ by \Cref{clm:sizeclawtd3}(i); and since $D \cap \{p_x^c,p_x^{c'},p_x^{c''}\} = \varnothing$ by \Cref{obs:cliquesgx}, necessarily $T_x \in D$ as $p_x^\ell$ should be dominated. Item (ii) follows by symmetry.
\end{claimproof}

\begin{claim}
\label{clm:phisatclawtd3}
$\Phi$ is satisfiable if and only if $\gamma_t(G) = 14|X| + 4|C|$. 
\end{claim}

\begin{claimproof}
Assume first that $\Phi$ is satisfiable and consider a truth assignment satisfying $\Phi$. We construct a TD set $D$ of $G$ as follows. For every variable $x \in X$ appearing in clauses $c,c'$ and $c''$, if $x$ is true then we include $\{v_x,T_x\} \cup \{P_{x,R}^\ell(1),P_{x,R}^\ell(3)~|~ \ell \in \{c,c',c''\} \text{ and } R \in \{T,F\}\}$ in $D$; otherwise, we include  $\{v_x,F_x\} \cup \{P_{x,R}^\ell(2),P_{x,R}^\ell(3)~|~ \ell \in \{c,c',c''\} \text{ and } R \in \{T,F\}\}$ in $D$. For every clause $c \in C$ containing variables $x,y$ and $z$, exactly one variable is true, say $x$ without loss of generality, in which case we include $\{v_c,q_c^x\} \cup \{p_c^y,p_c^z\}$ in $D$. It is easy to see that the constructed set $D$ is indeed a TD set of $G$ and since $|D| = 14|X| + 4|C|$, we conclude by \Cref{clm:sizeclawtd3} that $D$ is minimum.

Conversely, assume that $\gamma_t(G) = 14|X| + 4|C|$ and let $D$ be a minimum TD set of $G$. Consider a clause $c \in C$ containing variables $x,y$ and $z$. Since $|D \cap V(G_c^T)| = 2$ by \Cref{clm:sizeclawtd3}(i), at least one of $t_c^x,t_c^y$ and $t_c^z$ is not dominated by a vertex in $V(G_c^T)$, say $N(t_c^x) \cap D \subseteq V(G_x)$ without loss of generality. Then since $P_{x,T}^c(1) \in D$, it follows from \Cref{clm:truthvar}(i) that $T_x \in D$ and so, $F_x \notin D$ by \Cref{clm:sizeclawtd3}(ii). But then, $P_{x,F}^c(2) \notin D$ by \Cref{clm:truthvar}(ii) and since $f_c^x \notin D$ by \Cref{obs:larger3}, necessarily $q_c^x \in D$ as $f_c^x$ should be dominated. It then follows from \Cref{clm:sizeclawtd3}(i) that $D \cap \{q_c^y,q_c^z,f_c^y,f_c^z\} = \varnothing$, which implies that $P_{y,F}^c(2),P_{z,F}^c(2) \in D$ (one of $f_c^y$ and $f_c^z$ would otherwise not be dominated) and so, $F_y,F_z \in D$ by \Cref{clm:truthvar}(ii). Thus, the truth assignment obtained by setting a variable $x$ to true if $T_x \in D$ and to false otherwise, satisfies $\Phi$.
\end{claimproof}

\begin{claim}
\label{clm:larger}
Let $D$ be a TD set of $G$. If there exists a clause $c \in C$ containing variables $x,y$ and $z$ such that $p_c^\ell,q_c^\ell \notin D$ for some $\ell \in \{x,y,z\}$, then $|D \cap V(G_\ell)| > 14$.
\end{claim}

\begin{claimproof}
Let $c \in C$ be a clause containing variables $x,y$ and $z$, and suppose that $p_c^\ell,q_c^\ell \notin D$ for some $\ell \in \{x,y,z\}$. Then $P_{\ell,T}^c(1),P_{\ell,F}^c(2) \in D$ since $t_c^\ell$ and $f_c^\ell$ should be dominated. Thus, if $|D \cap V(G_\ell)| =14$ then by \Cref{clm:truthvar}(i) and (ii), $T_\ell,F_\ell \in D$, a contradiction to \Cref{clm:sizeclawtd3}(ii). \end{claimproof}

\begin{claim}
\label{clm:larger2}
Let $D$ be a TD set of $G$. If there exists a clause $c \in C$ such that $|D \cap V(G_c^T)| = |D \cap V(G_c^F)| = 2$ and $t_c^\ell \in D$ for some variable $\ell$ contained in $c$, then there exists a variable $v \neq \ell$ contained in $c$ such that $|D \cap V(G_v)| > 14$.
\end{claim}

\begin{claimproof}
Let $c \in C$ be a clause containing variables $x,y$ and $z$, and assume that $|D \cap V(G_c^T)| = |D \cap V(G_c^F)| = 2$ and $t_c^\ell \in D$ for some variable $\ell \in \{x,y,z\}$. Then $D \cap V(G_c^T) = \{t_c^\ell,p_c^\ell\}$: indeed, $D \cap V(K_c^T) \neq \varnothing$ for one of $p_c^x,p_c^y$ and $p_c^z$ would otherwise not be dominated; and if $p_c^v \in D$ for some $v \neq \ell$, then $p_c^v$ is not dominated. Now by \Cref{clm:sizeclawtd3}(i), $|D \cap V(K_c^F)| \leq 1$ and so, there exists a variable $v \in \{x,y,z\} \setminus \{\ell\}$ such that $q_c^v \notin D$; and since $P_c^v \notin D$ by the above, the result then follows from \Cref{clm:larger}.
\end{claimproof}

\begin{claim}
\label{clm:ct3clawtd3}
$ct_{\gamma_t}(G) = 3$ if and only if $\gamma_t(G) = 14|X| + 4|C|$.
\end{claim}

\begin{claimproof}
Assume first that $ct_{\gamma_t}(G) = 3$ and let $D$ be a minimum TD set of $G$. Let us show that $|D| = 14|X| + 4|C|$. To this end, consider first a variable $x \in X$ contained in clauses $c,c'$ and $c''$. Since $D$ contains no $P_3$ by \Cref{theorem:1totcontracdom}(i), it is clear that $|D \cap V(P)| \leq 2$ for every paw $P$ contained in $G_x$; and we conclude by \Cref{clm:sizeclawtd3} that, in fact, equality holds. Now if there exists $\ell \in \{c,c',c''\}$ such that $p_x^\ell \in D$ then $(D \setminus V(P_{x,T}^\ell)) \cup \{P_{x,T}^\ell(1),P_{x,T}^\ell(2),P_{x,T}^\ell(3)\}$ is a TD set of $G$ of size $|D| + 1 = \gamma_t(G) +1$ containing the $P_4$ $P_{x,T}^\ell(1)P_{x,T}^\ell(3)P_{x,T}^\ell(2)p_x^\ell$, a contradiction to \Cref{theorem:1totcontracdom}(ii). Thus, $D \cap \{p_x^c,p_x^{c'},p_x^{c''}\} = \varnothing$ and by symmetry, we conclude that $D \cap \{q_x^c,q_x^{c'},q_x^{c''}\} = \varnothing$ as well. Therefore, $|D \cap V(G_x)| = 14$. Second, consider a clause $c \in C$ containing variables $x,y$ and $z$. Suppose to the contrary that there exists $\ell \in \{x,y,z\}$ such that $t_c^\ell \in D$. If $P_{\ell,T}^c(1) \in D$ as well, then $D$ contains the $P_3$ $t_c^\ell P_{\ell,T}^c(1)P_{\ell,T}^c(3)$ by \Cref{clm:sizeclawtd3}(ii), a contradiction to \Cref{theorem:1totcontracdom}(i). Thus, $P_{\ell,T}^c(1) \notin D$; but then, by \Cref{clm:sizeclawtd3}(ii), $D \cup \{P_{\ell,T}^c(1)\}$ contains the $P_4$ $t_c^\ell P_{x,T}^c(1) P_{x,T}^c(3)w$ where $w \in D \cap \{P_{x,T}^c(2),P_{x,T}^c(4)\}$, a contradiction to \Cref{theorem:1totcontracdom}. Thus, $D \cap \{t_c^x,t_c^y,t_c^z\} = \varnothing$. Now if there exists $\ell \in \{x,y,z\}$ such that $f_c^\ell \in D$, then either $P_{\ell,F}^c(2) \in D$ in which case $D$ contains the $P_3$ $f_c^\ell P_{x,F}^c(2) P_{x,F}^c(3)$ by \Cref{clm:sizeclawtd3}(ii); or $P_{\ell,F}^c(2) \notin D$ in which case, by \Cref{clm:sizeclawtd3}(ii), $D \cup P_{\ell,F}^c(2)$ contains the $P_4$ $f_c^\ell P_{\ell,F}^c(2)P_{\ell,F}^c(3) w$ where $t \in D \cap \{P_{\ell,F}^c(1),P_{\ell,F}^c(4)\}$, a contradiction in both cases to \Cref{theorem:1totcontracdom}. Thus, $D \cap \{f_c^x,f_c^y,f_c^z\} = \varnothing$. Finally, since $D$ contains no $P_3$ by \Cref{theorem:1totcontracdom}(i), it is clear that $|D \cap V(K_c^T)| \leq 2$ and $|D \cap \{u_c,v_c,q_c^x,q_c^y,q_c^x\}| \leq 2$, and so, by \Cref{clm:sizeclawtd3}(i), $|D \cap V(G_c^T)| = |D \cap V(G_c^F)| = 2$. Therefore, $\gamma_t(G) = |D| = 14|X| + 4|C|$.\\

Conversely, assume that $\gamma_t(G) = 14|X| + 4|C|$ and consider a minimum TD set $D$ of $G$. Let us show that $D$ contains no $P_3$, which by \Cref{theorem:1totcontracdom}(i), would imply that $ct_{\gamma_t}(G) > 1$. Since for every clause $c \in C$, $|D \cap V(G_c^T)| = |D \cap V(G_c^F)| = 2$ by \Cref{clm:sizeclawtd3}(i), clearly $D \cap V(G_c)$ contains no $P_3$. Similarly, for every variable $x \in X$ appearing in clauses $c,c'$ and $c''$, $D \cap V(G_x)$ contains no $P_3$: indeed, by \Cref{clm:sizeclawtd3}(ii), $|D \cap V(P)| =2$ for every paw $P$ of $G_x$ and $D \cap \{q_x^\ell,p_x^\ell~|~\ell \in \{c,c',c''\}\} = \varnothing$ by \Cref{obs:cliquesgx}. Thus, if $D$ contains a $P_3$, then there must exist a clause $c \in C$ such that (1) $f_c^\ell \in D$ for some variable $\ell$ contained in $c$, or (2) $t_c^\ell \in D$ for some variable $\ell$ contained in $c$. However, by \Cref{obs:larger3}, (1) cannot hold; and by \Cref{clm:larger2}, (2) cannot hold. Thus, $D$ contains no $P_3$. 

Now suppose for a contradiction that $G$ has a TD set $D$ of size $\gamma_t(G) + 1$ containing a $P_4$, a $K_{1,3}$ or a $2P_3$ (see \Cref{theorem:1totcontracdom}(ii)). Then by \Cref{clm:sizeclawtd3}, there exists either a variable $x \in X$ such that $|D \cap V(G_x)| =15$ or a clause $c \in C$ such that $|D \cap V(G_c)| =5$. We next distinguish these two cases.\\

\noindent
\textbf{Case 1.} \emph{There exists a variable $x \in X$ such that $|D \cap V(G_x)| = 15$.} Let $c,c',c'' \in C$ be the three clauses in which $x$ is contained. First note that, since by \Cref{clm:sizeclawtd3}(ii), for every variable $v \in X \setminus \{x\}$, $|D \cap V(P)| = 2$ for every paw $P$ in $G_x$ and $q_v^\ell,p_v^\ell \notin D$ for every clause $\ell$ containing $v$, clearly $D \cap V(G_v)$ contains no $P_4$, $K_{1,3}$ or $2P_3$. Similarly, for every clause $\ell \in C$, $|D \cap V(G_\ell^T)| = |D \cap V(G_\ell^F)| = 2$ by \Cref{clm:sizeclawtd3}(i), and so, $D \cap V(G_\ell)$ contains no $P_4$, $K_{1,3}$ or $2P_3$. Since for every clause $\ell \in C \setminus \{c,c',c''\}$ and every variable $v \in X$ appearing in $\ell$, $f_\ell^v,t_\ell^v \notin D$ by \Cref{obs:larger3} and \Cref{clm:larger2}, it follows that $D \cap (V(G_\ell) \cup V(G_v))$ contains no $P_4$, $K_{1,3}$ or $2P_3$. Now suppose that there exists a clause $\ell \in \{c,c',c''\}$ containing, apart from $x$, variables $y$ and $z$, such that $t_\ell^v \in D$ for some $v \in \{y,z\}$ (note that by \Cref{clm:larger2}, $t_\ell^x \notin D$). We contend that $t_\ell^v$ cannot be part of a $P_4$, a $K_{1,3}$ or a $2P_3$. To prove this claim, let us show that $P_{v,T}^\ell(1) \notin D$ (since $|D \cap V(G_\ell^T)| = 2$, this would indeed imply our claim). Observe first that $D \cap V(G_\ell^T) = \{t_\ell^v,p_\ell^v\}$: indeed, $D \cap V(K_\ell^T) \neq \varnothing$ for one of $p_\ell^x,p_\ell^y$ and $p_\ell^z$ would otherwise not be dominated; and if $p_\ell^u \in D$ for some $u \neq v$, then $p_\ell^u$ is not dominated, a contradiction. By \Cref{clm:larger}, it must then be that $q_\ell^u \in D$ for $u \in \{y,z\} \setminus \{v\}$, as $p_\ell^u \notin D$ and $|D \cap V(G_u)| = 14$. It then follows from \Cref{clm:sizeclawtd3}(i) that $q_\ell^v \notin D$, which implies that $P_{v,F}^\ell(2) \in D$, as $f_\ell^v$ should be dominated. But then, $F_\ell \in D$ by \Cref{clm:truthvar}(ii), and since then, $T_\ell \notin D$ by \Cref{clm:sizeclawtd3}(ii), $P_{v,T}^\ell(1) \notin D$ by \Cref{clm:truthvar}(i), as claimed. Since for every $\ell \in \{c,c',c''\}$, $f_\ell^x \notin D$ by \Cref{obs:larger3} and $t_\ell^x \notin D$ by \Cref{clm:larger2}, it follows that any $P_4$, $K_{1,3}$ or $2P_3$ of $D$ is in fact contained in $D \cap V(G_x)$.

Now suppose that $|D \cap V(P)| = 3$ for some paw $P$ of $G_x$. Then by \Cref{clm:sizeclawtd3}(i), $D \cap \{p_x^\ell,q_x^\ell~|~ \ell \in \{c,c',c''\}\} = \varnothing$ and $|D \cap V(P')| =2$ for any paw $P' \neq P$ of $G_x$. Thus, $D \cap V(G_x)$ contains no $P_4$, $K_{1,3}$ or $2P_3$. Suppose next that $D \cap \{p_x^c,p_x^{c'},p_x^{c''}\} \neq \varnothing$ (the case where $D \cap \{q_x^c,q_x^{c'},q_x^{c''}\} \neq \varnothing$ is symmetric). Then by \Cref{clm:sizeclawtd3}(ii), in fact $|D  \cap \{p_x^c,p_x^{c'},p_x^{c''}\}| =1$; furthermore, $D \cap \{q_x^c,q_x^{c'},q_x^{c''}\} = \varnothing$ and $|D \cap V(P)| =2$ for every paw $P$ of $G_x$. Now assume without loss of generality that $p_x^c \in D$. If $T_x \notin D$, then it is easy to that $D$ at most one $P_3$ (namely $p_x^cP_{x,T}^c(2)P_{x,T}^c(3)$ if $P_{x,T}^c(2) \in D$). Similarly, if $P_{x,T}^c(2) \notin D$ then $D$ contains at most one $P_3$ (namely $p_x^cT_xv_x$ if $T_x \in D$). Thus, it must be that $T_x,P_{x,T}^c(2) \in D$ which by \Cref{clm:sizeclawtd3}(ii), implies that $F_x,P_{x,T}^c(1) \notin D$. Since $t_c^x$ should be dominated, it follows that $p_c^x \in D$. Furthermore, $P_{x,F}^c(2) \notin D$ by \Cref{clm:truthvar}(ii) and so, $q_c^x \in D$, as $f_c^x$ should be dominated. Since by \Cref{clm:sizeclawtd3}(i), $|D \cap V(G_c^T)| = |D \cap V(G_c^F)| = 2$ and $v_c \in D$, it follows that there exists a variable $v \neq x$ contained in $c$ such that $q_c^v,p_c^v \notin D$; but $|D \cap V(G_v)| =14$, a contradiction to \Cref{clm:larger}.\\

\noindent
\textbf{Case 2.} \emph{There exists a clause $c \in C$ such that $|D \cap V(G_c)| = 5$.} Let $x,y$ and $z$ be the variables contained in $c$. First note that, since by \Cref{clm:sizeclawtd3}(ii), for every variable $v \in X$, $|D \cap V(P)| = 2$ for every paw $P$ of $G_v$ and $q_v^\ell,p_v^\ell \notin D$ for every clause $\ell$ containing $v$, clearly $D \cap V(G_v)$ contains no $P_4$, $K_{1,3}$ or $2P_3$. Similarly, for every clause $\ell \in C \setminus \{c\}$, $|D \cap V(G_\ell^T)| = |D \cap V(G_\ell)^F)| =2 $ by \Cref{clm:sizeclawtd3}(i) and so, $D \cap V(G_\ell)$ contains no $P_4$, $K_{1,3}$ or $2P_3$. Since for every clause $\ell \in C \setminus \{c\}$ and every variable $v \in X$ appearing in $\ell$, $f_\ell^v,t_\ell^v \notin D$ by \Cref{obs:larger3} and \Cref{clm:larger2}, it follows that $D \setminus (V(G_c) \cup V(G_x) \cup V(G_y) \cup V(G_z)$ contains no $P_4$, $K_{1,3}$ or $2P_3$. 

Suppose first that $|D \cap V(G_c^F)| = 3$ (note that then $|D \cap V(G_c^T)| =2$). Suppose further that $D \cap \{t_c^x,t_c^y,t_c^z\} \neq \varnothing$, say $t_c^x \in D$ without loss of generality. We contend that $t_c^x$ cannot be part of a $P_4$, $K_{1,3}$ or $2P_3$ in $D$. To prove this claim, let us show that $P_{x,T}^(1) \notin D$ (since $|D \cap V(G_c^T)| =2$, this would indeed prove our claim). Observe first that $D \cap V(G_c^T) = \{p_c^x,t_c^x\}$, for if $p_c^x \notin D$, then one of $p_c^y$ and $p_c^z$ is not dominated. Since $|D \cap V(K_c^F)| \leq 2$ by \Cref{clm:sizeclawtd3}(i), it follows that either (1) $q_c^v \notin D$ for some $v \in \{y,z\}$, or (2) $q_c^x \notin D$. If (1) holds then $p_c^v,q_c^v \notin D$ by the above; but $|D \cap V(G_v)| = 14$, a contradiction to \Cref{clm:larger}. Thus, (2) holds; in particular, $D \cap V(K_c^F) = \{q_c^y,q_c^z\}$. Since $v_c \in D$ by \Cref{clm:sizeclawtd3}(i), it follows that $P_{x,F}^c(2) \in D$ as $f_c^x$ should be dominated. Then by \Cref{clm:truthvar}(ii), $F_x \in D$ which, by \Cref{clm:sizeclawtd3}(ii), implies that $T_x \notin D$; and so, by \Cref{clm:truthvar}(i), $P_{x,T}^c(1) \notin D$, as claimed. It follows that any $P_4$, $K_{1,3}$ or $2P_3$ in $D$ is in fact contained in $V(G_c^F) \cup V(G_x) \cup V(G_y) \cup V(G_z)$. Since for any $\ell \in \{x,y,z\}$, $D \cap V(G_\ell)$ contains no $P_3$ by \Cref{clm:sizeclawtd3}(ii) and \Cref{obs:cliquesgx}, it follows that $D \cap \{f_c^x,f_c^y,f_c^z\} \neq \varnothing$; and since $v_c \in D$ by \Cref{clm:sizeclawtd3}(i) and $v_c$ should be dominated, in fact $|D \cap \{f_c^x,f_c^y,f_c^z\}| = 1$. Assume without loss of generality that $f_c^x \in D$. Then $q_c^x \in D$: indeed, if $q_c^x \notin D$ then $D \cap V(G_c^F)$ contains no $P_3$; and since for every $\ell \in \{x,y,z\}$, $D \cap V(G_\ell)$ contains no $P_3$ by \Cref{clm:sizeclawtd3}(ii) and \Cref{obs:cliquesgx}, clearly $D \cap (V(G_c^F) \cup V(G_x) \cup V(G_y) \cup V(G_z))$ then contains no $P_4$, $K_{1,3}$ or $2P_3$. Similarly, $P_{x,F}^c(2) \in D$: indeed, if $P_{x,F}^c(2) \notin D$ then since for any $\ell \in \{x,y,z\}$, $D \cap V(G_\ell)$ contains no $P_3$, clearly $D \cap (V(G_c^F) \cup V(G_x) \cup V(G_y) \cup V(G_z))$ then contains no $P_4$, $K_{1,3}$ or $2P_3$. It then follows from \Cref{clm:truthvar}(ii) that $F_x \notin D$ which, by \Cref{clm:sizeclawtd3}(ii), implies that $T_x \notin D$ and so, $P_{x,T}^c(1) \notin D$ by \Cref{clm:truthvar}(i). Since $t_c^x$ should be dominated, it follows that $p_c^x \in D$ which implies that $|D \cap \{p_c^y,p_c^z\}| \leq 1$, as $|D \cap V(G_c^T)| =2$. Since $D \cap \{q_c^y,q_c^z\} = \varnothing$ by the above, there then exists $v \in \{y,z\}$ such that $p_c^v,q_c^v \notin D$; but $|D \cap V(G_v)| =14$, a contradiction to \Cref{clm:larger}. Thus, $|D \cap V(G_c^F)| <3$.

Second, suppose that $|D \cap V(G_c^T)| = 3$. Since then, $|D \cap V(G_c^F)| =2$, it follows from \Cref{obs:larger3} that $D \cap \{f_c^x,f_c^y,f_c^z\} = \varnothing$ and so, any $P_4$, $K_{1,3}$ or $2P_3$ in $D$ is in fact contained in $D \cap (V(G_c^T) \cup V(G_x) \cup V(G_y) \cup V(G_z)$. But then $D \cap \{t_c^x,t_c^y,t_c^z\} \neq \varnothing$: indeed, if $D \cap \{t_c^x,t_c^y,t_c^z\} = \varnothing$, since for any $\ell \in \{x,y,z\}$, $D \cap V(G_\ell)$ contains no $P_3$ by \Cref{clm:sizeclawtd3}(ii) and \Cref{obs:cliquesgx}, clearly $D \cap (V(G_c^T) \cup V(G_x) \cup V(G_y) \cup V(G_z)$ then contains no $P_4$, $K_{1,3}$ or $2P_3$. Assume without loss of generality that $t_c^x \in D$. Then $p_c^x \in D$: indeed, if $p_c^x \notin D$ then $D \cap V(G_c^T)$ contains no $P_3$ and since for any $\ell \in\{x,y,z\}$, $D \cap V(G_\ell)$ contains no $P_3$, clearly $D \cap (V(G_c^T) \cup V(G_x) \cup V(G_y) \cup V(G_z)$ contains no $P_4$, $K_{1,3}$ or $2P_3$. Similarly, $P_{x,T}^c(1) \in D$: if not, then since for any $\ell \in \{x,y,z\}$, $D \cap V(G_\ell)$ contains no $P_3$, clearly $D \cap (V(G_c^T) \cup V(G_x) \cup V(G_y) \cup V(G_z)$ contains no $P_4$, $K_{1,3}$ or $2P_3$. It then follows from \Cref{clm:truthvar}(i) that $T_x \in D$ which, by \Cref{clm:sizeclawtd3}(ii), implies that $F_x \notin D$ and so, $P_{x,F}^c(2) \notin D$ by \Cref{clm:truthvar}(ii). Since $f_c^x$ should be dominated, necessarily $q_c^x \in D$ and so, $D \cap \{q_c^y,q_c^z\} = \varnothing$ by \Cref{clm:sizeclawtd3}(i). Since $|D \cap \{p_c^y,p_c^z\}| \leq 1$ by the above, it follows that there exists $v \in \{y,z\}$ such that $p_c^v,q_c^v \notin D$; but $|D \cap V(G_v)| =14$, a contradiction to \Cref{clm:larger} which concludes the proof.
\end{claimproof}

Now by Claims \ref{clm:phisatclawtd3} and \ref{clm:ct3clawtd3}, $\Phi$ is satisfiable if and only if $ct_{\gamma_t}(G) =3$; and since $G$ is easily seen to be $K_{1,3}$-free, the lemma follows.
\end{proof}

%------------------------------------------------------------------------------------------------------------------------------------------------------------------------------------

\subsection{Proof of \Cref{thm:dictd2}}
\label{sec:dictd2}

Let $H$ be a graph. If $H$ contains a cycle then {\sc Contraction Number($\gamma_t$,2)} is $\mathsf{NP}$-hard by \Cref{lem:cyclestd1}. Assume henceforth that $H$ is a forest. If $H$ contains a vertex of degree at least three then {\sc Contraction Number($\gamma_t$,2)} is $\mathsf{NP}$-hard on $H$-free graphs by \Cref{lem:clawtd2}. Suppose therefore that $H$ is a linear forest. If $H$ contains a connected component on at least six vertices then {\sc Contraction Number($\gamma_t$,2)} is $\mathsf{coNP}$-hard on $H$-free graphs by \Cref{lem:p6td2}. Thus we may assume that every connected component of $H$ has size at most five. If $H$ contains at least two connected component on at least four vertices then {\sc Contraction Number($\gamma_t$,2)} is $\mathsf{NP}$-hard on $H$-free graphs by \Cref{lem:2p4td2}. Assume therefore that $H$ contains at most one connected component of size at least four and at most five. Now if $H$ contains a $P_5$ and at least one other connected component on at least two vertices then {\sc Contraction Number($\gamma_t$,2)} is $\mathsf{coNP}$-hard on $H$-free graphs by \Cref{lem:p5p2td2}; and if $H$ contains a $P_5$ and every other connected component has size one then {\sc Contraction Number($\gamma_t$,2)} is polynomial-time solvable on $H$-free graphs by \Cref{lem:eastd2}. Otherwise $H$ contains at most one connected component of size at most four while every other connected component has size at most three in which case {\sc Contraction Number($\gamma_t$,2)} is polynomial-time solvable on $H$-free graphs by \Cref{lem:easytd2}.

%------------------------------------------------------------------------------------------------------------------------------------------------------------------------------------

\section{Semitotal Domination}
\label{sec:STD}

In this section, we consider the {\sc Contraction Number($\gamma_{t2}$,$k$)} problem for $k =2,3$. In \Cref{sec:stdeasy}, we cover those polynomial-time solvable cases and in \Cref{sec:stdhard}, we examine hard cases. The proof of \Cref{thm:dicstd2} can be found in \Cref{sec:dicstd2}. 

%------------------------------------------------------------------------------------------------------------------------------------------------------------------------------------

\subsection{Polynomial-time solvable cases}
\label{sec:stdeasy}

The algorithms for {\sc Contraction Number($\gamma_{t2}$,2)} and {\sc Contraction Number($\gamma_{t2}$,3)} outlined thereafter will rely on the following key result.

\begin{lemma}
\label{lem:hp3std}
Let $H$ be a graph. If {\sc 2-Edge Contraction($\gamma_{t2}$)} is polynomial-time solvable on $H$-free graphs then it is polynomial-time solvable on $(H+P_3)$-free graphs.
\end{lemma}

\begin{proof}
Let $G$ be an $(H+P_3)$-free graph. We aim to show that if $G$ contains an induced $H$ and $G$ is a \no-instance for {\sc 2-Edge contraction($\gamma_{t2}$)}, then $\gamma_{t2}(G)$ is bounded by some function $f$ of $|V(H)|$ (and $|V(H)|$ only). Assuming for now that this claim is correct, the following algorithm solves the {\sc 2-Edge Contraction($\gamma_{t2}$)} problem on $(H+P_3)$-free graphs.\\

\begin{itemize}
\item[1.] If $G$ contains no induced $H$ then use the algorithm for $H$-free graphs to determine whether $G$ is a \yes-instance for {\sc 2-Edge Contraction($\gamma_{t2}$)} or not.
\item[2.] Check whether there exists an SD set of $G$ of size at most $f(|V(H)|)$.
\begin{itemize}
\item[2.1] If the answer is no then output \yes.
\item[2.2] Check whether there exists a minimum SD set of $G$ containing a friendly triple, or an SD set of $G$ of size $\gamma_{t2}(G)+1$ containing an ST-configuration (see \Cref{thm:friendlytriple}) and output the answer accordingly.
\end{itemize}
\end{itemize}

\bigskip 

Now observe that checking whether $G$ is $H$-free can be done in time $|V(G)|^{O(|V(H)|)}$ and that step 2 can be done in time $|V(G)|^{O(f(|V(H)|)}$ by simple brute force. Since {\sc 2-Edge Contraction(\allowbreak $\gamma_{t2}$)} is polynomial-time solvable on $H$-free graphs by assumption, the above algorithm indeed runs in polynomial time (for fixed $H$). We next show that $f(|V(H)|) = 8|V(H)|$.\\

Assume henceforth that $G$ contains an induced $H$. Let $A \subseteq V(G)$ be a set of vertices such that $G[A]$ is isomorphic to $H$, let $B \subseteq V(G) \setminus A$ be the set of vertices at distance one from $A$ and let $C = V(G) \setminus (A \cup B)$. Note that since $G$ is $(H+P_3)$-free, $C$ is a disjoint union of cliques. In the following, we denote by $\mathcal{K}$ the set of maximal cliques in $C$. 

Now observe that if $G$ is a \no-instance for {\sc 2-Edge Contraction($\gamma_{t2}$)}, then no minimum SD set of $G$ contains an edge by \Cref{lem:edge} and thus, given a minimum SD set $D$ of $G$, $\mathcal{K}$ can be partitioned into two sets: $\mathcal{K}^D_1 = \{K \in \mathcal{K}~|~|D \cap V(K)| =1\}$ and $\mathcal{K}^D_0 = \{K \in \mathcal{K}~|~D \cap V(K) = \varnothing\}$.

\begin{claim}
\label{clm:KD1}
If $G$ is a \no-instance for {\sc 2-Edge Contraction($\gamma_{t2}$)} and $D$ is a minimum SD set of $G$ then $$|(D \cap N[\mathcal{K}^D_1]) \cup w_D(D \cap N[\mathcal{K}^D_1])| \leq 4|A|.$$
\end{claim}

\begin{claimproof}
Assume that $G$ is a \no-instance for {\sc 2-Edge Contraction($\gamma_{t2}$)} and let $D$ be a minimum SD set of $G$. For every $K \in \mathcal{K}^D_1$, denote by $x_K \in V(K) \cap D$ the unique vertex in $K$ belonging to $D$. Observe that for every $K \in \mathcal{K}^D_1$, $N(x_K) \cap B \neq \varnothing$: indeed, if $N(x_K) \cap B = \varnothing$ for some $K \in \mathcal{K}^D_1$, then there must exist $y \in V(K)$ such that $y$ is adjacent to the unique witness $z$ for $x_K$ (recall that by \Cref{lem:edge}, $|w_D(x_K)| =1$ and $x_K$ is at distance exactly two from its witness); but then $(D \setminus \{x_K\}) \cup \{y\}$ is a minimum SD set of $G$ containing the edge $yz$, a contradiction to \Cref{lem:edge}. We now contend that there exist no two cliques $K,K' \in \mathcal{K}^D_1$ such that $d(x_K,x_{K'}) = 2$. Indeed, suppose to the contrary that there exist two such cliques $K,K' \in \mathcal{K}^D_1$. Let $z \in B$ be a common neighbour of $x_K$ and $x_{K'}$, and let $t \in A$ be a neighbour of $z$. Then $z \notin D$ ($x_K,z,x_{K'}$ would otherwise be a friendly triple) which implies, in particular, that $t$ is not dominated by $z$. Since $t$ should be dominated nonetheless, it follows that either $t \in D$ in which case $D \cup \{z\}$ contains the $O_5$ $z,x_K,x_{K'},t$; or $t$ has a neighbour $u \in D$ (note that $u \neq x_K,x_{K'}$ by construction) in which case $D \cup \{z\}$ contains the $O_6$ $x_K,z,x_{K'},u$, a contradiction in both cases to \Cref{thm:friendlytriple}. 

It follows that for any clique $K \in \mathcal{K}^D_1$, any witness for $x_K$ can only belong to $A \cup B$: let us denote by $W_A = A \cap w_D(\{x_K~|~K \in \mathcal{K}^D_1\})$ and by $W_B = B \cap w_D(\{x_K~|~K \in \mathcal{K}^D_1\})$. Further note that since by \Cref{lem:edge}, $|w_D(x)| = 1$ for any $x \in D$, $$|W_B| + |W_A| = |w_D(\{x_K~|~K \in \mathcal{K}^D_1\})| = |\mathcal{K}^D_1|.$$ Now if there exist two vertices $u,v \in W_B$ such that $u$ and $v$ have a common neighbour in $A$, then $d(u,v) \leq 2$, a contradiction to \Cref{lem:edge} as, by construction, $u$ and $v$ both have a witness in $\{x_K~|~K \in \mathcal{K}^D_1\}$. Thus, no two vertices in $W_B$ have a common neighbour in $A$, which implies that $|W_B| \leq |A|$. Since $|W_A| \leq |A|$, it follows that $|\{w_K~|~K \in \mathcal{K}^D_1\}| \leq 2|A|$. Now note that $|(D \cap N[\mathcal{K}^D_1]) \cup w_D(D \cap N[\mathcal{K}^D_1])| = |\mathcal{K}^D_1| + |w_D(\{x_K~|~ K\in \mathcal{K}^D_1\})|$: indeed, for every $K \in \mathcal{K}^D_1$, any vertex in $D \cap N[K]$ is at distance at most two from $x_K$, and by \Cref{lem:edge}, $|w_D(x)| = 1$ for every $x \in D$. Thus, the lemma follows.
\end{claimproof}

\begin{claim}
\label{clm:KD0}
If $G$ is a \no-instance for {\sc 2-Edge Contraction($\gamma_{t2}$)} and $D$ is a minimum SD set of $G$, then $$|(D \cap N[\mathcal{K}^D_0]) \cup w_D(D \cap N[\mathcal{K}^D_0])| \leq 2|A|.$$
\end{claim}

\begin{claimproof}
Assume that $G$ is a \no-instance for {\sc 2-Edge Contraction($\gamma_{t2}$)} and let $D$ be a minimum SD set of $G$. Since for every $K \in \mathcal{K}^D_0$, $D \cap V(K) = \varnothing$ by definition, necessarily $D \cap N[\mathcal{K}^D_0] = D \cap N(\mathcal{K}^D_0) \subseteq B$; in particular, every vertex in $D \cap N(\mathcal{K}^D_0)$ has at least one neighbour in $A$. Since no two vertices $x,y \in D \cap N(\mathcal{K}^D_0)$ such that $y \notin w_D(x)$, have a common neighbour in $A$, and for every $x \in D$, $|w_D(x)| = 1$ by \Cref{lem:edge}, denoting by 
$$ D_1 = \{x~|~ x \in D \cap N(\mathcal{K}^D_0) \text{ and } w_D(x) \cap (D \cap N(\mathcal{K}^D_0)) = \varnothing\}$$
and by 
$$ D_2 = \{\{x,y\}~|~ x,y \in D \cap N(\mathcal{K}^D_0) \text{ and } y \in w_D(x)\}$$
it follows that $|D_1| + |D_2| \leq |A|$. But 
$$|(D \cap N(\mathcal{K}^D_0)) \cup w_D(D \cap N(\mathcal{K}^D_0))| = 2|D_1| + 2|D_2|$$
and thus, the upperbound follows.
\end{claimproof}

To conclude, assume that $G$ is a \no-instance for {\sc 2-Edge Contraction($\gamma_{t2}$)} and let $D$ be a minimum SD set of $G$. Then $D' = D \setminus ((D \cap N[\mathcal{K}]) \cup (w_D(D \cap N[\mathcal{K}])))$ has size at most $2|A|$: indeed, since $D'$ dominates only vertices in $A \cup B$ and no vertex in $D \setminus D'$ is witnessed by a vertex $D'$ by \Cref{lem:edge}, if $|D'| > 2|A|$ then $(D \setminus D') \cup (A \cup N_A)$ where $N_A$ contains one neighbour for each vertex in $A$, is an SD set of $G$ of size strictly smaller than that of $D$, a contradiction to the minimality of $D$. Since 
$$|(D \cap N[\mathcal{K}]) \cup (w_D(D \cap N[\mathcal{K}]))| \leq \sum_{i \in \{0,1\}}  |(D \cap N[\mathcal{K}^D_i]) \cup (w_D(D \cap N[\mathcal{K}^D_i]))| \leq 6|A|$$
by Claims~\ref{clm:KD1} and \ref{clm:KD0}, we conclude that $\gamma_{t2}(G) = |D| \leq 8|A| = 8|V(H)|$, as claimed.
\end{proof}

\begin{lemma}
\label{lem:p8std2}
If $G$ is a $P_8$-free graph then $ct_{\gamma_{t2}}(G) \leq 2$.
\end{lemma}

\begin{proof}
Let $G$ be a $P_8$-free graph and let $D$ be a minimum SD set of $G$. If $D$ contains an edge or there exists $x \in D$ such that $|w_D(x)| \geq 2,$ then $ct_{\gamma_{t2}}(G) \leq 2$ by \Cref{lem:edge}. Thus, assume that $D$ is an independent set and for every $x \in D$, $|w_D(x)| = 1$. Let $u,v \in D$ be two vertices at distance two and further let $w \in D\setminus \{u,v\}$ be a closest vertex from $\{u,v\}$, that is, $d(w,\{u,v\}) = \min_{x \in D \setminus \{u,v\}} d(x,\{u,v\})$. Then $d(w,\{u,v\}) > 2$ by assumption, and since $d(w,\{u,v\}) \leq 3$ as shown in the proof of \Cref{lem:edge}, in fact $d(w,\{u,v\}) = 3$. Now assume without loss of generality that $d(w,\{u,v\}) = d(w,v)$, and let $t \in D$ be the witness for $w$. Note that if there exists $a \in N(u) \cap N(v)$ and $b \in N(w)$ such that $at \in E(G)$, then $D \cup \{a\}$ contains the $O_6$ $u,a,v,w$ and thus, $ct_{\gamma_{t2}}(G) \leq 2$. Assume henceforth that no common neighbour of $u$ and $v$ is adjacent to a neighbour of $w$ (or $t$), and that, symmetrically, no common neighbour of $w$ and $t$ is adjacent to a neighbour of $u$ (or $v$). Let $P=vxyw$ be a shortest path from $v$ to $w$. Then by the aforementioned assumption, $u$ is nonadjacent to $x$, and $t$ is nonadjacent to $y$ (also note that $u$ is nonadjacent to $y$, and $t$ is nonadjacent to $x$, by minimality of the distance between $v$ and $w$). Let $a$ be a common neighbour of $u$ and $v$, and let $b$ be a common neighbour of $w$ and $t$. Then both $a$ and $x$ are nonadjacent to $b$ by assumption; and similarly, $y$ is nonadjacent to $a$. It follows that $ax \in E(G)$ or $by \in E(G)$ for otherwise, $uavxywbt$ would induce a $P_8$; we next distinguish two cases depending on these adjacencies.

Suppose first that $ax,by \in E(G)$. Observe that if every private neighbour of $t$ is adjacent to $b$ or $y$, then $(D \setminus \{t\}) \cup \{b,y\}$ is an SD set of $G$ containing the $O_4$ $b,w,y,v$ and thus, $ct_{\gamma_{t2}}(G) \leq 2$ by \Cref{lem:edge}. Assume therefore that $t$ has a private neighbour $c$ which is nonadjacent to both $b$ and $y$, and that, symmetrically, $u$ has a private neighbour $d$ which is nonadjacent to both $a$ and $x$. Let us show that then $G$ contains an induced $P_8$. Indeed, since by assumption $d$ is nonadjacent to $b$, and $c$ is nonadjacent to $a$, it must be that $dc \in E(G)$, $dy \in E(G)$ or $cx \in E(G)$ for otherwise, $dyaxybtc$ would induce a $P_8$. However, if $dc \in E(G)$ then $vaudctbw$ induces a $P_8$. Thus, $dy \in E(G)$ or $cx \in E(G)$, say the latter holds without loss of generality, in which case $wbtcxaud$ induces a $P_8$. 
 
Second, suppose that either $a$ is nonadjacent to $x$, or $b$ is nonadjacent to $y$, say $ax \in E(G)$ and $by \notin E(G)$ (the other case is symmetric). As previously, we show that if $t$ has a private neighbour $c$ which is nonadjacent to both $b$ and $y$, and that $u$ has a private neighbour $d$ which is nonadjacent to both $a$ and $x$, then $G$ contains an induced $P_8$ (observe that if this does not hold, then $ct_{\gamma_{t2}}(G) \leq 2$ as mentioned above). Since $c$ is nonadjacent to $a$ by assumption, $c$ must be adjacent to $x$ for otherwise, $ctbwyxau$ would induce a $P_8$. But then, $c$ must be adjacent to $d$ for otherwise, $wbtcxaud$ induces a $P_8$; however, if $dc \in E(G)$ then $wbtcduav$ induces a $P_8$, a contradiction which concludes the proof.
\end{proof}

\begin{lemma}
\label{lem:2p4std2}
If $G$ is a $2P_4$-free graph then $ct_{\gamma_{t2}}(G) \leq 2$.
\end{lemma}

\begin{proof}
The proof is similar to that of \Cref{lem:p8std2}. Let $G$ be a $2P_4$-free graph and let $D$ be a minimum SD set of $G$. If $D$ contains an edge or there exists $x \in D$ such that $|w_D(x)| \geq 2$, then $ct_{\gamma_{t2}}(G) \leq 2$ by \Cref{lem:edge}. Thus, assume that $D$ is an independent set and for every $x \in D$, $|w_D(x)| = 1$. Let $u,v \in D$ be two vertices at distance two and further let $w \in D\setminus \{u,v\}$ be a closest vertex from $\{u,v\}$, that is, $d(w,\{u,v\}) = \min_{x \in D \setminus \{u,v\}} d(x,\{u,v\})$. Then $d(w,\{u,v\}) > 2$ by assumption, and since $d(w,\{u,v\}) \leq 3$ as shown in the proof of \Cref{lem:edge}, in fact $d(w,\{u,v\}) = 3$. Now assume without loss of generality that $d(w,\{u,v\}) = d(w,v)$, and let $t \in D$ be the witness for $w$. Note that if there exists $a \in N(u) \cap N(v)$ and $b \in N(w)$ such that $at \in E(G)$, then $D \cup \{a\}$ contains the $O_6$ $u,a,v,w$ and thus, $ct_{\gamma_{t2}}(G) \leq 2$ by \Cref{thm:friendlytriple}. Assume henceforth that no common neighbour of $u$ and $v$ is adjacent to a neighbour of $w$ (or $t$), and that, symmetrically, no common neighbour of $w$ and $t$ is adjacent to a neighbour of $u$ (or $v$). Let $P=vxyw$ be a shortest path from $v$ to $w$. Then by the aforementioned assumption, $u$ is nonadjacent to $x$, and $t$ is nonadjacent to $y$ (also note that $u$ is nonadjacent to $y$, and $t$ is nonadjacent to $x$, by minimality of the distance between $v$ and $w$). Let $a$ be a common neighbour of $u$ and $v$, and let $b$ be a common neighbour of $w$ and $t$. Then $a$ and $x$ are nonadjacent to $b$ by assumption; and similarly, $y$ is nonadjacent to $a$. Now if $b$ is nonadjacent to $y$, then every private neighbour $p$ of $u$ must be adjacent to $a$ or $y$ for otherwise, $puav$ and $ywbt$ would induce a $2P_4$ (note that $p$ is nonadjacent to $b$ by assumption); but then, $(D \setminus \{u\}) \cup \{a,y\}$ is an SD set of $G$ containing the $O_7$ $a,v,y,w$ and thus, $ct_{\gamma_{t2}}(G) \leq 2$ by \Cref{thm:friendlytriple}. Assume therefore that $b$ is adjacent to $y$. If every private neighbour of $u$ is adjacent to $a$, then $(D \setminus \{u\}) \cup \{a\}$ is a minimum SD set of $G$ containing an edge and thus, $ct_{\gamma_{t2}}(G) \leq 2$ by \Cref{lem:edge}. Thus, we may assume that $u$ has a private neighbour $p$ nonadjacent to $a$. Then every private neighbour $p'$ of $t$ must be adjacent to $b$ or $p$ for otherwise, $puav$ and $wbtp'$ would induce a $2P_4$ (note that $p$ is nonadjacent to $b$, and $p'$ is nonadjacent to $a$ by assumption). By symmetry, we conclude that every private neighbour of $w$ is adjacent to $p$ or $b$. Now if, in fact, no private neighbour of $w$ is adjacent to $p$, then $D' = (D \setminus \{w\}) \cup \{b\}$ is a minimum SD set of $G$, as every private neighbour of $w$ w.r.t. $D$ is then adjacent to $b$; but $D'$ contains an edge and thus, $ct_{\gamma_{t2}}(G) \leq 2$ by \Cref{lem:edge}. Thus, we may assume that at least one private neighbour of $w$ is adjacent to $p$; in particular, $d(w,p) = 2$. But then, $(D \setminus \{t\}) \cup \{p,b\}$ is an SD of $G$ containing the $O_7$ $b,w,p,u$ and thus, $ct_{\gamma_{t2}}(G) \leq 2$ by \Cref{thm:friendlytriple}.
\end{proof}

Since for any graph $G$, $G$ is a \yes-instance for {\sc 2-Edge Contraction($\gamma_{t2}$)} if and only if $G$ is a \no-instance for {\sc Contraction Number($\gamma_{t2}$,3)}, the following ensues from Lemmas~\ref{lem:hp3std}, \ref{lem:p8std2} and \ref{lem:2p4std2}.

\begin{corollary}
\label{lem:easystd3}
{\sc Contraction Number($\gamma_{t2}$,3)} is polynomial-time solvable on $H$-free graphs if $H \subseteq_i P_8+kP_3$ for some $k \geq 0$, or $H \subseteq 2P_4+kP_3$ for some $k \geq 0$.
\end{corollary}

\begin{lemma}
\label{lem:easystd2}
{\sc Contraction Number($\gamma_{t2}$,2)} is polynomial-time solvable on $H$-free graphs if $H \subseteq_i P_5+tK_1$ for some $t \geq 0$, or $H \subseteq_i P_3+tP_2$ for some $t \geq 0$.
\end{lemma}

\begin{proof}
Assume that $H \subseteq_i P_3+tP_2$ for some $t \geq 0$ (the case where $H \subseteq_i P_5+tK_1$ for some $t \geq 0$ is handled similarly) and let $G$ be an $H$-free graph. Since $H$ is a fortiori an induced subgraph of $P_8+tP_3$, we may use the polynomial-time algorithm for $(P_8+tP_3)$-free graphs given by \Cref{lem:easystd3} to determine whether $G$ is a \yes-instance for {\sc Contraction Number($\gamma_{t2}$,3)} or not. If the answer is yes then we output \no; otherwise, we use the polynomial-time algorithm for $(P_3+tP_2)$-free graphs given by \Cref{thm:dicstd1} to determine whether $G$ is a \yes-instance for {\sc Contraction Number($\gamma_{t2}$,1)} or not, and output the answer accordingly.
\end{proof}

%------------------------------------------------------------------------------------------------------------------------------------------------------------------------------------

\subsection{Hardness results}
\label{sec:stdhard}

In this section, we show that {\sc Contraction Number($\gamma_{t2}$,2)} and {\sc Contraction Number($\gamma_{t2}$,3)} are $\mathsf{NP}$-hard on a number of monogenic graph classes. We first consider the case $k=2$. 

Firstly, since for any graph $G$ such that $ct_{\gamma_{t2}}(G) \leq 2$, $G$ is a \yes-instance for {\sc Contraction Number($\gamma_{t2}$,1)} if and only if $G$ is a \no-instance for {\sc Contraction Number($\gamma_{t2}$,2)}, the following ensues from \Cref{thm:dicstd1} and \Cref{lem:p8std2}.

\begin{lemma}
\label{lem:linforeststd2}
{\sc Contraction Number($\gamma_{t2}$,2)} is $\mathsf{(co)NP}$-hard on $P_6$-free graphs, $2P_3$-free graphs and $(P_4+P_2)$-free graphs.
\end{lemma}

\begin{lemma}
\label{lem:clawstd2}
{\sc Contraction Number($\gamma_{t2}$,2)} is $\mathsf{NP}$-hard in $K_{1,3}$-free graphs.
\end{lemma}

\begin{proof}
We use the same construction as in \cite[Theorem 7]{semitotcon}. More precisely, we reduce from {\sc Positive Exactly 3-Bounded 1-In-3 3-Sat} (see \Cref{sec:prelim} for a precise definition of this problem): given an instance $\Phi$ of this problem, with variable set $X$ and clause set $C$, we construct an instance $G$ of {\sc Contraction Number($\gamma_{t2}$,2)} as follows. For every variable $x \in X$ contained in clauses $c,c'$ and $c''$, we introduce the gadget $G_x$ depicted in \Cref{fig:stdclawvargad} (where a rectangle indicates that the corresponding set of vertices is a clique). For every clause $c \in C$ containing variables $x,y$ and $z$, we introduce the gadget $G_c$, which is the disjoint union of the graph $G_c^T$ and the graph $G_c^F$ depicted in \Cref{fig:stdclawclausegad}, and further add the following edges. 
\begin{itemize}
\item For every $p \in \{x,y,z\}$,  we connect $P_{p,F}^c(2)$ to $f_c^{ab}$ if and only if $p \in \{a,b\}$.
\item For every $p \in \{x,y,z\}$, we connect $P_{p,T}^c(1)$ to $t_c^p$ and further connect $P_{p,T}^c(1)$ to $w_c^{ab}$ if and only if $p \in \{a,b\}$.   
\end{itemize}
We let $G$ be the resulting graph. Let us show that $\Phi$ is satisfiable if and only if $ct_{\gamma_{t2}}(G) =2$. To do so, we will rely on the following key result shown in \cite{semitotcon}.

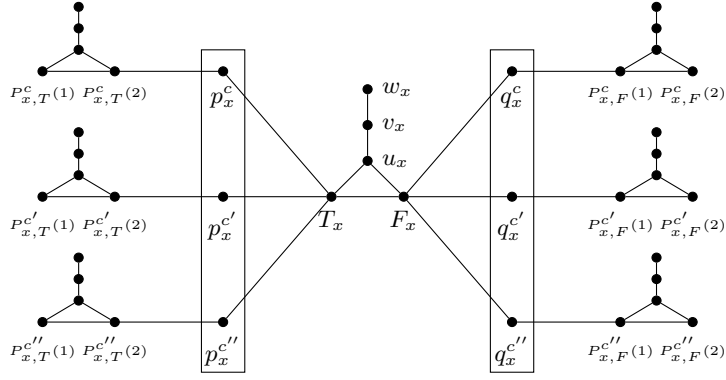
\begin{figure}
\centering
\begin{tikzpicture}[scale=.95]
\node[circ,label=below:{\small $q_x^c$}] (ac) at (2,1.75) {};
\longPaw{3.5}{1.75}{}{$P_{x,F}^c(1)$}{$P_{x,F}^c(2)$};
\draw (ac) -- (3.5,1.75);

\node[circ,label=below:{\small $q_x^{c'}$}] (ac') at (2,0) {};
\longPaw{3.5}{0}{}{$P_{x,F}^{c'}(1)$}{$P_{x,F}^{c'}(2)$};
\draw (ac') -- (3.5,0);

\node[circ,label=below:{\small $q_x^{c''}$}] (ac'') at (2,-1.75) {};
\longPaw{3.5}{-1.75}{}{$P_{x,F}^{c''}(1)$}{$P_{x,F}^{c''}(2)$};
\draw (ac'') -- (3.5,-1.75);

\node[circ,label=below:{\small $p_x^c$}] (dc) at (-2,1.75) {};
\longPaw{-4.5}{1.75}{}{$P_{x,T}^c(1)$}{$P_{x,T}^c(2)$};
\draw (dc) -- (-3.5,1.75);

\node[circ,label=below:{\small $p_x^{c'}$}] (dc') at (-2,0) {};
\longPaw{-4.5}{0}{}{$P_{x,T}^{c'}(1)$}{$P_{x,T}^{c'}(2)$};
\draw (dc') -- (-3.5,0);

\node[circ,label=below:{\small $p_x^{c''}$}] (dc'') at (-2,-1.75) {};
\longPaw{-4.5}{-1.75}{}{$P_{x,T}^{c''}(1)$}{$P_{x,T}^{c''}(2)$};
\draw (dc'') -- (-3.5,-1.75);

\node[circ,label=below:{\small $T_x$ }] (a) at (-.5,0) {};
\node[circ,label=below:{\small $F_x$}] (b) at (.5,0) {};
\node[circ,label=right:{\small $u_x$}] (c) at (0,.5) {};
\node[circ,label=right:{\small $v_x$}] at (0,1) {};
\node[circ,label=right:{\small $w_x$}] (d) at (0,1.5) {};

\draw (a) -- (b) 
(a) -- (c) 
(b) -- (c) 
(c) -- (d);

\draw (a) -- (dc)
(a) -- (dc')
(a) -- (dc'')
(b) -- (ac)
(b) -- (ac')
(b) -- (ac'');

\draw (-2.3,-2.45) rectangle (-1.7,2.05);
\draw (1.7,-2.45) rectangle (2.3,2.05);
\end{tikzpicture}
\caption{The gadget $G_x$ for a variable $x \in X$ contained in clauses $c,c'$ and $c''$ (a rectangle indicates that the corresponding set of vertices induces a clique).}
\label{fig:stdclawvargad}
\end{figure}

\begin{figure}
\centering
\begin{subfigure}{.45\textwidth}
\centering
\begin{tikzpicture}
\node[circ,label=left:{\small $w_c^{yz}$}] (wyz) at (0,0) {};
\node[circ,label=right:{\small $w_c^{xy}$}] (wxy) at (2,0) {};
\node[circ,label=above:{\small $w_c^{xz}$}] (wxz) at (1,1.73) {};

\draw (wyz) -- (wxz) node[circ,midway,label=left:{\small $t_c^z$}]  (tz) {};
\draw (wyz) -- (wxy) node[circ,midway,label=below:{\small $t_c^y$}] (ty) {};
\draw (wxy) -- (wxz) node[circ,midway,label=right:{\small $t_c^x$}] (tx) {};

\draw (tz) -- (ty) -- (tx) -- (tz);

\draw (wyz) edge[bend left=50] (wxz);
\draw (wyz) edge[bend right=55] (wxy);
\draw (wxy) edge[bend right=50] (wxz);

\node[circ,label=above:{\small $u_c$}] (uc) at (1,.45) {};
\draw (uc) -- (tz)
(uc) -- (ty)
(uc) -- (tx); 
\end{tikzpicture}
\caption{The graph $G_c^T$.}
\end{subfigure}
\hspace*{.5cm}
\begin{subfigure}{.45\textwidth}
\centering
\begin{tikzpicture}
\node[circ,label=below:{\small $f_c^{yz}$}] (wyz) at (0,0) {};
\node[circ,label=below:{\small $f_c^{xy}$}] (wxy) at (2,0) {};
\node[circ,label=above:{\small $f_c^{xz}$}] (wxz) at (1,1.73) {};

\draw (wyz) -- (wxy) -- (wxz) -- (wyz);
\end{tikzpicture}
\caption{The graph $G_c^F$.}
\end{subfigure}
\caption{The gadget $G_c$ for a clause $c = x \lor y \lor z$ is the disjoint union of $G^T_c$ and $G^F_c$.}
\label{fig:stdclawclausegad}
\end{figure}
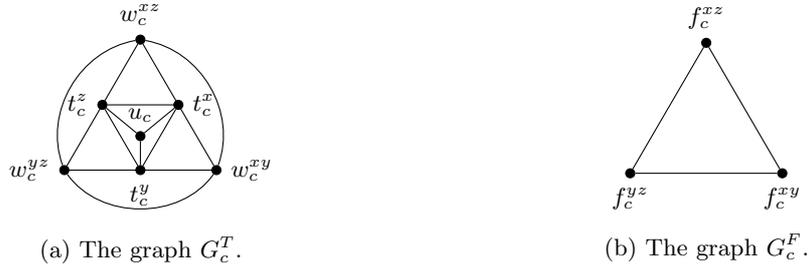

\begin{claim}[\cite{semitotcon}]
\label{clm:stdclaw1}
The following statements are equivalent.
\begin{itemize}
\item[(i)] $\Phi$ is satisfiable.
\item[(ii)] $\gamma_{t2}(G) = 14|X| + |C|$.
\item[(iii)] $ct_{\gamma_{t2}} (G) > 1$.
\end{itemize}
\end{claim}

Observe first that if $ct_{\gamma_{t2}}(G) =2$ then $\Phi$ is satisfiable by \Cref{clm:stdclaw1}. Assume, conversely, that $\Phi$ is satisfiable and consider a truth assignment satisfying $\Phi$. Let us show how to construct a minimum SD set $D$ of $G$ given this assignment. For every variable $x \in X$ appearing in clauses $c,c'$ and $c''$, if $x$ is true, we include $\{v_x,T_x\} \cup \{P_{x,R}^p(1),P_{x,R}^p(4)~|~R \in \{T,F\} \text{ and }p \in \{c,c',c''\}\}$ in $D$; otherwise, we include $\{v_x,F_x\} \cup \{P_{x,R}^p(2),P_{x,R}^p(4)~|~R \in \{T,F\} \text{ and }p \in \{c,c',c''\}\}$ in $D$. For every clause $c \in C$ containing variables $x,y$ and $z$, exactly one variable is set to true, say $x$ without loss of generality, in which case we add $t^y_c$ to $D$ (note that $w_c^{xz},t^x_c$ and $w_c^{xy}$ are already dominated by $P_{x,T}^c(1)$, which also serves as a witness for $t^y_c$). It is not difficult to see that the constructed set $D$ is indeed an SD set of $G$ and since $|D| = 14|X| + |C|$, we conclude by \Cref{clm:stdclaw1} that $D$ is minimum. Now consider a clause $c \in C$ containing variables $x,y$ and $z$, and assume without loss of generality that $x$ is true and $t^y_c \in D$. Then $D \cup \{t^x_c\}$ is an SD set of $G$ of size $\gamma_{t2}(G)+1$ containing the $O_4$ $t^y_c,t^x_c,P_{x,T}^c(1),P_{x,T}^c(4)$ and thus, $ct_{\gamma_{t2}}(G) \leq 2$ by \Cref{thm:friendlytriple}. But $ct_{\gamma_{t2}}(G) > 1$ by \Cref{clm:stdclaw1} and so, $ct_{\gamma_{t2}}(G) =2$. Since $G$ is $K_{1,3}$-free, the lemma follows.
\end{proof}

\begin{lemma}
\label{lem:c5std2}
{\sc Contraction Number($\gamma_{t2}$,2)} is $\mathsf{NP}$-hard on $\{C_k~|~k \geq 5\}$-free graphs.
\end{lemma}

\begin{proof}
The reduction is similar to that of \Cref{lem:2p4td2}. More precisely, we reduce from {\sc 3-Sat}: given an instance $\Phi$ of this problem, with variable set $X$ and clause set $C$, we construct an instance $G$ of {\sc Contraction Number($\gamma_{t2}$,2}) as follows. For every variable $x \in X$, we introduce a long paw $G_x$ on vertex set $\{G_x(1) = x, G_x(2) = \overline{x}, G_x(3) = u_x, G_x(4) = v_x, G_x(5) = w_x\}$ and refer to the vertices $x$ and $\overline{x}$ as \emph{literal vertices}. For every clause $c \in C$, we introduce a \emph{clause vertex}, denoted by $c$, and add an edge between $c$ and every literal vertex whose corresponding literal appears in the clause $c$. Finally, we add an edge between every two clauses vertices so that the set of clause vertices induces a clique, denoted by $K$ in the following. We let $G$ be the resulting graph. We next show that $G$ is a \yes-instance for {\sc Contraction Number($\gamma_{t2}$,2)} if and only if $\Phi$ is satisfiable through a series of claims. Observe first that since for any variable $x \in X$, the vertex $w_x$ should be dominated and any witness for a vertex dominating $w_x$ can only belong to $G_x$, the following holds.

\begin{observation}
\label{obs:domwx}
For any SD set $D$ of $G$ and any variable $x \in X$, $|D \cap V(G_x)| \geq 2$.
\end{observation}

\begin{claim}
\label{clm:phisatc5std2}
$\Phi$ is satisfiable if and only if $\gamma_{t2}(G) = 2|X|$.
\end{claim}

\begin{claimproof}
Assume first that $\Phi$ is satisfiable and consider a truth assignment satisfying $\Phi$. We construct an SD set $D$ of $G$ as follows. For every variable $x \in X$, if $x$ is true then we include $v_x$ and $x$ to $D$; otherwise, we include $v_x$ and $\overline{x}$ to $D$. It is easy to see that the constructed set $D$ is indeed an SD set of $G$ and since $|D| = 2|X|$, we conclude by \Cref{obs:domwx} that $D$ is minimum.

Conversely assume that $\gamma_{t2}(G) = 2|X|$ and let $D$ be a minimum SD set of $G$. Then by \Cref{obs:domwx}, $D \cap V(K) = \varnothing$ and thus, every clause vertex must have a neighbouring literal vertex in $D$. Therefore, the truth assignment obtained by setting a variable $x$ to true when $x \in D$ and to false otherwise, satisfies $\Phi$.
\end{claimproof}

\begin{claim}
\label{clm:c5cstd2}
$ct_{\gamma_{t2}}(G) = 2$ if and only if $\gamma_{t2}(G) = 2|X|$.
\end{claim}

\begin{claimproof}
Assume first that $ct_{\gamma_{t2}}(G) = 2$ and consider a minimum SD set $D$ of $G$. If there exists a variable $x \in X$ such that $|D \cap V(G_x)| \geq 3$, then $D \cap V(G_x)$ contains a friendly triple, a contradiction to \Cref{thm:friendlytriple}. Thus, $|D \cap V(G_x)| \leq 2$ for every variable $x \in X$, and we conclude by \Cref{obs:domwx} that, in fact, equality holds. Now suppose that $D \cap V(K) \neq \varnothing$, say $c \in D \cap V(K)$, and let $x \in X$ be a variable occurring in $c$, say $x$ appears positive in $c$ without loss of generality. Then the set $(D \setminus V(G_x)) \cup \{x,v_x\}$ is a minimum SD set of $G$ containing the friendly triple $c,x,v_x$, a contradiction to \Cref{thm:friendlytriple}. Thus, $D \cap V(K) = \varnothing$ and so, $\gamma_{t2}(G) = |D| = 2|X|$.

Conversely, assume that $\gamma_{t2}(G) = 2|X|$ and let $D$ be a minimum SD set of $G$. Then by \Cref{obs:domwx}, $D \cap V(K) = \varnothing$ and $|D \cap V(G_x)| = 2$ for every variable $x \in X$. Since for every variable $x \in X$, $w_x$ should be dominated, it follows that if for some variable $x \in X$, $D \cap V(G_x)$ contains an edge $e$, then $D \cap \{x,\overline{x}\} = \varnothing$; but then, $d(e,D \setminus e) \geq 3$ and thus, $e$ cannot be part of a friendly triple. Therefore, $D$ contains no friendly triple. Consider now a clause $c \in C$ and let $x,y \in X$ be two variables occurring in $c$, say $x$ and $y$ both appear positive in $c$ without loss of generality. Then the set $(D \setminus (V(G_x) \cup V(G_y))) \cup \{c,x,v_x,y,v_y\}$ is an SD set of $G$ of size $\gamma_{t2}(G) +1$ containing the $O_4$ $x,c,y,v_y$ and thus, $ct_{\gamma_{t2}}(G) = 2$ by \Cref{thm:friendlytriple}. 
\end{claimproof}

Now Claims~\ref{clm:phisatc5std2} and \ref{clm:c5cstd2}, $\Phi$ is satisfiable if and only if  $ct_{\gamma_{t2}}(G) = 2$ if and only if $\Phi$ is satisfiable. Since $G$ is easily seen to contain no induced cycle of length at least five, the lemma follows.
\end{proof}

\begin{lemma}
\label{lem:c3c4std2}
For every $\ell \in \{3,4\}$, {\sc Contraction Number($\gamma_{t2}$,2)} is $\mathsf{NP}$-hard on $C_\ell$-free graphs.
\end{lemma}

\begin{proof}
We reduce from {\sc Positive Not-All-Equal 3-Sat} which is a $\mathsf{NP}$-complete variant \cite{schaefer} of the {\sc 3-Sat} problem where given a formula in which all literals are positive, the problem is to determine whether there exists a truth assignment such that in no clause, all three literals have the same truth value. Given an instance $\Phi$ of this problem, with variable set $X$ and clause set $C$, we construct two instances $G_3$ and $G_4$ of {\sc Contraction Number($\gamma_{t2}$,2)}, one $C_3$-free and the other $C_4$-free, respectively. For both graphs, we start as follows. For every variable $x \in X$, we introduce the gadget $G_x$ depicted in \Cref{fig:vargadc3c4std2}, which has two distinguished \emph{truth vertices} $T_x$ and $F_x$. For every clause $c \in C$ containing variables $x,y$ and $z$, we introduce two \emph{clause vertices} $c$ and $\overline{c}$, and add for every $v \in \{x,y,z\}$, an edge between $c$ and $T_v$, and an edge between $\overline{c}$ and $F_v$. We denote by $V_C = \{c~|~ c \in C\}$ the set of \emph{positive} clause vertices, and by $V_{\overline{C}} = \{\overline{c}~|~c \in C\}$ the set of \emph{negated} clause vertices. This concludes the construction of $G_3$. For the graph $G_4$, we further add edges so that $V_C$ is a clique and $V_{\overline{C}}$ is a clique.

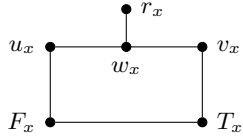
\begin{figure}
\centering
\begin{tikzpicture}
\node[circ,label=left:{\small $F_x$}] (f) at (0,0) {};
\node[circ,label=right:{\small $T_x$}] (t) at (2,0) {};
\node[circ,label=left:{\small $u_x$}] (u) at (0,1) {};
\node[circ,label=below:{\small $w_x$}] (w) at (1,1) {};
\node[circ,label=right:{\small $v_x$}] (v) at (2,1) {};
\node[circ,label=right:{\small $r_x$}] (r) at (1,1.5) {};

\draw (f) -- (t) -- (v) -- (u) -- (f)
(w) -- (r);
\end{tikzpicture}
\caption{The variable gadget $G_x$.}
\label{fig:vargadc3c4std2}
\end{figure}

We next show that for every $\ell \in \{3,4\}$, $ct_{\gamma_{t2}}(G_\ell) =2$ if and only if $\Phi$ is satisfiable through a series of claims. Observe first that since for every variable $x \in X$, the vertex $r_x$ should be dominated and any witness for a vertex dominating $r_x$ can only belong to $G_x$, the following holds.

\begin{observation}
\label{obs:domrx}
Let $\ell \in \{3,4\}$. Then for every SD set $D$ of $G_\ell$ and every variable $x \in X$, $|D \cap V(G_x)| \geq 2$.
\end{observation}

\begin{claim}
\label{clm:phisatc3c4std2}
Let $\ell \in \{3,4\}$. Then $\Phi$ is satisfiable if and only if $\gamma_{t2}(G_\ell) = 2|X|$.
\end{claim}

\begin{claimproof}
Assume first that $\Phi$ is satisfiable and consider a truth assignment satisfying $\Phi$. We construct an SD set $D$ for $G_\ell$. For every variable $x \in X$, if $x$ is true, then we include $w_x$ and $T_x$ in $D$; otherwise, we include $w_x$ and $F_x$ in $D$. It is not difficult to see that the constructed set $D$ is indeed an SD set for both $G_\ell$ and since $|D| = 2|X|$, we conclude by \Cref{obs:domrx} that $D$ is minimum.

Conversely, assume that $\gamma_{t2}(G_\ell) = 2|X|$ and let $D$ be a minimum SD set of $G_\ell$. Then by \Cref{obs:domrx}, $D \cap V_C = D \cap V_{\overline{C}} = \varnothing$ and thus, every (positive or negated) clause vertex must have a neighbouring truth vertex in $D$. Therefore, the truth assignment obtained by setting a variable $x$ to true when $T_x \in D$ and to false otherwise, satisfies $\Phi$.
\end{claimproof}

\begin{claim}
\label{clm:c3c4cstd2}
Let $\ell \in \{3,4\}$. Then $ct_{\gamma_{t2}}(G_\ell) = 2$ if and only if $\gamma_{t2}(G_\ell) = 2|X|$.
\end{claim}

\begin{claimproof}
Assume that $ct_{\gamma_{t2}}(G_\ell) = 2$ and let $D$ be a minimum SD set of $G_\ell$. If there exists a variable $x \in X$ such that $|D \cap V(G_x)| \geq 3$, then $D \cap V(G_x)$ contains a friendly triple, a contradiction to \Cref{thm:friendlytriple}. Thus, $|D \cap V(G_x)| \leq 2$ for every variable $x \in X$ and we conclude by \Cref{obs:domrx} that, in fact, equality holds. Furthermore, we may assume that for every variable $x \in X$, $D \cap \{F_x,T_x\} \neq \varnothing$ and $w_x \in D$, for it otherwise suffices to consider e.g. $(D \setminus V(G_x)) \cup \{F_x,w_x\}$ in place of $D$. Now suppose to the contrary that $D \cap V_C \neq \varnothing$. We distinguish cases depending on the value of $\ell$.

\textbf{Case 1.} \emph{$\ell = 3$.} Let $c \in V_C$ be a vertex in $D$ and let $x,y,z \in X$ be the three variables occurring in $c$. If there exists $v \in \{x,y,z\}$ such that $T_v \in D$, then $c,T_v,w_v$ is a friendly triple, a contradiction to \Cref{thm:friendlytriple}. Thus, for every $v \in \{x,y,z\}$, $T_v \notin D$ and so, $F_v \in D$ by the aforementioned assumption. But then, $(D \setminus \{c\}) \cup \{T_x\}$ is a minimum SD of $G_3$ (recall that for $\ell =3$, $V_C$ is an independent set) containing the friendly triple $F_x,T_x,w_x$, a contradiction to \Cref{thm:friendlytriple}.

\textbf{Case 2.} \emph{$\ell =4$.} Then $|D \cap V_C| = 1$: indeed, if there exist $c,c' \in D \cap V_C$ then $d(\{c,c'\}, D \setminus \{c,c'\}) \leq 2$ (recall that by assumption, $D \cap \{t_v,F_v\} \neq \varnothing$ for every variable $v \in X$) and so, $D$ contains a friendly triple, a contradiction to \Cref{thm:friendlytriple}. Now if there exists a variable $x \in X$ such that $T_x \in D$ then letting $c \in C$ be a clause containing $x$, the set $(D \setminus V_C) \cup \{c\}$ is a minimum SD set of $G_4$ containing the friendly triple $c,T_x,w_x$, a contradiction to \Cref{thm:friendlytriple}. It follows that for every variable $x \in X$, $F_x \in D$; but then, it suffices to consider, in place of $D$, the SD set $(D \setminus \{F_x\}) \cup \{T_x\}$ for some variable $x \in X$, and argue as previously.

Thus, in both cases, $D \cap V_C = \varnothing$; and by symmetry, we conclude that $D \cap V_{\overline{C}} = \varnothing$ as well. Therefore, $\gamma_{t2}(G_\ell) = |D| = 2|X|$.\\

Conversely, assume that $\gamma_{t2}(G_\ell) = 2|X|$ and let $D$ be a minimum SD set of $G_\ell$. Then by \Cref{obs:domrx}, $D \cap V_C = D \cap V_{\overline{C}} = \varnothing$ and $|D \cap V(G_x)| = 2$ for every variable $x \in X$. Since for every variable $x \in X$, $r_x$ should be dominated, it follows that if for some variable $x \in X$, $D \cap V(G_x)$ contains an edge $e$, then $D \cap \{F_x,T_x\} = \varnothing$; but then, $d(e, D \setminus e) \geq 3$ and thus, $e$ cannot be part of a friendly triple. Therefore, $D$ contains no friendly triple. Now let us assume, for the following, that for every variable $x \in X$, $D \cap \{F_x,T_x\} \neq \varnothing$ and $w_x \in D$ (if this does not hold for a variable $x \in X$, we may consider e.g. $(D \setminus V(G_x)) \cup \{F_x,w_x\}$ in place of $D$). Consider a clause $c \in C$ and let $x,y,z \in X$ be the three variables occurring in $c$. By assumption, either $|D \cap \{T_x,T_y,T_z\}| \geq 2$ or $|D \cap \{F_x,F_y,F_z\}| \geq 2$, say $F_x,F_y \in D$ without loss of generality. Then the set $D \cup \{\overline{c}\}$ is an SD set of $G_\ell$ of size $\gamma_{t2}(G) +1$ containing the $O_4$ $F_x,\overline{c},F_y,w_y$ and so, $ct_{\gamma_{t2}}(G) = 2$ by \Cref{thm:friendlytriple}.
\end{claimproof}

Now by Claims~\ref{clm:phisatc3c4std2} and \ref{clm:c3c4cstd2}, we conclude that for every $\ell \in \{3,4\}$, $ct_{\gamma_{t2}}(G_\ell) = 2$ if and only if $\Phi$ is satisfiable. Since for every $\ell \in \{3.4\}$, $G_\ell$ is easily seen to contain no induced $C_\ell$, the lemma follows.
\end{proof}

The following now ensues from Lemmas~\ref{lem:c5std2} and \ref{lem:c3c4std2}.

\begin{corollary}
\label{lem:cyclesstd2}
For every $\ell \geq 3$, {\sc Contraction Number($\gamma_{t2}$,2)} is $\mathsf{NP}$-hard on $C_\ell$-free graphs.
\end{corollary}

We next investigate the complexity of the {\sc Contraction Number($\gamma_{t2}$,3)} problem on several monogenic graph classes. 

\begin{lemma}
{\sc Contraction Number($\gamma_{t2}$,3)} is $\mathsf{NP}$-hard on $K_{1,3}$-free graphs.
\end{lemma}

\begin{proof}
We reduce from {\sc Positive Exactly 3-Bounded 1-In-3 3-Sat} (see \Cref{sec:prelim} for a precise definition of this problem): given an instance $\Phi$ of this problem, with variable set $X$ and clause set $C$, we construct an instance $G$ of {\sc Contraction Number ($\gamma_{t2}$,3)} as follows. For every variable $x \in X$, we introduce the gadget $G_x$ depicted in \Cref{fig:stdclawvargad} (it is the same gadget as the one used in the proof of \Cref{lem:clawstd2}). In the following, we denote by $P_x$ the long paw induced by $\{T_x,F_x,u_x,v_x,w_x\}$ and we may refer to the vertices of $P_x$ as $P_x(1),\ldots,P_x(5)$ where $P_x(1) = t_x$ and $P_x(2) = F_x$. For every clause $c \in C$ containing variables $x,y$ and $z$, we introduce the gadget $G_c$ which is the disjoint union of the graph $G_c^T$ and the graph $G_c^F$ depicted in \Cref{fig:clausegadclawstd3} (where a rectangle indicates that the corresponding set of vertices is a clique), and further add for every $v \in \{x,y,z\}$ an edge between $t_c^v$ and $P_{v,T}^c(1)$, and an edge between $f_c^v$ and $P_{v,F}^c(2)$. Note that $\{u_c^x,u_c^y,u_c^z,a_c\}$ and $\{v_c^x,v_c^y,v_c^z,b_c\}$ induce cliques; in particular, $G_c^T$ is $K_{1,3}$-free. We let $G$ be the resulting graph. We next show that $ct_{\gamma_{t2}}(G) = 3$ if and only if $\Phi$ is satisfiable through a series of claims.

\begin{figure}
\centering
\begin{subfigure}[b]{.45\textwidth}
\centering
\begin{tikzpicture}[scale=.8]
\node[circ,label=below:{\small $t_c^x$}] (tx) at (0,0) {};
\node[circ,label=below:{\small $u_c^x$}] (ux) at (1,0) {};
\node[circ,label=left:{\small $v_c^x$}] (vx) at (.5,.76) {};
\draw (tx) -- (ux) -- (vx) -- (tx);

\node[circ,label=below:{\small $v_c^y$}] (vy) at (3,0) {};
\node[circ,label=below:{\small $t_c^y$}] (ty) at (4,0) {};
\node[circ,label=right:{\small $u_c^y$}] (uy) at (3.5,.76) {};
\draw (ty) -- (uy) -- (vy) -- (ty);

\node[circ,label=left:{\small $u_c^z$}] (uz) at (1.5,2.5) {};
\node[circ,label=right:{\small $v_c^z$}] (vz) at (2.5,2.5) {};
\node[circ,label=above:{\small $t_c^z$}] (tz) at (2,3.26) {};
\draw (tz) -- (uz) -- (vz) -- (tz);

\draw (ux) -- (uy)
(uy) edge[bend right=15] (uz)
(uz) edge[bend right=15] (ux);

\draw (vx) -- (vy)
(vy) edge[bend right=15] (vz)
(vz) edge[bend right=15] (vx);

\node[circ,label=left:{\small $a_c$}] (a) at (1.75,1.35) {};
\draw (a) -- (ux)
(a) edge[bend right=10] (uy)
(a) -- (uz);

\node[circ,label=right:{\small $b_c$}] (b) at (2.25,1.35) {};
\draw (b) edge[bend left=10] (vx)
(b) -- (vy)
(b) -- (vz);
\end{tikzpicture}
\caption{The graph $G_c^T$.}
\end{subfigure}
\hspace*{.5cm}
\begin{subfigure}[b]{.45\textwidth}
\centering
\begin{tikzpicture}
\node[circ,label=below:{\small $w_c$}] (uc) at (0,1) {};
\node[circ,label=below:{\small $u_c$}] (vc) at (1,1) {};
\node[circ,label=below:{\small $v_c$}] at (.5,1) {};
\draw (uc) -- (vc);

\node[circ,label=below:{\small $q_c^x$}] (qx) at (2,0) {};
\node[circ,label=below:{\small $f_c^x$}] (fx) at (3,0) {};
\draw (qx) -- (fx);

\node[circ,label=below:{\small $q_c^y$}] (qy) at (2,1) {};
\node[circ,label=below:{\small $f_c^y$}] (fy) at (3,1) {};
\draw (qy) -- (fy);

\node[circ,label=below:{\small $q_c^z$}] (qz) at (2,2) {};
\node[circ,label=below:{\small $f_c^z$}] (fz) at (3,2) {};
\draw (qz) -- (fz);

\draw (1.75,-.55) rectangle (2.25,2.25);

\draw (vc) -- (qx)
(vc) -- (qy)
(vc) -- (qz);
\end{tikzpicture}
\caption{The graph $G_c^F$ (the rectangle indicates that the corresponding set of vertices is a clique).}
\end{subfigure}
\caption{The clause gadget $G_c$ for a clause $c= x \lor y \lor z$ is the disjoint union of $G_c^T$ and $G_c^F$.}
\label{fig:clausegadclawstd3}
\end{figure}
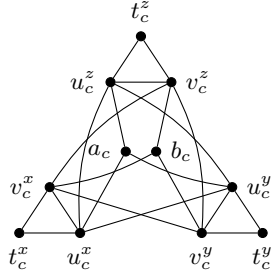
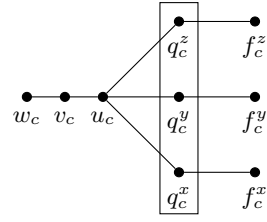

\begin{claim}
\label{clm:sizeclawstd3}
For any SD set $D$ of $G$, the following hold.
\begin{itemize}
\item[(i)] For every variable $x \in X$ and every long paw $P$ contained in $G_x$, $|D \cap V(P)| \geq 2$ and $D \cap \{P(4),P(5)\} \neq \varnothing$. In particular, $|D \cap V(G_x)| \geq 14$.
\item[(ii)]  For every clause $c \in C$ containing variables $x,y$ and $z$, $|D \cap \{u_c,v_c,w_c,q_c^x,q_c^y,q_c^z\}| \geq 2$ and $\min(|D \cap \{u_c^x,u_c^y,u_c^z,a_c\}|, |D \cap \{v_c^x,v_c^y,v_c^z,b_c\}|) \geq 1$. In particular, $|D \cap V(G_c)| \geq 4$. 
\end{itemize}
\end{claim}

\begin{claimproof}
Let $D$ be an SD set of $G$. Since in a long paw $P$, the vertex $P(5)$ should be dominated and any witness for a vertex dominating $P(5)$ can only belong to $P$, item (i) follows. To prove (ii), consider a clause $c \in C$ containing variables $x,y$ and $z$. Since $w_c$ must be dominated and any witness for a vertex dominating $w_c$ can only belong to $\{u_c,v_c,w_c,q_c^x,q_c^y,q_c^z\}$, the first part of item (ii) follows. Now since $a_c$ should be dominated, $D \cap \{u_c^x,u_c^y,u_c^z,a_c\} \neq \varnothing$; and similarly, since $b_c$ should be dominated, $D \cap \{v_c^x,v_c^y,v_c^z,b_c\} \neq \varnothing$ which proves the second part of item (ii).
\end{claimproof}

The following is an immediate consequence of \Cref{clm:sizeclawstd3}(i).

\begin{observation}
\label{obs:emptycliques}
For every SD set $D$ of $G$ and every variable $x \in X$ contained in clauses $c,c'$ and $c''$, if $|D \cap V(G_x)| =14$ then $D \cap \{q_x^\ell,p_x^\ell~|~ \ell \in \{c,c',c''\}\} = \varnothing$.
\end{observation}

\begin{claim}
\label{obs:nofs}
Let $D$ be an SD set of $G$ and let $c \in C$ be a clause containing variables $x,y$ and $z$. If $|D \cap V(G_c^F)| =2$ then $D \cap \{f_c^x,f_c^y,f_c^z\} = \varnothing$. Similarly, if $|D \cap V(G_c^T)| = 2$ then $D \cap \{t_c^x,t_c^y,t_c^z\} = \varnothing$.
\end{claim}

\begin{claimproof}
If $|D \cap V(G_c^F)| =2$ then by \Cref{clm:sizeclawstd3}(ii), $|D \cap V(G_c^F)| = |D \cap \{u_c,v_c,w_c,q_c^x,q_c^y,q_c^z\}| = 2$ and thus, $D \cap \{f_c^x,f_c^y,f_c^z\} = \varnothing$. Now if $|D \cap V(G_c^T)| =2$ then by \Cref{clm:sizeclawstd3}(ii), $|D \cap \{u_c^x,u_c^y,u_c^z,a_c\}| = |D \cap \{v_c^x,v_c^y,v_c^z,b_c\}| = 1$ and thus, $D \cap \{t_c^x,t_c^y,t_c^z\} = \varnothing$.
\end{claimproof}

\begin{claim}
\label{clm:TxFx}
Let $D$ be an SD set of $G$. If for some variable $x \in X$ contained in clauses $c,c'$ and $c''$, $|D \cap V(G_x)| = 14$ then the following hold.
\begin{itemize}
\item[(i)] If there exists $\ell \in \{c,c',c''\}$ such that $P_{x,T}^\ell(1) \in D$ then $T_x \in D$.
\item[(ii)] If there exists $\ell \in \{c,c',c''\}$ such that $P_{x,F}^\ell(2) \in D$ then $F_x \in D$.
\end{itemize}
In particular, if $P_{x,T}^\ell(1) \in D$ for some $\ell \in \{c,c',c''\}$, then $D \cap \{P_{x,F}^c(2),P_{x,F}^{c'}(2),P_{x,F}^{c''}(2)\} = \varnothing$. Similarly, if $P_{x,F}^\ell(2) \in D$ for some $\ell \in \{c,c',c''\}$, then $D \cap \{P_{x,T}^c(1),P_{x,T}^{c'}(1),P_{x,T}^{c''}(1)\} = \varnothing$. 
\end{claim}

\begin{claimproof}
Let $x \in X$ be a variable such that $|D \cap V(G_x)| = 14$ and let $c,c',c'' \in C$ be the clauses in which $x$ appears. Then by \Cref{obs:emptycliques}, $D \cap \{p_x^\ell,q_x^\ell~|~ \ell \in \{c,c',c''\}\} = \varnothing$ and by \Cref{clm:sizeclawstd3}(i), $|D \cap \{P(1),P(2)\}| \leq 1$ for every long paw $P$ of $G_x$. Thus if $P_{x,T}^\ell(1) \in D$ for some $\ell \in \{c,c',c''\}$, then $T_x \in D$ as $p_x^\ell$ should be dominated; but then, $F_x \notin D$ which implies that $\{P_{x,F}^c(1),P_{x,F}^{c'}(1),P_{x,F}^{c''}(1)\} \subseteq D$ (one of $q_x^c,q_x^{c'}$ and $q_x^{c''}$ would otherwise not be dominated) and so, $D \cap \{P_{x,F}^c(2),P_{x,F}^{c'}(2),P_{x,F}^{c''}(2)\} = \varnothing$. Item (ii) follows by symmetry.
\end{claimproof}

\begin{claim}
\label{clm:phisatclawstd3}
$\Phi$ is satisfiable if and only if $\gamma_{t2}(G) = 14|X| + 4|C|$.
\end{claim}

\begin{claimproof}
Assume first that $\Phi$ is satisfiable and consider a truth assignment satisfying $\Phi$. We construct an SD set $D$ of $G$ as follows. For every variable $x \in X$ contained in clauses $c,c'$ and $c''$, if $x$ is true then we include $\{v_x,T_x\} \cup \{P_{x,R}^\ell(1),P_{x,R}^\ell(4)~|~R \in \{T,F\} \text{ and }\ell \in \{c,c',c''\}\}$ in $D$; otherwise we include $\{v_x,F_x\} \cup \{P_{x,R}^\ell(2),P_{x,R}^\ell(4)~|~R \in \{T,F\} \text{ and }\ell \in \{c,c',c''\}\}$ in $D$. For every clause $c \in C$ containing variables $x,y$ and $z$, exactly one variable is true, say $x$ without loss of generality, in which case we include $\{w_c,q_c^x\} \cup \{v_c^y,u_c^z\}$ in $D$. Now it is not difficult to see that the constructed set $D$ is indeed an SD set of $G$ and since $|D| = 14|X| + 4|C|$, we conclude by \Cref{clm:sizeclawstd3} that $D$ is minimum. 

Conversely, assume that $\gamma_{t2}(G) = 14|X| + 4|C|$ and let $D$ be a minimum SD set of $G$. Consider a clause $c \in C$ and let $x,y,z \in X$ be the variables contained in $c$. Then $|D \cap \{q_c^x,q_c^y,q_c^z\}| \leq 1$: indeed, $D \cap \{w_c,v_c\} \neq \varnothing$ as $w_c$ should be dominated, and $|D \cap \{u_c,v_c,w_c,q_c^x,q_c^y,q_c^z\}| = 2$ by \Cref{clm:sizeclawstd3}(i). Since $D \cap \{f_c^x,f_c^y,f_c^z\} = \varnothing$ by \Cref{obs:nofs}, it follows that at least two vertices amongst $f_c^x,f_c^y$ and $f_c^z$ are not dominated by a vertex in $G_c^F$, say $f_c^y$ and $f_c^z$ without loss of generality. Then $P_{y,F}^c(2),P_{z,F}^c(2) \in D$ as $f_c^y$ and $f_c^z$ should be dominated, and so $F_y,F_z \in D$ by \Cref{clm:TxFx}(ii). It then follows from \Cref{clm:sizeclawstd3}(i) that $T_y,T_z \notin D$ which, by \Cref{clm:TxFx}(i), implies that $P_{y,T}^c(1),P_{z,T}^c(1) \notin D$; in particular, $t_c^y$ and $t_c^z$ must be dominated by vertices in $G_c^T$. But then, $D \cap \{t_c^x,v_c^x,u_c^x\} = \varnothing$: indeed, since $|D \cap V(G_c^T)| = 2$ by assumption, if $D \cap \{t_c^x,v_c^x,u_c^x\} \neq \varnothing$ then one of $t_c^y$ and $t_c^z$ is not dominated. Since $t_c^x$ should be dominated nonetheless, it follows that $P_{x,T}^c(1) \in D$ and thus, $T_x \in D$ by \Cref{clm:TxFx}(i). Therefore, the truth assignment obtained by setting a variable $x$ to true when $T_x \in D$ and to false otherwise, satisfies $\Phi$.
\end{claimproof}

\begin{claim}
\label{clm:notruedge}
Let $D$ be an SD set of $G$ and let $c \in C$ be a clause containing variables $x,y$ and $z$. If $|D \cap V(G_c)| = 4$ and $|D \cap V(G_\ell)| =14 $ for every $\ell \in \{x,y,z\}$ then $D \cap V(G_c^T)$ is an independent set.
\end{claim}

\begin{claimproof}
Assume that $|D \cap V(G_c)| = 4$ and $|D \cap V(G_\ell)| =14 $ for every $\ell \in \{x,y,z\}$ (note that then $|D \cap V(G_c^T)| = |D \cap V(G_c^F)| = 2$ by \Cref{clm:sizeclawstd3}(ii)). Suppose to the contrary that $D \cap V(G_c^T)$ contains an edge $uv$. Then $\{u,v\} = \{u_c^\ell,v_c^\ell\}$ for some $\ell \in \{x,y,z\}$: indeed, $\{u,v\} \cap \{u_c^x,u_c^y,u_c^z,a_c\} \neq \varnothing$ and $\{u,v\} \cap \{v_c^x,v_c^y,v_c^z,b_c\} \neq \varnothing$ by \Cref{clm:sizeclawstd3}(ii), and any edge between $\{u_c^x,u_c^y,u_c^z,a_c\}$ and $\{v_c^x,v_c^y,v_c^z,b_c\}$ connect vertices with the same superscript. Assume without loss of generality that $\{u,v\} = \{u_c^x,v_c^x\}$. Then $P_{y,T}^c(1),P_{z,T}^c(1) \in D$ as $t_c^y$ and $t_c^z$ should be dominated, which, by \Cref{clm:TxFx}, implies that $P_{y,F}^c(2),P_{z,F}^c(2) \notin D$. By \Cref{obs:nofs}, it must then be that $q_c^y,q_c^z \in D$ for one of $f_c^y$ and $f_c^z$ would otherwise not be dominated; but $|D \cap V(G_c^F)| = 2$ and thus, $w_c$ is not dominated.
\end{claimproof}

\begin{claim}
\label{clm:notruedge1}
Let $D$ be an SD set of $G$ and let $c \in C$ be a clause containing variables $x,y,z \in X$. If $|D \cap V(G_c)| = 4$ and $|D \cap V(G_\ell)| =14 $ for every $\ell \in \{x,y,z\}$ then $D \cap V(G_c^F)$ is an independent set.
\end{claim}

\begin{claimproof}
Assume that $|D \cap V(G_c)| = 4$ and $|D \cap V(G_\ell)| =14 $ for every $\ell \in \{x,y,z\}$ (note that then $|D \cap V(G_c^T)| = |D \cap V(G_c^F)| = 2$ by \Cref{clm:sizeclawstd3}(ii)). Suppose for a contradiction that $D \cap V(G_c^F)$ contains an edge $uv$. Then, since $w_c$ should be dominated, $\{u,v\} \subseteq \{u_c,v_c,w_c\}$; in particular $D \cap \{q_c^\ell,f_c^\ell~|~\ell \in \{x,y,z\}\} = \varnothing$ by Claims~\ref{clm:sizeclawstd3}(ii) and \ref{obs:nofs}. But then, $\{P_{\ell,F}^c(2)~|~\ell \in \{x,y,z\}\} \subseteq D$ as $f_c^x,f_c^y$ and $f_c^z$ should be dominated, which, by \Cref{clm:TxFx}, implies that $\{P_{\ell,T}^c(1)~|~ \ell \in \{x,y,z\}\} \subseteq D$. But $|D \cap V(G_c^T)| = 2$ and so, one of $t_c^x,t_c^y$ and $t_c^z$ is not dominated.
\end{claimproof}

\begin{claim}
\label{clm:edgeinvar}
Let $D$ be an SD set of $G$ and let $x \in X$ be a variable. If $|D \cap V(G_x)| = 14$ and $D \cap V(G_x)$ contains an edge $uv$ then $\{u,v\} = \{P(3),P(4)\}$ for some long paw $P$ of $G_x$.
\end{claim}

\begin{claimproof}
Assume that $|D \cap V(G_x)| =14$ and $D \cap V(G_x)$ contains an edge $uv$. Then by \Cref{obs:emptycliques}, $u,v \in V(P)$ for some long paw $P$ of $G_x$. But $|D \cap V(P)| =2$ and $D \cap \{P(4),P(5)\} \neq \varnothing$ by \Cref{clm:sizeclawstd3}(i) and so, $D \cap V(P) = \{u,v\} \subseteq \{P(3),P(4),P(5)\}$. Now suppose for a contradiction that $P(3) \notin D$. If $P = P_x$ then both $T_x$ and $F_x$ are not dominated since by \Cref{obs:emptycliques}, $q_x^\ell,p_x^\ell \notin D$ for every clause $\ell$ containing $x$. If $P = P_{x,T}^\ell$ for some clause $\ell$ containing $x$, then $P_{x,T}^\ell(2)$ is not dominated since $p_x^\ell \notin D$ by \Cref{obs:emptycliques}. Finally if $P = P_{x,F}^\ell$ for some clause $\ell$ containing $x$, then $P_{x,F}^\ell(1)$ is not dominated since $q_x^\ell \notin D$ by \Cref{obs:emptycliques}. Thus, $P(3) \in D$ and so, $\{u,v\} = \{P(3),P(4)\}$.  
\end{claimproof}

\begin{claim}
\label{clm:extremities}
Let $D$ be an SD set of $G$ and let $c \in C$ be a clause containing variables $x,y$ and $z$. If $|D \cap V(G_c)| = 4$ and $|D \cap V(G_\ell)| =14$ for every $\ell \in \{x,y,z\}$, then for every $\ell \in \{x,y,z\}$, the following hold.
\begin{itemize}
\item[(i)] $D \cap \{u_c^\ell,v_c^\ell\} = \varnothing$ if and only if $P_{\ell,T}^c(1) \in D$.
\item[(ii)] $q_c^\ell \notin D$ if and only if $P_{\ell,F}^c(2) \in D$.
\item[(iii)] $P_{\ell,F}^c(2) \in D$ if and only if $P_{\ell,T}^c(1) \notin D$
\end{itemize}
\end{claim}

\begin{claimproof}
Assume that $|D \cap V(G_c)| = 4$ and $|D \cap V(G_\ell)| =14$ for every $\ell \in \{x,y,z\}$. Now consider $\ell \in \{x,y,z\}$ and let us first prove (i). If $D \cap \{u_c^\ell,v_c^\ell\} = \varnothing$ then since $t_c^\ell \notin D$ by \Cref{obs:nofs}, necessarily $P_{\ell,T}^c(1) \in D$ as $t_c^\ell$ should be dominated. Conversely, assume that $P_{\ell,T}^c(1) \in D$ and suppose for a contradiction that $D \cap \{u_c^\ell,v_c^\ell\} \neq \varnothing$. Then since $|D \cap V(G_c^T)| = 2$ by \Cref{clm:sizeclawstd3}(ii), it must be that $D \cap \{u_c^p,v_c^p\} = \varnothing$ for some $p \neq \ell$; and since $t_c^p \notin D$ by \Cref{obs:nofs}, necessarily $P_{p,T}^c(1) \in D$ as $t_c^p$ should be dominated. Then $P_{\ell,F}^c(2),P_{p,F}^c(2) \notin D$ by \Cref{clm:TxFx} and since $f_c^\ell,f_c^p \notin D$ by \Cref{obs:nofs}, it follows that $q_c^\ell,q_c^p \in D$ as $f_c^\ell$ and $f_c^p$ should be dominated. But $|D \cap V(G_c^F)| =2$ by \Cref{clm:sizeclawstd3} and thus, $w_c$ is not dominated. 

To prove (ii), observe that if $q_c^\ell \notin D$ then since $f_c^\ell \notin D$ by \Cref{obs:nofs}, necessarily $P_{\ell,F}^c(2) \in D$ for $f_c^\ell$ would otherwise not be dominated. Conversely, assume that $P_{\ell,F}^c(2) \in D$ for some $\ell \in \{x,y,z\}$, and suppose for a contradiction that $q_c^\ell \in D$. Then since $|D \cap V(G_c^F)|=2$ by \Cref{clm:sizeclawstd3}(ii) and $w_c$ must be dominated, $q_c^p \notin D$ for every $p \neq \ell$; and since by \Cref{obs:nofs}, $f_c^p \notin D$ for every $p \neq \ell$, necessarily $P_{p,F}^c(2) \in D$ for every $p \neq \ell$ ($f_c^p$ would otherwise not be dominated). It then follows from \Cref{clm:TxFx} that $D \cap \{P_{x,T}^c(1),P_{y,T}^c(1),P_{z,T}^c(1)\} = \varnothing$; but $|D \cap V(G_c^T)| = 2$ by \Cref{clm:sizeclawstd3}(ii) and so, one of $t_c^x,t_c^y$ and $t_c^z$ is not dominated.

Let us finally prove (iii). If $P_{\ell,F}^c(2) \in D$ then by \Cref{clm:TxFx}, $P_{\ell,T}^c(1) \notin D$. Conversely, assume that $P_{\ell,T}^c(1) \notin D$. Then by item (i), $D \cap \{u_c^\ell,v_c^\ell\} \neq \varnothing$ which implies that $D \cap \{u_c^p,v_c^p\} = \varnothing$ for some $p \neq \ell$, since $|D \cap V(G_c^T)| = 2$ by \Cref{clm:sizeclawstd3}(ii). It then follows from item (i) that $P_{p,T}^c(1) \in D$ and so, $P_{p,F}^c(2) \notin D$ by \Cref{clm:TxFx}. Thus, by item (ii), $q_c^p \in D$; but $|D \cap V(G_c^F)| =2$ by \Cref{clm:sizeclawstd3}(ii) and $w_c$ must be dominated, and so $q_c^\ell \notin D$ which, by item (ii), implies that $P_{\ell,F}^c(2) \in D$.
\end{claimproof}

\begin{claim}
\label{clm:novargadedge}
Let $D$ be an SD set of $G$ and let $x \in X$ be a variable contained in clauses $c,c'$ and $c''$. If $|D \cap V(G_\ell)| =4$ for every $\ell \in \{c,c''c''\}$ and $|D \cap V(G_v)|=14$ for every variable $v$ appearing in $c,c'$ or $c''$, then $D \cap V(G_x)$ is an independent set.
\end{claim}

\begin{claimproof}
Assume that $|D \cap V(G_\ell)| =4$ for every $\ell \in \{c,c''c''\}$ and $|D \cap V(G_v)|=14$ for every variable $v$ appearing in $c,c'$ or $c''$. Suppose for a contradiction that $D \cap V(G_x)$ contains an edge $uv$. Then by \Cref{clm:edgeinvar}, $\{u,v\}= \{P(3),P(4)\}$ for some long paw $P$ of $G_x$. If $P=P_x$ then since $D \cap \{q_x^\ell,p_x^\ell~|~\ell \in \{c,c',c''\}\} = \varnothing$ by \Cref{obs:emptycliques}, necessarily $P_{x,T}^c(2),P_{x,F}^c(1) \in D$ (one of $p_x^c$ and $q_x^c$ would otherwise not be dominated); but then by \Cref{clm:sizeclawstd3}(i), $P_{x,T}^c(1),P_{x,F}^c(2) \notin D$, a contradiction to \Cref{clm:extremities}(iii). Thus $P = P_{x,R}^\ell$ for some $\ell \in \{c,c',c''\}$ and $R \in \{T,F\}$. Assume that $R = T$ (the case where $R=F$ is symmetric). Then since $D \cap \{q_x^\ell,p_x^\ell~|~\ell \in \{c,c',c''\}\} = \varnothing$ by \Cref{obs:emptycliques}, necessarily $T_x \in D$ ($p_x^\ell$ would otherwise not be dominated) and so, $F_x \notin D$ by \Cref{clm:sizeclawstd3}(i). But then by \Cref{clm:TxFx}(ii), $P_{x,F}^\ell(2) \notin D$, a contradiction to \Cref{clm:extremities}(ii) as $P_{x,T}^\ell(1) \notin D$.\end{claimproof}

\begin{claim}
\label{clm:inter}
If $\gamma_{t2}(G) = 14|X| + 4|C|$ then for every minimum SD set $D$ of $G$, the following hold.
\begin{itemize}
\item[(i)] $D$ is an independent set.
\item[(ii)] For every $u \in D$, $|w_D(u)|=1$.
\end{itemize}
\end{claim}

\begin{claimproof}
Assume that $\gamma_{t2}(G) = 14|X| + 4|C|$ and let $D$ be a minimum SD set of $G$. Then by Claims~\ref{obs:nofs}, \ref{clm:notruedge}, \ref{clm:notruedge1} and \ref{clm:novargadedge}, $D$ contains no edge. Let us next show that for every $u \in D$, $|w_D(u)| = 1$. Consider first a clause $c \in C$ containing variables $x,y,z \in X$ and let $u \in D \cap V(G_c)$. Note that since $|D \cap V(G_c^T)| = |D \cap V(G_c^F)| = 2$ by \Cref{clm:sizeclawstd3}(ii), if $|w_D(u)| \geq 2$ then $d(u, D \setminus V(G_c)) \leq 2$. Now if $u \in \{u_c^\ell,v_c^\ell\}$ for some $\ell \in \{x,y,z\}$, then by \Cref{clm:extremities}(i), $P_{\ell,T}^c(1) \notin D$ and thus, $d(u,D \setminus V(G_c)) > 2$. Similarly, if $u = q_c^\ell$ for some $\ell \in \{x,y,z\}$, then by \Cref{clm:extremities}(ii), $P_{\ell,F}^c(2) \notin D$ and thus, $d(u,D \setminus V(G_c)) > 2$. Since $u \notin \{f_c^\ell,t_c^\ell~|~\ell \in \{x,y,z\}\}$ by \Cref{obs:nofs} and $d(w, D \setminus V(G_c)) \geq 3$ for every $w \in \{a_c,b_c\}$, it follows that $|w_D(u)|=1$. Consider next a variable $x \in X$ contained in clauses $c,c',c'' \in C$, and suppose to the contrary that there exists a vertex $u \in D \cap V(G_x)$ such that $|w_D(u)| \geq 2$. Since $|D \cap V(P)| =2$ for ever long paw $P$ in $G_x$ and $D \cap \{q_x^\ell,p_x^\ell~|~\ell \in \{c,c',c''\}$, necessarily $u \in \{P(1),P(2)\}$ for some long paw $P$ of $G_x$. Suppose first that $u = T_x$ (the case where $u = F_x$ is symmetric). Then by \Cref{obs:emptycliques}, there must exist $\ell \in \{c,c',c''\}$ such that $P_{x,T}^\ell(2) \in D$; in particular, $P_{x,T}^\ell(1) \notin D$ by \Cref{clm:sizeclawstd3}(i). But $F_x \notin D$ by \Cref{clm:sizeclawstd3}(i) and so, $P_{x,F}^\ell(2) \notin D$ by \Cref{clm:TxFx}(ii), a contradiction to \Cref{clm:extremities}(iii). Suppose next that $u \in \{P_{x,R}^\ell(1), P_{x,R}^\ell(2)\}$ for some $\ell \in \{c,c',c'\}$ and $R \in \{T,F\}$. Assume that $R = F$ (the case where $R = T$ is symmetric). Then $u$ must have a witness in $D \cap V(G_\ell^F)$: indeed, if $u$ had two witnesses in $D \cap V(G_x)$ then one of them would be $F_x$ and so, $|w_D(F_x)| \geq 2$ which is excluded by the previous case. However, $f_\ell^x \notin D$ by \Cref{obs:nofs}, and if $q_\ell^x \in D$ then $P_{x,F}^\ell(2) \notin D$ by \Cref{clm:extremities}(ii); thus $d(u,D \cap V(G_\ell^F)) > 2$, a contradiction. Therefore, every vertex in $D \cap V(G_x)$ has a unique witness.
\end{claimproof}

\begin{claim}
\label{clm:notmin}
Let $D$ be an SD set of $G$ and let $x \in X$ be a variable. If $|D \cap V(G_x)| =14$ then for every clause $c \in C$ containing $x$, the following hold.
\begin{itemize}
\item[(i)] If $D \cap \{P_{x,T}^c(1),P_{x,T}^c(3)\} \neq \varnothing$ then $P_{x,F}^c(2) \notin D$.
\item[(ii)] If $D \cap \{P_{x,F}^c(2),P_{x,F}^c(3)\} \neq \varnothing$ then $P_{x,T}^c(1) \notin D$.
\end{itemize}
\end{claim}

\begin{claimproof}
Assume that $|D \cap V(G_x)| = 14$ and let $c \in C$ be a clause containing $x$. Suppose that $D \cap \{P_{x,T}^c(1),P_{x,T}^c(3)\} \neq \varnothing$. Then by \Cref{clm:sizeclawstd3}(i), $P_{x,T}^c(2) \notin D$ which, by \Cref{obs:emptycliques}, implies that $T_x \in D$ ($p_x^c$ would otherwise not be dominated); but then $F_x \notin D$ by \Cref{clm:sizeclawstd3}(i) and so, $P_{x,F}^c(2) \notin D$ by \Cref{clm:TxFx}(ii). Item (ii) follows by symmetry.
\end{claimproof}

\begin{claim}
\label{clm:clawcstd3}
$ct_{\gamma_{t2}}(G) =3$ if and only if $\gamma_{t2}(G) = 14|X| + 4|C|$.
\end{claim}

\begin{claimproof}
Assume first that $ct_{\gamma_{t2}}(G) = 3$ and let $D$ be a minimum SD set of $G$. Consider a variable $x \in X$ contained in clauses $c,c'$ and $c''$. If $|D \cap V(P)| \geq 3$ for some long paw $P$ of $G_x$ then $D \cap V(P)$ contains a friendly triple, a contradiction to \Cref{thm:friendlytriple}. Thus $|D \cap V(P)| \leq 2$ for every long paw $P$ of $G_x$ and we conclude by \Cref{clm:sizeclawstd3}(i) that, in fact, equality holds. Let us assume, in the following, that for every long paw $P$ of $G_x$, $D \cap \{P(1),P(2)\} \neq \varnothing$ and $P(4) \in D$ (if this does not hold for a long paw $P$ then it suffices to consider, e.g., $(D \setminus V(P)) \cup \{P(1),P(4)\}$ in place of $D$). Now suppose that $D \cap \{p_x^c,p_x^{c'},p_x^{c''}\} \neq \varnothing$, say $p_x^c \in D$ without loss of generality. Then $T_x \notin D$ ($D$ would otherwise contain the friendly triple $p_x^c,T_x,v_x$) and so, $(D \setminus \{p_x^c\}) \cup \{T_x\}$ is a minimum SD set of $G$ containing the friendly triple $T_x,F_x,w_x$, a contradiction to \Cref{thm:friendlytriple}. Thus $D \cap \{p_x^c,p_x^{c'},p_x^{c''}\} = \varnothing$ and by symmetry, we conclude that $D \cap \{q_x^c,q_x^{c'},q_x^{c''}\} = \varnothing$ as well. Therefore, $|D \cap V(G_x)| = 14$. 

Next, consider a clause $c \in C$ and let $x,y,z \in X$ be the variables contained in $c$. Then $|D \cap \{u_c,v_c,w_c,q_c^x,q_c^y,q_c^z\}| \leq 2$ for $D \cap \{u_c,v_c,w_c,q_c^x,q_c^y,q_c^z\}$ would otherwise contain a friendly triple; and we conclude by \Cref{clm:sizeclawstd3} that, in fact, equality holds. Now suppose that $D \cap \{f_c^x,f_c^y,f_c^z\} \neq \varnothing$, say $f_c^x \in D$ without loss of generality. Then $q_c^x \notin D$: indeed, $D$ would otherwise contain the friendly triple $f_c^x,q_c^x,t$, where $t \in D \cap \{u_c,v_c,w_c,q_c^y,q_c^z\}$, a contradiction to \Cref{thm:friendlytriple}. But by assumption, $P_{x,F}^c(2)$ is not a private neighbour of $f_c^x$ and the vertex in $D \cap \{P_{x,F}^c(1),P_{x,F}^c(2)\}$ is witnessed by $P_{x,F}^c(4)$, which implies that $D' = (D \setminus \{f_c^x\}) \cup \{q_c^x\}$ is a minimum SD set of $G$ where $|D' \cap \{u_c,v_c,w_c,q_c^x,q_c^y,q_c^z\}| \geq 3$, a contradiction to the above. Thus, $D \cap \{f_c^x,f_c^y,f_c^z\} = \varnothing$ and so, $|D \cap V(G_c^F)| = 2$. Suppose next that $|D \cap V(G_c^T)| \geq 3$. If $D \cap \{t_c^x,t_c^y,t_c^z\} = \varnothing$, say $t_c^x \in D$ without loss of generality, then surely one of $u_c^x$ and $v_c^x$ does not belong to $D$ ($D$ would otherwise contain the friendly triple $t_c^x,u_c^x,v_c^x$), say $u_c^x \notin D$ without loss of generality. Since by assumption, $P_{x,T}^c(1)$ is not a private neighbour of $t_c^x$ and the vertex in $D \cap \{P_{x,T}^c(1),P_{x,T}^c(2)\}$ is witnessed by $P_{x,T}^c(4)$, the set $D' = (D \setminus \{t_c^x\} \cup \{u_c^x\}$ is a minimum SD of $G$ (note that $u_c^x$ is within distance two of every vertex in $D \cap V(G_c^T)$). By repeating this argument if necessary, we obtain a minimum SD set $D''$ of $G$ such that $D'' \cap \{t_c^x,t_c^y,t_c^z\} = \varnothing$ and $|D'' \cap V(G_c^T)| \geq 3$; but then, it is easy to see that $D'' \cap (\{a_c,b_c\} \cup \{u_c^\ell,v_c^\ell~|~ \ell \in \{x,y,z\}\}$ contains a friendly triple, a contradiction to \Cref{thm:friendlytriple}. Thus, $|D \cap V(G_c^T)| \leq 2$ and we conclude by \Cref{clm:sizeclawstd3}(ii) that, in fact, equality holds. Therefore, $|D \cap V(G_c)| = 4$ and thus, $\gamma_{t2}(G) = |D| = 14|X| + 4|C|$.\\

Conversely, assume that $\gamma_{t2}(G) = 14|X| + 4|C|$. Observe first that since by \Cref{clm:inter}, every minimum SD set of $G$ is an independent set, no minimum SD set of $G$ contains a friendly triple; thus by \Cref{thm:friendlytriple}(i), $ct_{\gamma_{t2}}(G) >1$. Now suppose to the contrary that $ct_{\gamma_{t2}}(G) =2$ and let $D$ be an SD set of size $\gamma_{t2}(G) +1$ containing an ST configuration (see \Cref{thm:friendlytriple}(ii)). Then by \Cref{clm:sizeclawstd3}, there exists either a clause $c \in C$ such that $|D \cap V(G_c)| =5$ or a variable $x \in X$ such that $|D \cap V(G_x)| =15$. We next distinguish these two cases.\\

\noindent
\textbf{Case 1.} \emph{There exists a clause $c \in C$ such that $|D \cap V(G_c)| =5$.} Let $x,y,z \in X$ be the three variables appearing in $c$. Then for any clause $\ell \in C \setminus \{c\}$, $D \cap V(G_\ell)$ contains no edges by Claims~\ref{clm:notruedge} and \ref{clm:notruedge1}; and for any variable $v \in X \setminus \{x,y,z\}$, $D \cap V(G_v)$ contains no edge by \Cref{clm:novargadedge}. Thus by \Cref{obs:nofs}, $D \setminus (V(G_c) \cup V(G_x) \cup V(G_y) \cup V(G_z))$ contains no edge and a fortiori, no ST-configurations. Further observe if for some $\ell \in \{x,y,z\}$, $D \cap V(G_\ell)$ contains an edge $uv$ then by \Cref{clm:edgeinvar}, $\{u,v\} = \{P(3),P(4)\}$ for some long paw $P$ in $G_\ell$. Now if $P \neq P_{\ell,T}^c,P_{\ell,F}^c$ then by \Cref{obs:emptycliques} and \Cref{obs:nofs}, $d(\{u,v\}, D \setminus \{u,v\}) \geq 3$ and so $uv$ cannot be part of an ST-configuration in $D$.; and if $P= P_{\ell,T}^c$ (resp. $P= P_{\ell,F}^c$) then $d(\{u,v\}, D\setminus \{u,v\}) \geq 3$ unless $t_c^\ell \in D$ (resp. $f_c^\ell \in D$). Thus any edge $uv$ of an ST-configuration in $D$ satisfies one of the following.
\begin{itemize}[leftmargin=16mm]
\item[(aT)] $\{u,v\} = \{P_{\ell,T}^c(3),P_{\ell,T}^c(4)\}$ for some $\ell \in \{x,y,z\}$ and $t_c^\ell \in D$.
\item[(aF)] $\{u,v\} = \{P_{\ell,F}^c(3),P_{\ell,F}^c(4)\}$ for some $\ell \in \{x,y,z\}$ and $f_c^\ell \in D$.
\item[(bT)] $\{u,v\} = \{P_{\ell,T}^(1),t_c^\ell\}$ for some $\ell \in \{x,y,z\}$.
\item[(bF)] $\{u,v\} = \{P_{\ell,F}^(2),f_c^\ell\}$ for some $\ell \in \{x,y,z\}$.
\item[(cT)] $u,v \in V(G_c^T)$.
\item[(cF)] $u,v \in V(G_c^F)$.
\end{itemize}
Note that since for any $\ell \in \{x,y,z\}$ and any $R \in \{T,F\}$, $|D \cap \{P_{\ell,R}^c(1),P_{\ell,R}^c(2),P_{\ell,R}^c(3)\}| \leq 1$ by \Cref{clm:sizeclawstd3}(i), if there is an edge satisfying (aR) and an edge satisfying (bR), then these two edges are vertex-disjoint. We now distinguish cases depending on whether $|D \cap V(G_c^T)| =3$ or $|D \cap V(G_c^F)| = 3$.\\

\textbf{Case 1.a.} \emph{$|D \cap V(G_c^T)| =3.$} Then $D \cap \{f_c^x,f_c^y,f_c^z\} = \varnothing$ by \Cref{obs:nofs}, which implies that no edge of an ST-configuration in $D$ satisfies (aF) or (bF). Now if $D \cap V(G_c^F)$ contains an edge $uv$ then $\{u,v\} \subseteq \{w_c,v_c,u_c\}$ as $w_c$ should be dominated; but then $d(\{u,v\},D \setminus \{u,v\}) \geq 3$ and so, $uv$ cannot be part of an ST-configuration in $D$. Thus no edge of an ST-configuration in $D$ satisfies (cF). Further note that since $|D \cap \{t_c^x,t_c^y,t_c^z\}| \leq 1$ by \Cref{clm:sizeclawstd3}(ii), there is at most one edge of an ST-configuration in $D$ satisfying (aT) or (bT).

Suppose first that $D \cap V(G_c^T)$ is an independent set and let $uv$ be an edge of an ST-configuration in $D$. If $uv$ satisfies (aT) for some $\ell \in \{x,y,z\}$, then $D \cap \{P_{\ell,T}^(1),P_{\ell,T}^c(2)\} = \varnothing$ by \Cref{clm:sizeclawstd3}(i) and so, $uv$ cannot be part of a $P_3$. Similarly, if $uv$ satisfies (bT) for some $\ell \in \{x,y,z\}$, then $D \cap \{P_{\ell,T}^c(2),P_{\ell,T}^c(3)\} = \varnothing$ by \Cref{clm:sizeclawstd3}(i) and so, $uv$ cannot be part of a $P_3$. It follows that $D$ contains no $P_3$, which implies that $D$ must contain an $O_3$ or an $O_7$; in particular, $D$ contains two vertex-disjoint edges. However, any such edge would satisfy (aT) or (bT), a contradiction to the above.

Suppose next that $D \cap V(G_c^T)$ contains an edge. We claim that there exists at most one variable $\ell \in \{x,y,z\}$ such that $D \cap \{P_{\ell,T}^c(1),P_{\ell,T}^c(3)\} \neq \varnothing$. Indeed, suppose for a contradiction that there are two such variables, say $x$ and $y$ without loss of generality. Then $P_{x,F}^c(2),P_{y,F}^c(2) \notin D$ by \Cref{clm:notmin}, which implies that $|D \cap \{q_c^x,q_c^y,f_c^x,f_c^y\}| \geq 2$ as $f_c^x$ and $f_c^y$ should be dominated; but $|D \cap V(G_c^F)| = 2$ and so, $w_c$ is not dominated, a contradiction. Thus assume, without loss of generality, that $D \cap \{P_{x,T}^c(1),P_{x,T}^c(3)\} = D \cap \{P_{y,T}^c(1),P_{y,T}^c(3)\} = \varnothing$ (note that $D \cap (\{t_c^x\} \cup V(P_{x,T}^c))$ and $D \cap (\{t_c^y\} \cup V(P_{y,T}^c))$ can then contain no edge). Since $t_c^x$ and $t_c^y$ should be dominated, necessarily $D \cap \{t_c^x,u_c^x,v_c^x\} \neq \varnothing$ and $D \cap \{t_c^y,u_c^y,v_c^y\} \neq \varnothing$. Now if $D \cap \{t_c^z,u_c^z,v_c^z\} \neq \varnothing$ then since $D \cap V(G_c^T)$ contains an edge by assumption, it must be that $u_c^a,u_c^b \in D$ or $v_c^a,v_c^b \in D$ for two variables $a,b \in \{x,y,z\}$. In the first case, since $D \cap \{v_c^x,v_c^y,v_c^z\} \neq \varnothing$ by \Cref{clm:sizeclawstd3}(ii), necessarily $D \cap \{t_c^x,t_c^y,t_c^z\} = \varnothing$; and in the second case, since $D \cap \{u_c^x,u_c^y,u_c^z\} \neq \varnothing$ by \Cref{clm:sizeclawstd3}(ii), necessarily $D \cap \{t_c^x,t_c^y,t_c^z\} = \varnothing$. Thus, in both cases, we conclude that $D$ contains no edge satisfying (aT) or (bT); but then $D$ contains only one edge, a contradiction. Thus $D \cap \{t_c^z,v_c^z,v_c^z\} = \varnothing$ which implies, in particular, $D \cap (\{t_c^z\} \cup V(P_{z,T}^c))$ contains no edge. Since by assumption $\{t_c^x\} \cup V(P_{x,T}^c)$ and $\{t_c^y\} \cup V(P_{y,T}^c)$ contain no edge as well, it follows that $D \cap V(G_c^T)$ must contain a $P_3$ ($D$ would otherwise contain only one edge). However, since $D \cap \{t_c^x,u_c^x,v_c^x\} \neq \varnothing$ and $D \cap \{t_c^y,u_c^y,v_c^y\} \neq \varnothing$, any such $P_3$ must then be contained in $\{u_c^x,v_c^x,u_c^y,v_c^y\}$ by \Cref{clm:sizeclawstd3}(ii); but $d(\{u_c^x,v_c^x,u_c^y,v_c^y\}, D \setminus V(G_c^T)) \geq 3$, a contradiction.\\

\textbf{Case 1.b.} \emph{$|D \cap V(G_c^F)| = 3$.} Then $D \cap \{t_c^x,t_c^y,t_c^z\} = \varnothing$ by \Cref{obs:nofs}, which implies that no edge of an ST-configuration satisfies (aT) or (bT). Let us next show that no edge of an ST-configuration in $D$ satisfies (cT). Suppose that $D \cap V(G_c^T)$ contains an edge $uv$. Then by \Cref{clm:szieclawstd3}(ii), $\{u,v\} = \{u_c^\ell,v_c^\ell\}$ for some $\ell \in \{x,y,z\}$, say $x$ without loss of generality. If $P_{x,T}^c(1) \notin D$ then $d(\{u,v\}, D \setminus \{u,v\} \geq 3$ and so, $uv$ cannot be part of an ST-configuration in $D$. Thus suppose that $P_{x,T}^c(1) \in D$. Since $t_c^y$ and $t_c^z$ should be dominated and $D \cap \{t_c^y,u_c^y,v_c^y,t_c^z,u_c^z,v_c^z\} = \varnothing$, necessarily $P_{y,T}^c(1),P_{z,T}^c(1) \in D$. Then $D \cap \{P_{x,F}^c(2),P_{y,F}^c(2),P_{z,F}^c(2)\} = \varnothing$ by \Cref{clm:notmin} and so, $|D \cap \{f_c^\ell,q_c^\ell~|~\ell \in \{x,y,z\}\}| \geq 3$ as $f_c^x,f_c^y$ and $f_c^z$ should be dominated. But $|D \cap V(G_c^F)| =3$ and so, $w_c$ is not dominated, a contradiction. Thus no edge of an ST-configuration satisfies (cT). Now observe that $|D \cap \{f_c^x,f_c^y,f_c^z\}| \leq 1$: indeed, if, say, $t_c^x,t_c^y \in D$ then $D \cap \{u_c,q_c^x,q_c^y,q_c^z,f_c^z\} \neq \varnothing$ as $q_c^z$ should be dominated; but then $D \cap \{w_c,v_c\} = \varnothing$ and so, $w_c$ is not dominated, a contradiction. This implies, in particular, that there is at most one edge of an ST-configuration in $D$ satisfying (aF) or (bF).

Now suppose that $D \cap V(G_c^F)$ is an independent set and let $uv$ be an edge of an ST-configuration in $D$. If $uv$ satisfies (aF) for some $\ell \in \{x,y,z\}$, then $D \cap \{P_{\ell,F}^(1),P_{\ell,F}^c(2)\} = \varnothing$ by \Cref{clm:sizeclawstd3}(i) and so, $uv$ cannot be part of a $P_3$. Similarly, if $uv$ satisfies (bF) for some $\ell \in \{x,y,z\}$, then $D \cap \{P_{\ell,T}^c(1),P_{\ell,T}^c(3)\} = \varnothing$ by \Cref{clm:sizeclawstd3}(i) and so, $uv$ cannot be part of a $P_3$. It follows that $D$ contains no $P_3$, which implies that $D$ must contain an $O_3$ or an $O_7$; in particular, $D$ contains two vertex-disjoint edges. However, any such edge would satisfy (aF) or (bF), a contradiction to the above.

Suppose next that $D \cap V(G_c^F)$ contains an edge. We claim that there exists at least one variable $\ell \in \{x,y,z\}$ such that $D \cap \{P_{\ell,F}^c(2),P_{\ell,F}^c(3)\} = \varnothing$. Indeed, if there is no such variable then by \Cref{clm:notmin}, $D \cap \{P_{x,T}^c(1),P_{y,T}^c(1),P_{z,T}^c(1)\} = \varnothing$; but $|D \cap V(G_c^T)| =2$ and so, one of $t_c^x,t_c^y$ and $t_c^z$ is not dominated.Thus assume, without loss of generality, that $D \cap \{P_{x,F}^c(2),P_{x,F}^c(3)\} = \varnothing$ (note that $D \cap (\{f_c^x\} \cup V(P_{x,F}^c))$ can then contain no edge). Since $f_c^x$ should be dominated, necessarily $D \cap \{q_c^x,f_c^x\} \neq \varnothing$. Now if $q_c^x,f_c^x \in D$ then since $|D \cap \{w_c,v_c\}| =1$, $D \cap V(G_c^F)$ contains only one edge and so, $D$ contains no ST-configuration, a contradiction. Suppose next that $f_c^x \in D$. Since $D \cap \{w_c,v_c\} \neq \varnothing$ and $D \cap V(G_c^F)$ contains an edge by assumption, it must then be that $D \cap \{q_c^y,f_c^y,q_c^z,f_c^z\} = \varnothing$; but then $D \cap V(G_c^F)$ contains only one edge and so, $D$ contains no ST-configuration, a contradiction. Suppose finally that $q_c^x \in D$. Since $D \cap \{w_c,v_c\} \neq \varnothing$ and $D \cap V(G_c^F)$ contains an edge, necessarily $D \cap \{f_c^y,f_c^z\} = \varnothing$. Now if $D \cap \{q_c^y,q_c^z\} \neq \varnothing$ in which case $D \cap V(G_c^F)$ contains only one edge and so, $D$ contains no ST-configuration. Thus it must be that $|D \cap \{w_c,v_c,u_c\}| =2$; but then $d(D \cap V(G_c^F),D \setminus V(G_c^F)) \geq 3$ and so, $D$ contains no ST-configuration.\\

\noindent
\textbf{Case 2.} \emph{There exists a variable $x \in X$ such that $|D \cap V(G_x)| = 15$.} Let $c,c',c'' \in C$ be the three clauses in which $x$ appears. If for some variable $\ell \neq x$ appearing in $c,c'$ or $c''$, $D \cap V(G_\ell)$ contains an edge $uv$ then by \Cref{clm:edgeinvar}, $\{u,v\} = \{P(3),P(4)\}$ for some long paw $P$ of $G_\ell$. But then $d(\{u,v\},D \setminus \{u,v\}) \geq 3$ by \Cref{obs:nofs} and so, $uv$ cannot be part of an ST-configuration in $D$. Now for any clause $\ell \in C \setminus \{c,c',c''\}$, $D \cap V(G_\ell)$ contains no edges by Claims~\ref{clm:notruedge} and \ref{clm:notruedge1}; and for any variable $\ell \in X$ not appearing in $c,c'$ or $c''$, $D \cap V(G_v)$ contains no edge by \Cref{clm:novargadedge}. Thus by \Cref{obs:nofs}, any ST-configuration in $D$ is in fact contained in $D \cap (V(G_x) \cup V(G_c) \cup V(G_{c'}) \cup V(G_{c''}))$. 

Now suppose that for some clause $\ell \in \{c,c',c''\}$, $D \cap V(G_\ell)$ contains an edge $uv$. If $u,v \in V(G_\ell^F)$ then, since $w_\ell$ must be dominated, necessarily $\{u,v\} \subseteq \{w_c,v_c,u_c\}$. But then $d(\{u,v\}, D\setminus \{u,v\}) \geq 3$ and so, $uv$ cannot be part of an ST-configuration in $D$. Thus assume that $u,v \in V(G_\ell^T)$ and let $y,z \in X$ be the other two variables appearing in $\ell$. Then by \Cref{clm:sizeclawstd3}(ii), $\{u,v\} = \{u_\ell^a,v_\ell^a\}$ for some variable $a \in \{x,y,z\}$. Now if $P_{a,T}^\ell(1) \notin D$ then $d(\{u,v\}, D\setminus \{u,v\}) \geq 3$ and so, $uv$ cannot be part of an ST-configuration in $D$. Thus suppose that $P_{a,T}^\ell(1) \in D$. Since for every $b \neq a$, $t_\ell^b$ should be dominated, it follows that $P_{b,T}^\ell(1) \in D$ as well. Then by \Cref{clm:notmin}, $P_{y,F}^\ell(2),P_{z,F}^\ell(2) \notin D$ and so $|D \cap \{q_\ell^y,f_\ell^y,q_\ell^z,f_\ell^z\}| \geq 2$ since $f_\ell^y$ and $f_\ell^z$ should be dominated, But $|D \cap V(G_\ell^F)| =2$ and so, $w_c$ is not dominated, a contradiction. Thus no edge of an ST-configuration in $D$ is contained in $D \cap (V(G_c) \cup V(G_{c'}) \cup V(G_{c''})$. By \Cref{obs:nofs}, it follows that every edge of an ST-configuration in $D$ is contained in $G_x$. We next distinguish cases depending on whether $D \cap \{q_x^\ell,p_x^\ell~|~\ell \in \{c,c',c''\}\} \neq \varnothing$ or $|D \cap V(P)| =3$ for some long paw $P$ of $G_x$.\\

\textbf{Case 2.a.} \emph{$D \cap \{q_x^\ell,p_x^\ell~|~\ell \in \{c,c',c''\}\} \neq \varnothing$.} Then by \Cref{clm:sizeclawstd3}(i), $|D \cap \{q_x^\ell,p_x^\ell~|~\ell \in \{c,c',c''\}\}| = 1$. Now suppose that $D \cap \{q_x^c,q_x^{c'},q_x^{c''}\} \neq \varnothing$ (the case where $D \cap \{p_x^c,p_x^{c'},p_x^{c''}\} \neq \varnothing$ is symmetric), say $q_x^c \in D$ without loss of generality. Suppose first that $F_x \in D$. Since $f_c^x \notin D$ by \Cref{obc:nofs} and $P_{x,F}^c(2)$ must be dominated, $D \cap \{P_{x,F}^c(1),P_{x,F}^c(2),P_{x,F}^c(3)\} \neq \varnothing$ and so, by \Cref{clm:sizeclawstd3}(i), $D'=D \setminus \{q_x^c\}$ is an SD set of $G$. Now since $D'$ is in fact minimum, it follows from \Cref{clm:inter} that every vertex in $D'$ has a unique witness; in particular, $D' \cap \{P_{x,F}^c(1),P_{x,F}^{c'}(1),P_{x,F}^{c''}(1)\} = \varnothing$ as $F_x \in D$. On the other hand, since $D$ is not minimal, $D$ must contain an $O_6$ by \Cref{lem:noedge} and \Cref{clm:inter}; in particular, $D \cap V(G_x)$ contains a $P_3$ $Q$. Now since $|D \cap V(P)| = 2$ for every long paw $P$ of $G_x$ and $D \cap \{q_x^{c'},q_x^{c''},p_x^c,p_x^{c'},p_x^{c''}\} = \varnothing$, necessarily $q_x^c \in V(Q)$. However $N(q_x^c) \cap D = \{F_x\}$; and since $D \cap \{v_x,w_x\} \neq \varnothing$ by \Cref{clm:sizeclawstd3}(i), we also have that $N(F_x) \cap D = \{q_c^x\}$, a contradiction. 

Second, suppose that $F_x \notin D$. Then $D \cap V(G_c)$ contains no $P_3$: indeed, by \Cref{clm:sizeclawstd3}(i), $D \cap \{F_x,q_x^{c'},q_x^{c''},p_x^c,p_x^{c'},p_x^{c''}\} = \varnothing$ and for every long paw $P$ of $G_x$, $|D \cap \{P(4),P(5)\}|  \geq 1$ and $|D \cap \{P(1),P(2),P(3)\}| \leq 1$. It follows that $D$ contains an $O_3$ or an $O_7$; in particular, $D \cap V(G_x)$ contains two vertex-disjoint edges. Now suppose that some long paw $P$ of $G_x$ contains an edge $uv$. Then since $D \cap \{P(4),P(5)\} \neq \varnothing$ by \Cref{clm:sizeclawstd3}(i), necessarily $\{u,v\} \subseteq \{P(3),P(4),P(5)\}$. Thus, if $P \neq P_x,P_{x,F}^c$ then $d(\{u,v\}, D \setminus \{u,v\}) \geq 3$ by \Cref{obs:nofs} and so, $uv$ cannot be part of an ST-configuration in $D$. It follows that $D \cap V(P_x)$ must contain an edge, as $D \cap (\{q_x^c\} \cup V(P_{x,F}^c))$ can contain at most one edge of an ST-configuration in $D$; in particular, $D \cap V(P_x) = \{v_x,u_x\}$. Note that since $p_x^c$ should be dominated and $T_x \notin D$, $P_{x,T}^c(2) \in D$ and so, $P_{x,T}^c(1) \notin D$.  Now, similarly, $D \cap (\{q_x^c\} \cup V(P_{x,F}^c))$ must contain an edge $uv$; and either $q_x^c \in \{u,v\}$ in which case $D \cap \{P_{x,F}^c(1),P_{x,F}^c(2),P_{x,F}^(3)\} = \{P_{x,F}^c(1)\}$, or $u,v \in V(P_{x,F}^c)$ in which case $D \cap V(P_x) = \{P_{x,F}^c(3),P_{x,F}^c(4)\}$. In both cases, we conclude that $P_{x,F}^c(2) \notin D$. Since, as shown above, $P_{x,T}^c(1) \notin D$ as well, it follows from \Cref{obs:nofs} that $q_c^x \in D$ and $D \cap \{u_c^x,v_c^x\} \neq \varnothing$ as $f_c^x$ and $t_c^x$ should be dominated, respectively. Now $|D \cap V(G_c^T)| = |D \cap V(G_c^F)| = 2$ and so, by \Cref{clm:sizeclawstd3}(ii), there must exist a variable $\ell \neq x$ contained in $c$ such that $D \cap \{q_c^\ell,u_c^\ell,v_c^\ell\} = \varnothing$. But $f_c^\ell$ and $t_c^\ell$ should be dominated and so, by \Cref{obs:nofs}, $P_{\ell,F}^c(2),P_{\ell,T}^c(1) \in D$, a contradiction to \Cref{clm:TxFx}.\\

\textbf{Case 2.b.} \emph{$|D \cap V(P)|= 3$ for some long paw $P$ of $G_x$.} Suppose first that $P = P_x$. If $D \cap V(P) = \{P(3),P(4),P(5)\}$ then $D \setminus \{P(3)\}$ is an SD set of $G$ and so, by \Cref{lem:noedge} and \Cref{clm:inter}, $D$ must contain an $O_6$; in particular, $D \cap V(G_x)$ contains a $P_3$ $Q$. Since $D \cap \{q_x^\ell,p_x^\ell~|~\ell \in \{c,c',c''\}\} = \varnothing$ and for any long paw $P' \neq P$, $|D \cap V(P')| =2$, necessarily $Q = P(3),P(4),P(5)$; but then, by \Cref{obs:nofs} and because $D \cap \{q_x^\ell,p_x^\ell~|~\ell \in \{c,c',c''\}\} = \varnothing$, $d(P(4), D \setminus \{P(3),P(4),P(5)\}) \geq 3$, a contradiction. Suppose next that $|D \cap \{P(3),P(4),P(5)\}| = 2$. If $P(4),P(5) \in D$ then $D \setminus \{P(5)\}$ is an SD set of $G$ and so, by \Cref{lem:noedge} and \Cref{clm:inter}, $D$ must contain an $O_6$; but $D$ contains no $P_3$ in this case, a contradiction. Now suppose that $D \cap \{P(3),P(4),P(5)\} = \{P(3),P(4)\}$ and assume that $P(1)\in D$ (the case where $P(2) \in D$ is symmetric). Then $D'  = D \setminus \{P(3)\}$ is a minimum SD set of $G$ and so, $D$ must contain an $O_6$ by \Cref{lem:noedge} and \Cref{clm:inter}; in particular, $D \cap V(G_c)$ contains a $P_3$ $Q$. Since $D \cap \{q_x^\ell,p_x^\ell~|~\ell \in \{c,c',c''\}\} = \varnothing$ and for any long paw $P' \neq P$, $|D \cap V(P')| =2$, necessarily $Q = P(1),P(3),P(4)$. However, by \Cref{clm:inter}(ii), $w_{D'}(P(1)) = \{P(4)\}$ and so $d(P(3), D \setminus \{P(1),P(3),P(4)\}) \geq d(P(1),D' \setminus \{P(1),P(4)\}) \geq 3$, a contradiction. Suppose finally that $D \cap \{P(3),P(4),P(5)\} = \{P(3),P(5)\}$ and assume that $P(1) \in D$ (the case where $P(2) \in D$ is symmetric). Then, since $D \cap \{q_x^\ell,p_x^\ell~|~\ell \in \{c,c',c''\}\} = \varnothing$ and for every long paw $P' \neq P$, $|D \cap V(P')| =2$,  $D \cap V(G_x)$ contains no $P_3$. It follows that $D$ contains an $O_3$ or $O_6$; in particular, $D \cap V(G_x)$ must contain an edge $uv$ distinct from $P(1)P(3)$. Now $uv$ can only be contained in some long paw $P' \neq P$; and since $|D \cap \{P'(4),P'(5)\}| \geq 1$ and $|D \cap \{P'(1),P'(2),P'(3)\}| \leq 1$ by \Cref{clm:sizeclawstd3}(i), in fact $\{u,v\} \subseteq \{P'(3),P'(4),P'(5)\}$. But then, by \Cref{obs:nofs} and because $D \cap \{q_x^\ell,p_x^\ell~|~\ell \in \{c,c',c''\}\} = \varnothing$, $d(\{u,v\},D \setminus \{u,v\}) \geq 3$, a contradiction. Thus, it must be that $|D \cap \{P(3),P(4),P(5)\}| = 1$; and since $P(5)$ must be dominated and any vertex dominating $P(5)$ must have a witness, in fact $D \cap \{P(3),P(4),P(5)\} = P(4)$. Since $D \cap \{q_x^\ell,p_x^\ell~|~\ell \in \{c,c',c''\}\} = \varnothing$ and for every long paw $P' \neq P$, $|D \cap V(P')| =2$, it follows that $D$ contains no $P_3$ and so, $D$ contains an $O_3$ or an $O_6$. In particular, $D$ must contain an edge $uv$ distinct from $P(1)P(2)$; but then, we conclude, as in the previous case, that $\{u,v\} \subseteq \{P'(3),P'(4),P'(5)\}$ for some long paw $P' \neq P$ and $d(\{u,v\}, D \setminus \{u,v\}) \geq 3$, a contradiction which concludes the proof.
\end{claimproof}

Now by Claims~\ref{clm:phisatclawstd3} and \ref{clm:clawcstd3}, $\Phi$ is satisfiable if and only if $ct_{\gamma_{t2}}(G) = 3$; and since $G$ is $K_{1,3}$-free, the lemma follows.
\end{proof}

\begin{lemma}
\label{lem:std33p4}
{\sc Contraction Number($\gamma_{t2}$,3)} is $\mathsf{NP}$-hard on $3P_4$-free graphs.
\end{lemma}

\begin{proof}
We reduce from {\sc Positive Exactly 3-Bounded 1-In-3 3-Sat} (see \Cref{sec:prelim} for a precise definition of this problem): given an instance $\Phi$ of this problem, with variable set $X$ and clause set $C$, we construct an instance $G$ of {\sc Contraction Number($\gamma_{t2}$,3)} as follows. For every clause $c \in C$, we introduce the gadget $G_c$ which is constructed as follows. Let $x,y,z \in X$ be the variables appearing in $c$. For every $\ell \in \{x,y,z\}$, denote by $a_\ell$ and $b_\ell$ the two other clauses in $C$ in which $\ell$ appears. To construct the gadget $G_c$ for $c$, we first construct the auxiliary graph $H_c$ depicted in \Cref{fig:hc} (where a rectangle indicates that the corresponding set is a clique). More precisely, $H_c$ consists of two cliques $V_c= \{p_c,q_c,v_c^x,v_c^y,v_c^z\}$ and $K_c = \{u_c^{\ell,a_\ell},u_c^{\ell,b_\ell}~|~ \ell \in \{x,y,z\}\} \cup \{t_c^{x,y},t_c^{x,z},t_c^{y,z}\}$, and the following edges.
\begin{itemize}
\item For every $\ell \in \{x,y,z\}$, $u_c^\ell$ is complete to $\{u_c^{\ell,a_\ell},u_c^{\ell,b_\ell}\}$.
\item For every $\ell \in \{x,y,z\}$, $u_c^\ell$ is adjacent to $t_c^{a,b}$ if and only if $\ell \in \{a,b\}$.
\end{itemize}  

\begin{figure}
\centering
\begin{tikzpicture}
\node[circ,label=below:{\small $t_c^{x,y}$}] (txy) at (2,0) {};
\node[circ,label=below:{\small $t_c^{x,z}$}] (txz) at (4,0) {};
\node[circ,label=below:{\small $t_c^{y,z}$}] (tyz) at (6,0) {};
\node[circ,label=above:{\small $p_c$}] (q) at (5,2.75) {};
\node[circ,label=above:{\small $v_c^x$}] (vx) at (2,2) {};
\node[circ,label=above:{\small $v_c^y$}] (vy) at (4,2) {};
\node[circ,label=above:{\small $v_c^z$}] (vz) at (6,2) {};
\node[circ,label=above:{\small $q_c$}] (p) at (3,2.75) {};
\node[circ,label=below:{\small $u_c^{x,a_x}$}] (ax) at (.5,1) {};
\node[circ,label=below:{\small $u_c^{x,b_x}$}] (bx) at (1.5,1) {};
\node[circ,label=below:{\small $u_c^{y,a_y}$}] (ay) at (3.5,1) {};
\node[circ,label=below:{\small $u_c^{y,b_y}$}] (by) at (4.5,1) {};
\node[circ,label=below:{\small $u_c^{z,a_z}$}] (az) at (6.5,1) {};
\node[circ,label=below:{\small $u_c^{z,b_z}$}] (bz) at (7.5,1) {};

\draw (vx) -- (ax)
(vx) -- (bx)
(vx) -- (txy);
\draw (vx) edge[bend right=25] (txz);

\draw (vy) -- (ay)
(vy) -- (by)
(vy) -- (txy)
(vy) -- (tyz);

\draw (vz) -- (az) 
(vz) -- (bz) 
(vz) -- (tyz);
\draw (vz) edge[bend left=25] (txz);

\draw (0,-.65) rectangle (8,1.25);
\draw (1.7,1.75) rectangle (6.3,3.25);
\end{tikzpicture}
\caption{The graph $H_c$ (a rectangle indicates that the corresponding set of vertices is a clique).}
\label{fig:hc}
\end{figure}
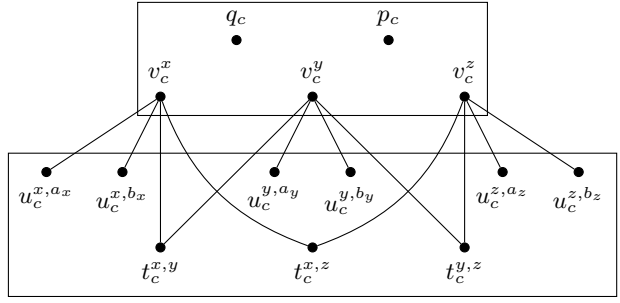

The gadget $G_c$ then consists of two disjoint copies of $H_c$, denoted by $H_{c,1}$ and $H_{c,2}$, with the following additional edges. For $i \in [2]$, let us denote by $\{p_{c,i},q_{c,i},v_{c,i}^x,v_{c,c}^y,v_{c,i}^z\} \cup \{u_{c,i}^{\ell,a_\ell},u_{c,i}^{\ell,b_\ell}~|~ \ell \in \{x,y,z\}\} \cup \{t_{c,i}^{x,y},t_{c,i}^{x,z},t_{c,i}^{y,z}\}$ the vertex set of $H_{c,i}$, and by $K_{c,i}$ (resp. $V_{c,i}$) the copy of $K_c$ (resp. $V_c$) in $H_{c,i}$.
\begin{itemize}
\item[(1)] For every $\ell \in \{x,y,z\}$, $v_{c,1}^\ell$ is adjacent to $t_{c,2}^{a,b}$ if and only if $\ell \notin \{a,b\}$.
\item[(2)] For every $\ell \in \{x,y,z\}$, $v_{c,2}^\ell$ is adjacent to $t_{c,1}^{a,b}$ if and only if $\ell \notin \{a,b\}$.
\item[(3)] $p_{c,1}$ is complete to $V_{c,2} \setminus \{p_{c,2}\}$.
\item[(4)] $p_{c,2}$ is complete to $V_{c,1} \setminus \{p_{c,1}\}$.
\end{itemize}
This completes the construction of $G_c$ (see \Cref{fig:std33p4} for an illustration - the edges of (1) and (2) have been omitted for clarity). 

\begin{figure}
\centering
\begin{tikzpicture}[scale=.9]
\node[circ,label=below:{\small $t_{c,1}^{x,y}$}] (txy) at (2,0) {};
\node[circ,label=below:{\small $t_{c,1}^{x,z}$}] (txz) at (4,0) {};
\node[circ,label=below:{\small $t_{c,1}^{y,z}$}] (tyz) at (6,0) {};
\node[circ,label=above:{\small $q_{c,1}$}] (q) at (.5,2) {};
\node[circ,label=above:{\small $v_{c,1}^x$}] (vx) at (2,2) {};
\node[circ,label=above:{\small $v_{c,1}^y$}] (vy) at (4,2) {};
\node[circ,label=above:{\small $v_{c,1}^z$}] (vz) at (6,2) {};
\node[circ,label=above:{\small $p_{c,1}$}] (p) at (5,4) {};
\node[circ,label=below:{\small $u_{c,1}^{x,a_x}$}] (ax) at (.5,1) {};
\node[circ,label=below:{\small $u_{c,1}^{x,b_x}$}] (bx) at (1.5,1) {};
\node[circ,label=below:{\small $u_{c,1}^{y,a_y}$}] (ay) at (3.5,1) {};
\node[circ,label=below:{\small $u_{c,1}^{y,b_y}$}] (by) at (4.5,1) {};
\node[circ,label=below:{\small $u_{c,1}^{z,a_z}$}] (az) at (6.5,1) {};
\node[circ,label=below:{\small $u_{c,1}^{z,b_z}$}] (bz) at (7.5,1) {};

\draw (vx) -- (ax)
(vx) -- (bx)
(vx) -- (txy);
\draw (vx) edge[bend right=30] (txz);

\draw (vy) -- (ay)
(vy) -- (by)
(vy) -- (txy)
(vy) -- (tyz);

\draw (vz) -- (az) 
(vz) -- (bz) 
(vz) -- (tyz);
\draw (vz) edge[bend left=30] (txz);

\draw (0,-.7) rectangle (8,1.25);
\draw (.15,1.75) rectangle (6.35,2.75);

\node[circ,label=below:{\small $t_{c,2}^{x,y}$}] (txy1) at (2+8.5,0) {};
\node[circ,label=below:{\small $t_{c,2}^{x,z}$}] (txz1) at (4+8.5,0) {};
\node[circ,label=below:{\small $t_{c,2}^{y,z}$}] (tyz1) at (6+8.5,0) {};
\node[circ,label=above:{\small $q_{c,2}$}] (q2) at (16,2) {};
\node[circ,label=above:{\small $v_{c,2}^x$}] (vx1) at (2+8.5,2) {};
\node[circ,label=above:{\small $v_{c,2}^y$}] (vy1) at (4+8.5,2) {};
\node[circ,label=above:{\small $v_{c,2}^z$}] (vz1) at (6+8.5,2) {};
\node[circ,label=above:{\small $p_{c,2}$}] (p1) at (3+8.5,4) {};
\node[circ,label=below:{\small $u_{c,2}^{x,a_x}$}] (ax1) at (.5+8.5,1) {};
\node[circ,label=below:{\small $u_{c,2}^{x,b_x}$}] (bx1) at (1.5+8.5,1) {};
\node[circ,label=below:{\small $u_{c,2}^{y,a_y}$}] (ay1) at (3.5+8.5,1) {};
\node[circ,label=below:{\small $u_{c,2}^{y,b_y}$}] (by1) at (4.5+8.5,1) {};
\node[circ,label=below:{\small $u_{c,2}^{z,a_z}$}] (az1) at (6.5+8.5,1) {};
\node[circ,label=below:{\small $u_{c,2}^{z,b_z}$}] (bz1) at (7.5+8.5,1) {};

\draw (vx1) -- (ax1)
(vx1) -- (bx1)
(vx1) -- (txy1);
\draw (vx1) edge[bend right=30] (txz1);

\draw (vy1) -- (ay1)
(vy1) -- (by1)
(vy1) -- (txy1)
(vy1) -- (tyz1);

\draw (vz1) -- (az1) 
(vz1) -- (bz1) 
(vz1) -- (tyz1);
\draw (vz1) edge[bend left=30] (txz1);

\draw (8.45,-.7) rectangle (16.45,1.25);
\draw (10.15,1.75) rectangle (16.35,2.75);

\draw[very thick] (p) -- (.15,2.75);
\draw[very thick] (p) -- (6.35,2.75);
\draw[very thick] (p) -- (10.15,2.75);
\draw[very thick] (p) -- (16.35,2.75);

\draw[very thick] (p1) -- (.15,2.75);
\draw[very thick] (p1) -- (6.35,2.75);
\draw[very thick] (p1) -- (10.15,2.75);
\draw[very thick] (p1) -- (16.35,2.75);
\end{tikzpicture}
\caption{The graph $G_c$ (a rectangle indicates that the corresponding set of vertices is a clique, and thick edges between a vertex and a rectangle indicates that the vertex is complete to the corresponding set).}
\label{fig:std33p4}
\end{figure}
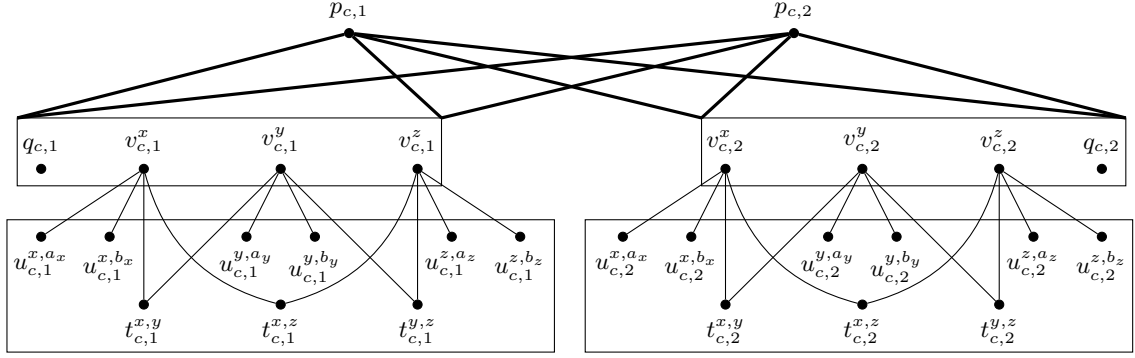

To complete the construction of $G$, we proceed as follows. For every two clauses $a,b \in C$ containing a same variable $u \in X$, we further add edges between $G_a$ and $G_b$ as described below. 
\begin{itemize}
\item $v_{a,1}^u$ is adjacent to $u_{b,2}^{v,a}$, for every variable $v \neq u$ appearing in $b$.
\item $v_{a,2}^u$ is adjacent to $u_{b,1}^{v,a}$, for every variable $v \neq u$ appearing in $b$.
\item $v_{b,1}^u$ is adjacent to $u_{a,2}^{v,b}$, for every variable $v \neq u$ appearing in $a$.
\item $v_{b,2}^u$ is adjacent to $u_{a,1}^{v,b}$, for every variable $v \neq u$ appearing in $a$.
\end{itemize}
Finally, we add edges so that $K_1 = \bigcup_{c \in C} K_{c,1}$ is a clique and $K_2 = \bigcup_{c \in C} K_{c,2}$ is a clique. We let $G$ be the resulting graph. We next show that $\Phi$ is satisfiable if and only if $ct_{\gamma_{t2}}(G) =3$. 

Before turning to the proof, let us briefly explain the idea behind the construction. In a clause gadget $G_c$, the vertices $v_{c,1}^\ell$ and $v_{c,2}^\ell$ should be seen as representing the variable $\ell$. As we will show, for any minimum SD set $D$ of $G$, $D \cap V(G_c)$ contains only one vertex from $V_{c,1} \setminus \{p_{c,1},q_{c,1}\}$ and one vertex from $V_{c,2} \setminus \{p_{c,2},q_{c,2}\}$, both corresponding to the same variable $\ell$. This variable $\ell$ can then be set to true to satisfy $c$. To ensure that the choice in every clause gadget is consistent (that is, if $\ell$ is chosen in $G_c$ then $\ell$ should be chosen in $G_a$, for every clause $a$ containing $\ell$), we make use of the vertices $u_{c,i}^{v,a}$, which should be seen as representing the clause $a$ in $G_c$. By construction, $u_{c,i}^{v,a}$ is adjacent to only $v_{c,i}^v$ (representing $v$) and $v_{a,(i+1) \mod 2}^u$ (representing $u$), for $u \neq v$. Thus, if $v_{c,i}^v$ is not in the SD set then $v_{a,i}^v$ is not in the SD set either, as $u_{c,i}^{v,a}$ would otherwise not be dominated.  

\begin{claim}
\label{clm:size3p4}
For any SD $D$ of $G$ and any clause $c \in C$, $|D \cap (V_{c,1} \cup V_{c,2})| \geq 2$.
\end{claim}

\begin{claimproof}
Let $D$ be an SD set of $G$ and let $c \in C$ be a clause. Since $q_{c,1}$ should be dominated, $D \cap (V_{c,1} \cup \{p_{c,2}\} \neq \varnothing$; and similarly, since $q_{c,2}$ should be dominated, $D \cap (V_{c,2} \cup \{p_{c,1}\}) \neq \varnothing$. Suppose first that $p_{c,1} \in D$. Then $D \cap (V_{c,2} \cup V_{c,1} \setminus \{p_{c,1}\}) \neq \varnothing$, as $p_{c,2}$ should be dominated, and so $|D \cap (V_{c,1} \cup V_{c,2})| \geq 2$. We conclude symmetrically if $p_{c,2} \in D$. Otherwise, $D \cap \{p_{c,1},p_{c,2}\} = \varnothing$ and so, $D \cap (V_{c,i} \setminus \{p_{c,i}\}) \neq \varnothing$ for every $i \in [2]$.  
\end{claimproof}

\begin{claim}
\label{clm:size3p41}
If $\gamma_{t2}(G) = 2|C|$ then for every minimum SD set $D$ of $G$ and every clause $c \in C$, $D \cap V(G_c) = \{v_{c,1}^\ell, v_{c,2}^\ell\}$ for some variable $\ell$ appearing in $c$. 
\end{claim}

\begin{claimproof}
Assume that $\gamma_{t2}(G)=2|C|$ and let $D$ be a minimum SD set of $G$. Consider a clause $c \in C$ containing variables $x,y$ and $z$. We contend that $D \cap \{p_{c,1},q_{c,1},p_{c,2},q_{c,2}\} = \varnothing$. Indeed, observe first that by \Cref{clm:size3p4}, $|D \cap \{p_{c,1},q_{c,1},p_{c,2},q_{c,2}\}| \leq 2$. Now if $|D \cap \{p_{c,1},q_{c,1},p_{c,2},q_{c,2}\}| =2$ then the vertices in $\{t_{c,i}^{x,y},t_{c,i}^{x,z},t_{c,i}^{y,z}~|~i \in [2]\}$ are not dominated, as $D \cap (K_1 \cup K_2) = \varnothing$ by \Cref{clm:size3p4}. Similarly, if $|D \cap \{p_{c,1},q_{c,1},p_{c,2},q_{c,2}\}|=1$ then by \Cref{clm:size3p4}, $|D \cap \{v_{c,i}^x,v_{c,i}^y,v_{c,i}^z~|~i \in [2]\}| =1$, say $v_{c,i}^\ell \in D$, and so the vertex $t_{c,i}^{a,b}$ with $\ell \notin \{a,b\}$, is not dominated. Thus, $D \cap \{p_{c,1},q_{c,1},p_{c,2},q_{c,2}\} = \varnothing$ which by \Cref{clm:size3p4}, implies that $|D \cap \{v_{c,i}^x,v_{c,i}^y,v_{c,i}^z~|~i \in [2]\}| = 2$; and since $q_{c,1}$ and $q_{c,2}$ should be dominated, in fact $|D \cap \{v_{c,1}^x,v_{c,1}^y,v_{c,1}^z\}| = |D \cap \{v_{c,2}^x,v_{c,2}^y,v_{c,2}^z\}| = 1$. Now if $v_{c,1}^u,v_{c,2}^v \in D$ for two distinct variables $u,v \in \{x,y,z\}$, then the vertex $t_{c,1}^{v,w}$ where $w \in \{x,y,z\} \setminus \{u,v\}$, is not dominated. Thus, $D \cap V(G_c) = \{v_{c,1}^\ell,v_{c,2}^\ell\}$ for some variable $\ell \in \{x,y,z\}$.
\end{claimproof}

\begin{claim}
\label{clm:phisat3p4}
$\Phi$ is satisfiable if and only if $\gamma_{t2}(G) = 2|C|$.
\end{claim}

\begin{claimproof}
Assume first that $\Phi$ is satisfiable and consider a truth assignment satisfying $\Phi$. We construct an SD set $D$ of $G$ as follows: for every clause $c \in C$, exactly one variable in $c$ is true, say $\ell$, in which case we add $v_{c,1}^\ell$ and $v_{c,2}^\ell$ to $D$. Let us show that the constructed set $D$ is indeed an SD set of $G$. Consider a clause $c \in C$ containing variables $x,y$ and $z$, and assume without loss of generality that $x$ is true. For every $\ell \in \{x,y,z\}$, let $a_\ell$ and $b_\ell$ be the two other clauses in which $\ell$ appears. Note that since by construction, $D \cap V(G_c) = \{v_{c,1}^x,v_{c,2}^x\}$ and $d(v_{c,1}^x,v_{c,2}^x) = 2$, $v_{c,1}^x$ and $v_{c,2}^x$ witness each other. Now it is clear that every vertex in $V_{c,1} \cup V_{c,2} \cup \{u_{c,i}^{x,a_x},u_{c,i}^{x,b_x}~|~ i \in [2]\} \cup \{t_{c,i}^{x,y},t_{c,i}^{x,z}~|~ i \in [2]\}$ is dominated. Furthermore, by construction, $v_{c,1}^x$ is adjacent to $t_{c,2}^{y,z}$, and $v_{c,2}^x$ is adjacent to $t_{c,1}^{y,z}$, and thus, these two vertices are dominated as well. There remains to show that for every $i \in [2]$ and every $\ell \neq x$, the vertices $u_{c,i}^{\ell,a_\ell}$ and $u_{c,i}^{\ell,b_\ell}$ are dominated. But this readily holds true: for every $\ell \in \{y,z\}$ and every $c_\ell \in \{a_\ell,b_\ell\}$, since $\ell$ is false, $D \cap V(G_{c_\ell}) = \{v_{c_\ell,1}^v,v_{c_\ell,2}^v\}$ for some variable $v \neq \ell$ appearing in $c_\ell$; and by construction, $v_{c_\ell,1}^v$ is adjacent to $u_{c,2}^{\ell,c_\ell}$ and $v_{c_\ell,2}^v$ is adjacent to $u_{c,1}^{\ell,c_\ell}$. Thus, $D$ is an SD set of $G$ and since $|D| = 2|C|$, we conclude by \Cref{clm:size3p4} that $D$ is minimum.

Conversely, assume that $\gamma_{t2}(G) = 2|C|$ and let $D$ be a minimum SD set of $G$. Then by \Cref{clm:size3p41}, $D \cup (K_1 \cup K_2) = \varnothing$ and for every clause $c \in C$, $D \cap V(G_c) = \{v_{c,1}^\ell,v_{c,2}^\ell\}$ for some variable $\ell$ contained in $c$. We claim that if $a,b \in C$ are two clauses containing a same variable $\ell$ and $D \cap V(G_a) = \{v_{a,1}^\ell,v_{a,2}^\ell\}$, then $D \cap V(G_b) = \{v_{b,1}^\ell,v_{b,2}^\ell\}$. Indeed, suppose for a contradiction that $D \cap V(G_b) = \{v_{b,1}^v,v_{b,2}^v\}$ for some variable $v \neq \ell$, and let $p,q \in X$ be the other two variables appearing in $a$. Then $u_{b,1}^{\ell,a}$ and $u_{b,2}^{\ell,a}$ are not dominated: indeed, $N(u_{b,1}^{\ell,a}) \setminus (K_1 \cup K_2) = \{v_{b,1}^\ell,v_{a,2}^p,v_{a,2}^q\}$ and $N(u_{b,2}^{\ell,a}) \setminus (K_1 \cup K_2)= \{v_{b,2}^\ell,v_{a,1}^p,v_{a,1}^q\}$, and so $D \cap N(u_{b,1}^{\ell,a}) = D \cap N(u_{b,2}^{\ell,a}) = \varnothing$. Thus, $V(G_b) = \{v_{b,1}^\ell,v_{b,2}^\ell\}$ as claimed. It follows that the truth assignment obtained by setting a variable $\ell$ to true if $v_{c,1}^\ell \in D$ for some clause $c$ containing $\ell$, and to false otherwise, satisfies $\Phi$.
\end{claimproof}

\begin{claim}
\label{clm:gt33p4}
$ct_{\gamma_{t2}}(G) = 3$ if and only if $\gamma_{t2}(G) = 2|C|$.
\end{claim}

\begin{claimproof}
Assume first that $ct_{\gamma_{t2}}(G) = 3$ and let $D$ be a minimum SD set of $G$. If there exists a clause $c \in C$ such that $|D \cap (V_{c,1} \cup V_{c,2})| \geq 3$ then $D \cap (V_{c,1} \cup V_{c,2})$ contains a friendly triple, a contradiction to \Cref{thm:friendlytriple}. Thus, $|D \cap (V_{c,1} \cup V_{c,2})| \leq 2$ for every clause $c \in C$, and we conclude by \Cref{clm:size3p4} that, in fact, equality holds. Now suppose that $D \cap K \neq \varnothing$, say $u \in D \cap (K_{c,1} \cup K_{c,2}$ for some clause $c \in C$. If $u$ has a neighbour $v \in D \cap (V_{c,1} \cup V_{c,2})$, then $u,v,w$ where $w \in D \cap (V_{c,1} \cup V_{c,2} \setminus \{v\})$, is a friendly triple, a contradiction to \Cref{them:friendlytriple}. Thus, $N(u) \cap (D \cup (V_{c,1} \cup V_{c,2})) = \varnothing$. Now let $v \in V_{c,1} \in V_{c,2}$ be a neighbour of $u$, and let $w_1,w_2 \in D \cap (V_{c,1} \cup V_{c,2})$. Then the vertices in $\{v,w_1,w_2\}$ are pairwise at distance at most two, and at least two of them must be adjacent. However, if $w_1w_2 \in E(G)$ then $D \cup \{v\}$ contains the $O_7$ $u,v,w_1,w_2$; and if $vw_i \in E(G)$ for some $i \in [2]$, then $D \cup \{v\}$ contains the $O_4$ $u,v,w_1,w_2$, a contradiction in both cases to \Cref{thm:friendlytriple}. Thus, $D \cap K = \varnothing$ and so, $\gamma_{t2}(G) = |D| = 2|C|$.\\

Conversely, assume that $\gamma_{t2}(G) = 2|C|$ and let $D$ be a minimum SD set of $G$. Then by \Cref{clm:size3p41}, $D \cap (K_1 \cup K_2) = \varnothing$ and for every clause $c \in C$, $D \cap V(G_c) = \{v_{c,1}^\ell,v_{c,2}^\ell\}$ for some variable $\ell$ contained in $c$. Since for every clause $c \in C$ and every variable $v \in X$ contained in $c$, $d(v_{c,1}^v,v_{c,2}^v) = 2$, it follows that $D$ is an independent set and thus, $D$ contains no friendly triple. Furthermore, as shown in the proof of \Cref{clm:phisat3p4}, if two clauses $a,b \in C$ contain a same variable $\ell$ and $D \cap V(G_a) = \{v_{a,1}^\ell,v_{a,2}^\ell\}$, then $D \cap V(G_b) = \{v_{b,1}^\ell,v_{b,2}^\ell\}$. It follows that for any two clauses $a,b \in C$ containing a same variable, $d(D \cap V(G_a),D\cap V(G_b)) \geq 3$; and since for any two clauses $a,b \in C$ with no common variable, $d(V_{a,1} \cup V_{a,2}, V_{b,1} \cup V_{b,2}) \geq 3$, every vertex in $D$ has a unique witness. Thus, the following hold.

\begin{observation}
\label{obs:min3p4}
If $\gamma_{t2}(G) = 2|C|$ then for every minimum SS set $D$ of $G$, the following hold.
\begin{itemize}
\item[(i)] $D$ is an independent set.
\item[(ii)] For every $x \in D$, $|w_D(x)| =1$.
\end{itemize}
\end{observation}

Now suppose for a contradiction that $G$ has an SD set $D$ of size $\gamma_{t2}(G)+1$ containing an ST-configuration (see \Cref{thm:friendlytriple}). Then by \Cref{clm:size3p4}, either there exists a clause $c \in C$ such that $|D \cap (V_{c,1} \cup V_{c,2})| =3$, or $|D \cap (K_1 \cup K_2)| = 1$. We next distinguish these two cases.\\

\noindent
\textbf{Case 1.} \emph{There exists a clause $c \in C$ such that $|D \cap (V_{c,1} \cup V_{c,2})| = 3$.} Let us first show that for every clause $a \in C \setminus \{c\}$, $D \cap (V_{a,1} \cup V_{a,2})$ contains no edge. Suppose for a contradiction that there exists a clause $a \in C \setminus \{a\}$ such that $D \cap (V_{a,1} \cup V_{a,2})$ contains an edge, and let $x,y,z \in X$ be the three variables appearing in $a$. Since by \Cref{clm:size3p4}, $|D \cap (V_{a,1} \cup V_{a,2})| =2$, necessarily $D \cap \{q_{a,i},v_{a,i}^x,v_{a,i}^y,v_{a,i}^z\} = \varnothing$ for some $i \in [2]$. But $D \cap K_i = \varnothing$ and $|D \cap \{v_{a,j}^x,v_{a,j}^y,v_{a,j}^z\}| \leq 2$ for $j \neq i$ and so, one at least of $t_{a,i}^{x,y},t_{a,i}^{x,z}$ and $t_{a,i}^{y,z}$ is not dominated, a contradiction.

Now let $x,y,z \in X$ be the three variables appearing in $c$. We contend that for every $i \in [2]$, $D \cap \{v_{c,i}^x,v_{c,i}^y,v_{c,i}^z\} \neq \varnothing$. Indeed, suppose for a contradiction that $D \cap \{v_{c,i}^x,v_{c,i}^y,v_{c,i}^z\} = \varnothing$ for some $i \in [2]$, say $i=1$ without loss of generality. Then since every vertex in $\{t_{c,1}^{x,y},t_{c,1}^{x,z},t_{c,1}^{y,z}\}$ must be dominated and $D \cap K_1 = \varnothing$, necessarily $v_{c,2}^\ell \in D$ for every $\ell \in \{x,y,z\}$. But then, $D \cap (V_{c,1} \cup V_{c,2}) = \{v_{c,2}^x,v_{c,2}^y,v_{c,2}^z\}$ and so, $q_{c,1}$ is not dominated, a contradiction. It follows that $D \cap \{p_{c,1},p_{c,2}\} \neq \varnothing$: indeed, if $p_{c,1},p_{c,2} \notin D$ then $D \cap (V_{c,1} \cup V_{c,2})$ contains only one edge; and since $D \setminus V_{c,1} \cup V_{c,2})$ contains no edge, as shown above, $D$ contains no ST-configuration, a contradiction to our assumption. Now assume without loss of generality that $p_{c,1} \in D$, and for $i \in [2]$, let $w_i \in D \cap \{v_{c,i}^x,v_{c,i}^y,v_{c,i}^z\}$. Then since $D \setminus \{p_{c,1}\}$ is an SD set of $G$, it follows from  \Cref{lem:noedge} and \Cref{obs:min3p4} that $D$ contains an $O_6$; but $w_1p_{c,1}w_2$ is the only $P_3$ contained in $D$ and $d(p_{c,1}, D \setminus \{w_1,w_2,p_{c,1}\}) \geq 3$, a contradiction.\\

\noindent
\textbf{Case 2.} \emph{$|D \cap (K_1 \cup K_2)| =1$.} Assume without loss of generality that $D \cap K_1 \neq \varnothing$. Let us first show that for every clause $c \in C$, $|D \cap (V_{c,1} \setminus \{p_{c,1}\})| = D \cap (V_{c,2} \setminus \{p_{c,2}\}| = 1$. Suppose for a contradiction that there exists a clause $c \in C$ such that $D \cap (V_{c,i} \setminus \{p_{c,i}\}) = \varnothing$ for some $i \in [2]$, and let $x,y,z \in X$ be the three variables appearing in $c$. Then $i \neq 2$: indeed, since $D \cap K_2 = \varnothing$ and $|D \cap \{v_{c,1}^,v_{c,1}^y,v_{c,1}^z\}| \leq 2$, if $i=2$ then one of $t_{c,2}^{x,y},t_{c,2}^{x,z},t_{c,2}^{y,z}$ is not dominated. Thus, it must be that $i =1$; but then, $D \cap \{p_{c,1},p_{c,2}\} \neq \varnothing$ as $q_{c,1}$ should be dominated, which implies that $|D \cap \{q_{c,2},v_{c,2}^x,v_{c,2}^y,v_{c,2}^z\}| \leq 1$, and so, one of $t_{c,2}^{x,y},t_{c,2}^{x,z}$ and $t_{c,2}^{y,z}$ is not dominated, a contradiction. Thus, $|D \cap (V_{c,1} \setminus \{p_{c,1}\})| = D \cap (V_{c,2} \setminus \{p_{c,2}\}| = 1$ for every clause $c \in C$; in particular, $D \cap (V_{c,1} \cup V_{c,2})$ contains no edge.

Now let $u \in D \cap K_1$, say $u \in K_{c,1}$ for some clause $c \in C$. Then $|N(u) \cap D| > 1$: if not then $u$ is the endvertex of at most one edge in $D$; and since $D \setminus \{u\}$ contains no edge, as shown above, $D$ then contains no ST-configuration. Suppose first that $u \in \{t_{c,1}^{x,y},t_{c,1}^{x,z},t_{c,1}^{y,z}\}$, say $u = t_{c,1}^{x,y}$ without loss of generality. Then since $|D \cap \{v_{c,1}^,v_{c,1}^y,v_{c,1}^z\}| \leq 1$ as shown above, it must be that $D \cap \{v_{c,1}^x,v_{c,1}^y\} \neq \varnothing$ and $v_{c,2}^z \in D$ ($u$ would otherwise has at most one neighbour in $D$); but then, $t_{c,2}^{x,y}$ is not dominated, a contradiction. Second, suppose that $u \in K_{c,1} \setminus \{t_{c,1}^{x,y},t_{c,1}^{x,z},t_{c,1}^{y,z}\}$, say $u = u_{c,1}^{x,a}$ where $a \neq c$ is a clause containing $x$. Then since $N(u_{c,1}^{x,a} \subseteq V_{c,1} \cup V_{a,2}$ and $|D \cap V_{c,1}| = |D \cap V_{a,2}| = 1$ as shown above, it must be that $D \cap V_{c,1} = \{v_{c,1}^x\}$ and $D \cap V_{a,2} = \{v_{a,2}^\ell\}$ for some variable $\ell \neq x$ contained in $a$; but then, $u_{a,2}^{x,c}$ is not dominated, a contradiction which concludes the proof.
\end{claimproof}

Now by Claims~\ref{clm:phisat3p4} and \ref{clm:gt33p4}, $\Phi$ is satisfiable if and only if $ct_{\gamma_{t2}}(G) =3$. There remains to show that $G$ is $3P_4$-free. To see this, observe that for any clause $c \in C$, $G[V_{c,1} \cup V_{c,2}]$ is $P_4$-free. Thus, if $G$ contains a $P_4$ $P$, then $V(P) \cap (K_1 \cup K_2) \neq \varnothing$; but $K_1$ are $K_2$ are both cliques and so, $G$ contains no induced $3P_4$.
\end{proof}

%------------------------------------------------------------------------------------------------------------------------------------------------------------------------------------

\subsection{Proof of \Cref{thm:dicstd2}}
\label{sec:dicstd2}

Let $H$ be a graph. If $H$ contains a cycle then {\sc Contraction Number($\gamma_{t2}$,2)} is $\mathsf{NP}$-hard on $H$-free graphs by \Cref{lem:cyclesstd2}. Assume henceforth that $H$ is a forest. If $H$ contains a vertex of degree at least three then {\sc Contraction Number($\gamma_{t2}$,2)} is $\mathsf{NP}$-hard on $H$-free graphs by \Cref{lem:clawstd2}. Suppose therefore that $H$ is a linear forest. If $H$ has a connected component on at least six vertices then {\sc Contraction Number($\gamma_{t2}$,2)} is $\mathsf{(co)NP}$-hard on $H$-free graphs by \Cref{lem:linforeststd2}. Thus we may assume that every connected component of $H$ has size at most five. Now suppose that $H$ has a connected component on at least four vertices, Then if every other connected component of $H$ has size one, {\sc Contraction Number($\gamma_{t2}$,2)} is polynomial-time solvable on $H$-free graphs by \Cref{lem:easystd2}; otherwise {\sc Contraction Number($\gamma_{t2}$,2)} is $\mathsf{(co)NP}$-hard on $H$-free graphs by \Cref{lem:linforeststd2}. Assume finally that every connected component of $H$ has at most three vertices. If $H$ has at least two connected component of size three then {\sc Contraction Number($\gamma_{t2}$,2)} is $\mathsf{(co)NP}$-hard on $H$-free graphs by \Cref{lem:linforeststd2}. Otherwise $H$ has at most one connected component of size three in which case {\sc Contraction Number($\gamma_{t2}$,2)} is polynomial-time solvable on $H$-free graphs by \Cref{lem:easystd2} which concludes the proof.

%------------------------------------------------------------------------------------------------------------------------------------------------------------------------------------

\bibliographystyle{plainnat}
\bibliography{bib}

\end{document}